\documentclass[onefignum,onetabnum]{siamonline171218}

\pdfoutput=1

\usepackage{lipsum}
\usepackage{amsfonts}
\usepackage{graphicx}
\usepackage{epstopdf}
\usepackage{algorithmic}
\usepackage[outline]{contour}
\ifpdf
  \DeclareGraphicsExtensions{.eps,.pdf,.png,.jpg}
\else
  \DeclareGraphicsExtensions{.eps}
\fi

\usepackage{enumitem}
\setlist[enumerate]{leftmargin=.5in}
\setlist[itemize]{leftmargin=.5in}

\usepackage{mathtools}

\newsiamremark{remark}{Remark}
\newsiamremark{hypothesis}{Hypothesis}
\crefname{hypothesis}{Hypothesis}{Hypotheses}
\crefname{subsection}{Section}{Sections}
\newsiamthm{claim}{Claim}

\headers{Hypergraph Spectral Clustering}{P. S. Chodrow, N. Eikmeier, and J. L. Haddock}

\title{Nonbacktracking Spectral Clustering of Nonuniform Hypergraphs}

\author{
    Philip S. Chodrow\thanks{Department of Computer Science, Middlebury College, Middlebury, VT (\href{mailto:pchodrow@middlebury.edu}{pchodrow@middlebury.edu})}\and 
    Nicole Eikmeier\thanks{Department of Computer Science, Grinnell College, Grinnell, IA}\and
    Jamie Haddock\thanks{Department of Mathematics, Harvey Mudd College, Claremont, CA\funding{JH is partially supported by NSF DMS \#2211318}}
    }
    
\usepackage{amsopn}
\DeclareMathOperator{\diag}{diag}

\makeatletter
\newcommand*{\addFileDependency}[1]{
  \typeout{(#1)}
  \@addtofilelist{#1}
  \IfFileExists{#1}{}{\typeout{No file #1.}}
}
\makeatother

\newcommand*{\myexternaldocument}[1]{%
    \externaldocument{#1}%
    \addFileDependency{#1.tex}%
    \addFileDependency{#1.aux}%
}

\usepackage[utf8]{inputenc}
\usepackage[english]{babel}
\usepackage{amsmath}
\usepackage{amssymb}
\usepackage{MnSymbol}
\usepackage{hyperref}
\usepackage{xcolor}

\usepackage{cleveref}
\usepackage{graphicx}
\usepackage{bbm}
\usepackage{todonotes}
\usepackage{algorithmic}
\usepackage{etoolbox}
\usepackage{mathtools}

\newcommand{\mA}{\mathbf{A}}
\newcommand{\bA}{\mathbf{A}}
\newcommand{\mD}{\mathbf{D}}
\newcommand{\bD}{\mathbf{D}}
\newcommand{\mB}{\mathbf{B}}
\newcommand{\bB}{\mathbf{B}}
\newcommand{\vc}{\mathbf{c}}
\newcommand{\mI}{\mathbf{I}}
\newcommand{\mM}{\mathbf{M}}
\newcommand{\bM}{\mathbf{M}}

\newcommand{\bL}{\mathbf{L}}

\newcommand{\bu}{\mathbf{u}}

\newcommand{\mE}{\mathbf{E}}
\newcommand{\mW}{\mathbf{W}}
\newcommand{\bW}{\mathbf{W}}
\newcommand{\mS}{\mathbf{S}}
\newcommand{\bS}{\mathbf{S}}

\newcommand{\mT}{\mathbf{T}}
\newcommand{\mP}{\mathbf{P}}
\newcommand{\bT}{\mathbf{T}}
\newcommand{\bJ}{\mathbf{J}}

\newcommand{\mG}{\mathbf{G}}

\newcommand{\mF}{\mathbf{F}}
\newcommand{\mC}{\mathbf{C}}

\newcommand{\mJ}{\mathbf{J}}
\newcommand{\mX}{\mathbf{X}}
\newcommand{\mY}{\mathbf{Y}}
\newcommand{\mZ}{\mathbf{Z}}

\newcommand{\mL}{\mathbf{L}}
\newcommand{\mK}{\mathbf{K}}
\newcommand{\mH}{\mathbf{H}}

\newcommand{\vmu}{\boldsymbol{\mu}}

\newcommand{\cZ}{\mathcal{Z}}
\newcommand{\cE}{\mathcal{E}}
\newcommand{\cN}{\mathcal{N}}
\newcommand{\cG}{\mathcal{G}}
\newcommand{\cP}{\mathcal{P}}
\newcommand{\cH}{\mathcal{H}}

\newcommand{\cM}{\mathcal{M}}

\newcommand{\cR}{\mathcal{R}}

\newcommand{\vz}{\mathbf{z}}
\newcommand{\vx}{\mathbf{x}}
\newcommand{\vy}{\mathbf{y}}
\newcommand{\vv}{\mathbf{v}}
\newcommand{\vu}{\mathbf{u}}
\newcommand{\ve}{\mathbf{e}}
\newcommand{\vzero}{\mathbf{0}} 
\newcommand{\abs}[1]{\lvert #1 \rvert}
\newcommand{\norm}[1]{\lVert #1 \rVert}

\newcommand{\vsigma}{\boldsymbol{\sigma}}

\newcommand{\vtheta}{\boldsymbol{\theta}}
\newcommand{\vq}{\mathbf{q}}
\newcommand{\bracket}[1]{\langle #1 \rangle}

\newcommand{\q}[1]{q^{(#1)}}
\newcommand{\qhat}[1]{\hat{q}^{(#1)}}
\newcommand{\ck}[2]{%
    \ifstrempty{#2}{c_{#1}}{c_{#1}^{(#2)}}%
    }
\newcommand{\gk}[2]{%
    \ifstrempty{#2}{g_{#1}}{g_{#1}^{(#2)}}%
    }

\newcommand{\mk}[2]{%
    \ifstrempty{#2}{m_{#1}}{m_{#1}^{(#2)}}%
    }

\newcommand{\ckin}{
    c_k^{\mathrm{in}}
}

\newcommand{\ckout}{
    c_k^{\mathrm{out}}
}

\newcommand{\gkin}{
    g_k^{\mathrm{in}}
}

\newcommand{\gkout}{
    g_k^{\mathrm{out}}
}

\newcommand{\E}{\mathbb{E}}

\newcommand{\R}{\mathbb{R}}

\newcommand{\vepsilon}{\boldsymbol{\epsilon}}
\newcommand{\mPi}{\boldsymbol{\Pi}}

\newcommand{\HG}{\cH}
\newcommand{\G}{\cG}

\newcommand{\nodes}{\cN}
\newcommand{\edges}{\cE}
\newcommand{\tuples}{\cR}
\newcommand{\alphabet}{\cZ}
\newcommand{\labelvec}{\vz}

\newcommand{\mess}[3]{\mu_{#1#2}^{(#3)}}
\newcommand{\barmess}[3]{\bar{\mu}_{#1#2}^{(#3)}}
\newcommand{\varmess}[3]{\nu_{#2#1}^{(#3)}}
\newcommand{\pmess}[3]{\epsilon_{#1#2}^{(#3)}}
\newcommand{\vmess}{\vmu}
\newcommand{\vpmess}{\vepsilon}
\newcommand{\messcoords}[4]{#1_{#2#3}^{(#4)}}

\newcommand{\maps}[2]{%
    \ifstrempty{#2}{L(#1)}{L(#1,#2)}%
}
\newcommand{\vectorspace}[1]{V(#1)}

\newcommand{\kmax}{\bar{k}}
\newcommand{\Kset}{K}

\newcommand{\pointed}[1]{\vec{#1}}
\newcommand{\indicator}[1]{\mathbbm{1}[#1]}

\newsiamthm{conj}{Conjecture}
\newsiamthm{cor}{Corollary}

\newcounter{jh}
\definecolor{pc_comment}{RGB}{120,205,205}
\newcounter{pc}

\definecolor{nee_comment}{RGB}{144,238,144}
\newcounter{nee}

\newcommand{\datahighschool}{%
    \texttt{contact-high-school}%
    }
\newcommand{\dataprimaryschool}{\texttt{contact-primary-school}}
\newcommand{\datasenate}{\texttt{senate-bills}}

\newcommand{\oursimplealg}{NBHSC}
\newcommand{\oursimplealgorithm}{nonbacktracking hypergraph spectral clustering}
\newcommand{\OurSimpleAlgorithm}{Nonbacktracking Hypergraph Spectral Clustering}
\newcommand{\ouralg}{BPHSC}
\newcommand{\ouralgorithm}{belief-propagation hypergraph spectral clustering}
\newcommand{\OurAlgorithm}{Belief-Propagation Hypergraph Spectral Clustering}
\newcommand{\graphsimplealgorithm}{nonbacktracking spectral clustering}
\newcommand{\graphsimplealg}{NBSC}
\newcommand{\graphalg}{BPPGSC}
\newcommand{\graphalgorithm}{belief-propagation projected graph spectral clustering}

\newcommand{\eqdef}{\coloneqq}

\DeclarePairedDelimiter{\paren}{(}{)}
\DeclarePairedDelimiter{\sqbracket}{[}{]}
\DeclarePairedDelimiter{\curlybrace}{\{}{\}}

\Crefname{conj}{Conjecture}{Conjectures}

\newcommand{\revision}[1]{#1}
\newcommand{\rev}[1]{\revision{#1}} 

\usepackage{tabularx}
\usepackage{booktabs}
\usepackage[numbers]{natbib}

\myexternaldocument{supplement} 

\begin{document}
   
\maketitle 

\begin{abstract}
    Spectral methods offer a tractable, global framework for clustering in graphs via eigenvector computations on graph matrices. 
    Hypergraph data, in which entities interact on edges of arbitrary size, poses challenges for matrix representations and therefore for spectral clustering. 
    We study spectral clustering for nonuniform hypergraphs based on the hypergraph nonbacktracking operator. 
    After reviewing the definition of this operator and its basic properties, we prove a theorem of Ihara-Bass type \revision{which allows eigenpair computations to take place on a smaller matrix, often enabling faster computation.} 
    We then propose an alternating algorithm for inference in a hypergraph stochastic blockmodel via linearized belief-propagation \revision{which involves a spectral clustering step again using nonbacktracking operators. 
    We provide} proofs \revision{related to this algorithm} that both formalize and extend several previous results. 
    We pose several conjectures about the limits of spectral methods and detectability in hypergraph stochastic blockmodels \revision{in general}\revision{, supporting these with in-expectation analysis of the eigeinpairs of our operators}. 
    We perform experiments in real and synthetic data that \revision{demonstrate} the benefits of hypergraph methods over graph-based ones when interactions of different sizes carry different information about cluster structure. 
\end{abstract}

  
  %


\section{Introduction} \label{sec:intro}
Graphs provide a classical representation for systems with pairwise relationships: components are modeled by nodes, and relationships are modeled by edges. 
In many systems however, relationships simultaneously involve more than two components. 
Examples include three scholars writing a paper together, or three genes interacting to influence phenotypic traits. 
While certain graph techniques can be used in such cases, there is often value in keeping multiway interactions intact~\cite{AktasCritical,bojanek2020cyclic}. 
In such cases, \emph{polyadic} or \emph{higher-order} representations are useful, and analysis based on such representations can lead to qualitatively and qualitatively different conclusions~\cite{chodrowConfigurationModelsRandom2020}. 
\revision{Polyadic representations include hypergraphs,  simplicial complexes, and various generalizations. 
There has been a wealth of } recent work using \revision{such structures} for modeling complex systems~\cite{bickWhatAreHigherorder2021,battistonNetworksPairwiseInteractions2020,battiston2021physics}.

In this paper, we focus on the community detection problem for hypergraphs. 
The community detection problem asks us to partition the nodes of a hypergraph into useful or insightful subsets, which are often called \emph{communities}, \emph{clusters}, or simply \emph{groups}. 
Many algorithms exist for dyadic graphs~\cite{porter2009communities}. 
There is also a growing body of algorithms for hypergraph community detection; \citet{chodrowGenerativeHypergraphClustering2021a} give a brief survey of extant approaches and applications.
    
In the context of \revision{dyadic} graphs, spectral methods are a well-studied family of clustering algorithms. 
A spectral method for clustering a graph $\G$ proceeds by computing a distinguished set of eigenvectors of some matrix $\mM$ associated to $\G$. 
A \revision{Euclidean} clustering algorithm may then be used to obtain clusters from the embedding space defined by these eigenvectors. 
There are \revision{many} possibilities for the choice of matrix $\mM$. 
\revision{These include} the graph adjacency matrix~\cite{nadakuditiGraphSpectraDetectability2012}, various Laplacian matrices~\cite{shi2000normalized,von2007tutorial}, the modularity matrix~\cite{newman2006modularity}, and the nonbacktracking matrix~\cite{krzakalaSpectralRedemptionClustering2013}. 
\revision{While} greedy methods such as the famous \revision{modularity-maximizing} Louvain algorithm~\cite{blondel2008fast} \revision{are guaranteed only to find locally optimal solutions to cluster optimization problems}, many spectral methods can be interpreted as approximations of \revision{globally optimal} solutions \revision{to such problems}.

\revision{
    In this paper we study extensions of nonbacktracking spectral clustering to the setting of hypergraphs. 
    Several spectral methods, including nonbacktracking methods, exist for \emph{uniform} hypergraphs---in which all edges contain the same number of nodes. 
    Many interesting hypergraph data sets, however, are nonuniform, containing edges of multiple sizes~\cite{bensonSimplicialClosureHigherorder2018}. 
    Our extension focuses on the challenges and opportunities posed by such data. 
    In addition to the development of new spectral algorithms, we also contribute conjectures related to the possibility of community detection in generative random hypergraph models. 
}

The paper proceeds as follows. 
In \Cref{sec:related-work}, we survey related work on the dyadic nonbacktracking operator and its uses in graph data science, \revision{especially including its role in results on the theory of the detectability of communities in graphs}. 
We \revision{also} discuss spectral clustering techniques applied to hypergraphs, and discuss the hypergraph nonbacktracking operator $\mB$ introduced by~\citet{stormZetaFunctionHypergraph2006a}. 
In \Cref{sec:ihara-bass-nonuniform} we prove a theorem of Ihara-Bass type relating the spectrum of the nonbacktracking operator to that of a related matrix which is usually smaller. 
We also provide a first clustering algorithm based on this matrix.
In \Cref{sec:HSBM} we describe a hypergraph stochastic blockmodel (HSBM) and state in-expectation eigenrelations for the nonbacktracking operator.  
Then, in \Cref{sec:BP}, we study the relationship between the nonbacktracking operator and the belief-propagation algorithm in the HSBM. 
We prove a precise statement of the known heuristic~\cite{krzakalaSpectralRedemptionClustering2013,angeliniSpectralDetectionSparse2015} that the stability of an uninformative fixed point of the belief-propagation dynamics can be studied via the nonbacktracking operator. 
We also prove a relationship \revision{of Ihara-Bass type} between the Jacobian governing the stability of this fixed point and a smaller matrix, \revision{often} enabling faster computations. 
We use this relationship to propose an alternating spectral clustering algorithm based on belief-propagation. 
In \Cref{sec:thresholds}, we pose several conjectures on the spectral properties of hypergraph nonbacktracking operators, and derive from them conjectured thresholds bounding the ability of our proposed algorithms to detect clusters. 
We support these conjectures with experiments on synthetic data. 
We move on to several empirical data sets in \Cref{sec:experiments}, finding that \revision{hypergraph} spectral clustering \revision{using the belief-propagation Jacobian} outperforms methods based on the projected graph when edges of different sizes play statistically distinct roles. 
We conclude in \Cref{sec:discussion} with a discussion of limitations of our algorithm and suggestions for future work.



\section{Related Work} \label{sec:related-work}

We now introduce our primary notation and briefly discuss related work in clustering methods for graphs and hypergraphs. \revision{Table~\ref{tab:notation} gives a summary of our notation.}

\subsection{Notation and Preliminaries} \label{sec:notation}

A hypergraph $\HG$ is a tuple $(\nodes, \edges)$ consisting of a node set $\nodes$ and an edge set $\edges$. 
The node set $\nodes$ contains $n \eqdef \abs{\nodes}$ nodes. 
The edge set $\edges$ consists of subsets of $\nodes$. 
For each $k$ in a set $\Kset$ of possible edge sizes, we let $\edges_k$ denote the set of edges of size $k$. 
We let $\partial_ki\subset \cE_k$ denote the set of edges of size $k$ containing $i$, and we let $\partial i = \cup_{k \in \Kset} \partial_ki$. 
We let $m_k = \abs{\edges_k}$, and $m = \sum_{k \in \Kset} m_k$. 
We usually assume $\Kset = \{2,3,\ldots,\bar{k}\}$ for some maximum edge size $\bar{k}$, and set $\kappa \eqdef \abs{\Kset}$. 
A hypergraph is $k$-\emph{uniform} if we may take $\Kset = \{k\}$, and \emph{nonuniform} if it is not $k$-uniform for any $k$. 
A graph is a 2-uniform hypergraph; in this case we write $\G$ instead of $\HG$. 
We let $\bracket{k} = \frac{1}{m} \sum_{k\in \Kset}km_k$ be the empirical average edge size. 

We will consider many linear maps defined on structures related to graphs and hypergraphs. 
Accordingly, if $S$ is a finite set, we let $\vectorspace{S}$ be the vector space of formal sums of the form $\sum_{s \in S} sa_s$ for scalars $a_s$. 
For integer $j$, we define the notation $\vectorspace{jS} \eqdef \vectorspace{S} \oplus \vectorspace{S} \cdots \oplus \vectorspace{S}$ with $j$ direct summands. 
Given two finite sets $S$ and $T$, we let $\maps{S}{T}$ be the set of linear maps between $\vectorspace{S}$ and $\vectorspace{T}$. 
We abbreviate $\maps{S}{} \eqdef \maps{S}{S}$. 
We speak of elements of $\maps{S}{T}$ interchangeably as either linear maps or matrices. 
\revision{In the latter case, we regard $S$ and $T$ as bases} of $\vectorspace{S}$ and $\vectorspace{T}$, respectively. 
So defined, $\vectorspace{S}$ is isomorphic to $\R^{\abs{S}}$ and $\maps{S}{T}$ is isomorphic to $\R^{\abs{S}\times \abs{T}}$. 
Our notation is intended to emphasize the relationship between the many linear maps we will encounter and the structures on which they act. 

Elements of $\vectorspace{S}$ are written in lowerbase bold: $\vv \in \vectorspace{S}$. 
An entry of $\vv$ is $v_s$ for some $s \in S$. 
Elements of $\maps{S}{T}$ are written in uppercase bold: $\mA \in \maps{S}{T}$. 
An entry of $\mA$ is $a_{s,t}$ for $s \in S$ and $t \in T$. 
In many cases we will need to consider multiply-indexed structures; for example, an element $\vu \in \vectorspace{S\times T}$ has elements of the form $v_{st}$ for $s \in S$ and $t \in T$. 
An element $\mC \in \maps{S\times T}{S'\times T'}$ has elements of the form $c_{st,s't'}$. 
When considering an indexed family of matrices such as $\{\mA_1,\ldots,\mA_k\}$, we separate the family index with a semicolon when writing the entries, e.g., $a_{k;i,j}$ is the $(i,j)$-th entry of $\mA_k$. 
\revision{
    When later considering objects indexed by community labels, we use upper indexing with parentheses. 
    For example, a vector $\vc$ has typical entry $c^{(s)}$ when $s$ is a community label, and a matrix $\mC$ has entries $c^{(s,t)}$ when $t$ is also a community label.
} 
We use $\indicator{P}$ to denote the indicator function of the proposition $P$, and the shorthand $\delta_{i,j} \eqdef \indicator{i = j}$. 

We now define the nonbacktracking operator on hypergraphs.
\revision{This definition is due to}~\citet{stormZetaFunctionHypergraph2006a}. 
\begin{definition}[Pointed Line Graph] \label{dfn:pointed-line-graph}
    A \emph{pointed edge} $iQ$ in a hypergraph $\HG$ consists of an edge $Q \in \edges$ along with a choice of \emph{point} $i \in Q$. 
    Define $\pointed{\edges}$ to be the set of pointed edges. 
    \revision{Let $\pointed{\partial}_ki$ be the set of pointed edges of size $k$ with point $i$, and let $\pointed{\partial} i \eqdef \cup_{k \in \Kset}\pointed{\partial}_ki$}
    We say that pointed edge $jR$ \emph{follows} pointed edge $iQ$, written $iQ \rightarrow jR$, if $j \in Q\setminus i$ and $Q \neq R$. 
    The \emph{pointed line graph} $\cP$ of $\HG$ is a directed graph whose nodes are elements of $\pointed{\edges}$. 
    There is a directed edge $(iQ, jR)$ in $\cP$ if $iQ \rightarrow jR$.  
\end{definition}

\begin{definition}[Nonbacktracking Operator] \label{def:nonbacktracking}
    The nonbacktracking operator $\mB \in \maps{\pointed{\edges}}{}$ associated to a hypergraph $\HG$ is the directed adjacency operator of the pointed line graph $\cP$. 
    Its entries are $b_{iQ, jR} \eqdef \indicator{iQ \rightarrow jR}$. 
\end{definition}

\subsubsection{Nonbacktracking Methods in Network Data Science}

    The nonbacktracking matrix $\mB$ has found several applications in the the study of graphs (i.e., 2-uniform hypergraphs). 
    \Citet{alon2007non} show that random walks on graphs governed by the nonbacktracking matrix $\mB$ mix more rapidly than random walks governed by the adjacency matrix $\mA$, implying that these walks may be more efficient for graph exploration tasks. 
    \Citet{martin2014localization} propose an eigenvector centrality measure based on $\mB$, and show that this centrality avoids the pathological localization of classical adjacency-based eigenvector centrality in sparse graphs. 
    \Citet{torresNonbacktrackingCyclesLength2019} and~\citet{mellor2019graph} impose metrics on the space of point clouds in the complex plane. 
    These metrics enable the comparison of the spectra of two nonbacktracking operators, and induce a pseudometric on simple graphs. 
    The resulting pseudometrics can then be used for graph clustering tasks. 
    \Citet{torresNonbacktrackingEigenvaluesNode2021} develop perturbation theory for the eigenvalues of $\mB$ in order to identify influential nodes in spreading processes on graphs. 

    The nonbacktracking matrix plays an important role in the community detection problem on graphs. 
    In a simple graph with two planted clusters, the eigenvectors of $\mB$ can be used to assign labels to nodes correlated with the true clusters~\cite{krzakalaSpectralRedemptionClustering2013}. 
    \revision{We refer to this approach to community detection} as \emph{\graphsimplealgorithm{}} (\graphsimplealg{}). 
    
\revision{
    A standard way to study the theoretical behavior of many clustering algorithms, including \graphsimplealg{}, is to consider their behavior under generative models of random clustered graphs. 
    A common such model is the sparse binary planted partition model. 
    We begin with two ground-truth labeled communities of $n/2$ nodes. 
    We then draw edges between nodes independently in such a way that each node has, in expectation, $a$ edges joining it to other nodes in its own community and $b$ edges joining it to nodes in the opposite community.\footnote{See \citet{abbeCommunityDetectionStochastic2017} for a more detailed description of this model.} 
    The expected performance of a given algorithm can then be studied as a function of the parameters $a$ and $b$, usually as $n\rightarrow \infty$.
}

\revision{Several} spectral algorithms admit statistical guarantees for the recovery of planted clusters in the $n\rightarrow \infty$ limit~\cite{von2008consistency,lei2015consistency}. 
    \revision{
    There are also important limitations. 
    Heuristically, an algorithm \emph{detects communities} in a generative model if that algorithm is able to reliably return labels that have better-than-random correlation with ground truth.
    A series of deep results~\cite{nadakuditiGraphSpectraDetectability2012,krzakalaSpectralRedemptionClustering2013,bordenaveNonbacktrackingSpectrumRandom2015} have shown that various spectral graph clustering methods, including \graphsimplealg{}, are able to detect communities in the binary planted partition model as $n\rightarrow \infty$ if and only if 
    \begin{equation}
        \phi \eqdef \frac{1}{2}\frac{(a - b)^2}{a + b} > 1\;. \label{eq:graph-detectability}
    \end{equation}
    The quantity $\phi$ may be viewed as a signal-to-noise ratio for the clustering problem. 
    The general structure of such results is to show that, with high probability as $n$ grows large, the graph matrix used for clustering possesses an automatically selectable eigenvalue whose eigenvector can be used to approximately recover the two planted communities. 
    Formal concentration results along these lines for the nonbacktracking matrix $\mB$ are recently available~\cite{bordenaveNonbacktrackingSpectrumRandom2015}. 
    These results show that, when $\phi > 1$, with high probability, the second eigenvalue of $\mB$ is real and holds cluster information.
    } 

\revision{ 
    The equation $\phi = 1$ is often called the \emph{detectability threshold} for the sparse binary planted partition model. 
    It was conjectured~\cite{decelleAsymptoticAnalysisStochastic2011} and later proven~\cite{mossel2015reconstruction,mossel2018proof,massoulie2014community} that the condition $\phi > 1$ is both necessary and sufficient for the existence of algorithms (of any kind) that detect communities in this model. 
    The analysis of spectral methods can thus offer insight on the conditions under which the cluster detection task is possible at all.
    }

\subsection{Spectral Clustering Methods for Hypergraphs}

    There is a wide range of algorithms for clustering and community detection in hypergraphs; see~\citet{chodrowGenerativeHypergraphClustering2021a} for a recent overview. 
    We focus here on spectral methods, i.e., methods which make use of the eigenvectors of some object associated to the hypergraph. 
    \revision{One}  approach \revision{begins by replacing} the hypergraph with a dyadic graph. 
    This is frequently done by clique-projection, in which each edge $e$ of size $k$ is replaced by a clique of $\binom{k}{2}$, possibly-weighted 2-edges between each pair of nodes contained in $e$. 
    Graph spectral methods can then be applied~\cite{zhou2006learning}. 
    Multiple versions of this approach possess asymptotic consistency guarantees under hypergraph planted partition models~\cite{ghoshdastidarConsistencySpectralHypergraph2017,dumitriuPartialRecoveryWeak2021}. 
    A challenge for such approaches is the need to treat edges of multiple sizes.  
    \revision{I}n order to realize the desired statistical guarantees it is usually necessary to assume either a $k$-uniform hypergraph or a generative model in which edges of varying sizes carry similar information about community structure. 
    Both assumptions can be restrictive in practice. 
    
    Tensor methods can provide explicit representations of polyadic relationships. 
    A $k$-uniform hypergraph can be represented as a symmetric \revision{adjacency} $k$-tensor $\mA$.
    In the $k=3$ case, for example, we would have entry $a_{i,j,\revision{h}} = 1$ iff $\curlybrace*{i,j,\revision{h}} \in \edges$.  
    There are several spectral methods for clustering $k$-uniform hypergraphs, many of which rely on concepts connected to tensor eigenvalues and eigenvectors~\cite{lim2005singular}. 
    \citet{keCommunityDetectionHypergraph2019} use a normalized tensor power iteration to compute eigenvectors, while \citet{chang2020hypergraph} take an explicit optimization approach. 
    \citet{huMultiwaySphericalClustering2022} derive a clustering method for uniform hypergraphs based on angular separability, and provide several statistical guarantees. 
    The representation of hypergraphs using adjacency tensors is constrained by the need to represent all edges in a tensor of fixed dimensions; in particular, it is unclear how edges of multiple sizes should \revision{be} represented in the same tensor. 
    A proposal was recently made in this direction by~\citet{galuppiSpectralTheoryWeighted2021}, but the applications of the resulting tensor for data analysis problems remains to be explored. 
    We also note recent work by Mulas and collaborators developing spectral theory for hypergraphs via tensors~\cite{galuppiSpectralTheoryWeighted2021}, random walks~\cite{mulasRandomWalksLaplacians2021}, and Laplace operators~\cite{mulasSpectralTheoryLaplace2021,jostNormalizedLaplaceOperators2021}. 
    To our knowledge, the application of these techniques in the data scientific context has not yet been studied.

    The use of \revision{the nonbacktracking operator} $\mB$ for clustering uniform hypergraphs was considered by~\citet{angeliniSpectralDetectionSparse2015}.  
    These authors studied properties of this operator in the context of uniform hypergraphs, including a weakened version of the Ihara-Bass theorem and some conjectures regarding the locations of the eigenvalues of the operator. 
    The authors offered a conjecture on the detectability threshold in a \revision{uniform} hypergraph planted partition model. 
    Several of \rev{these conjectures} were recently formalized and proved by \citet{stephanSparseRandomHypergraphs2022}.


\section{An Ihara-Bass Theorem for Nonuniform Hypergraphs} \label{sec:ihara-bass-nonuniform}

A theorem often attributed to both \citet{iharaDiscreteSubgroupsTwo1966} and \citet{bassIharaSelbergZetaFunction1992}\rev{, originally formulated for graphs,} relates the spectrum of $\mB$ to the spectrum of a \revision{different matrix $\mB'$ which is usually smaller}. 
We now extend this theorem to nonuniform hypergraphs. 
Our result generalizes a computation by \citet{angeliniSpectralDetectionSparse2015} and theorem by \citet{stephanSparseRandomHypergraphs2022}. 
    
It is useful to distinguish blocks of $\mB$ by edge size. 
Let $\pointed{\edges}_k$ be the set of pointed edges of size $k$, let $\pointed{m}_k = \abs{\pointed{\edges}_k}$, and let $\pointed{m} = \sum_{k \in \Kset} \pointed{m}_k$. 
\begin{definition}[Size-Restricted Nonbacktracking Operators]\label{def:nonbacktracking-uniform}
    Let $\mB_{k'\rightarrow k} \in \maps{\pointed{\edges}_{k}}{\pointed{\edges}_{k'}}$ have entries $b_{k'\rightarrow k; iQ, jR} \eqdef \mathbbm{1}[iQ \rightarrow jR]\;, $
    where $iQ \in \pointed{\edges}_{k'}$ and $jR \in \pointed{\edges}_k$.
    The \emph{$k$-th nonbacktracking operator} $\bB_{k} \in \maps{\pointed{\edges}}{}$ associated to hypergraph $\HG$ has the block form 
    \begin{align*}
        \bB_k = \sqbracket*{\begin{matrix}
            \mathbf{0}_{2m_2 \times 2m_2} & \mathbf{0}_{2m_2 \times 3m_3} & \cdots & \mathbf{B}_{2\rightarrow k} & \cdots & \mathbf{0}_{2m_2 \times \kmax m_{\kmax}} \\
            \mathbf{0}_{3m_3 \times 2m_2} & \mathbf{0}_{3m_3 \times 3m_3} & \cdots & \mathbf{B}_{3\rightarrow k} & \cdots & \mathbf{0}_{3m_3 \times \kmax m_{\kmax}}\\ 
            \vdots & \vdots & \vdots & \vdots & \vdots & \vdots \\ 
            \mathbf{0}_{\kmax m_{\kmax} \times 2m_2} & \mathbf{0}_{\kmax m_{\kmax} \times 3m_3} & \cdots & \mathbf{B}_{\kmax\rightarrow k} & \cdots & \mathbf{0}_{\kmax m_{\kmax} \times \kmax m_{\kmax}}
        \end{matrix}}\;,
    \end{align*} 
    where $\mathbf{0}_{m_1\times m_2}$ is the matrix of zeros of specified dimensions. 
\end{definition}
By construction,  $\mB = \sum_{k \in \Kset} \mB_k$. 
Because the only indices in $\bB_k$ permitted to be nonzero in both rows and columns are the indices corresponding to $\mB_{k\rightarrow k}$, any eigenvalue of $\bB_k$ must either be zero or an eigenvalue of $\mB_{k\rightarrow k}$.

\begin{table}
	\centering
	\caption{Table of selected major notation used.}
	\label{tab:notation}
	\begin{tabularx}{\linewidth}{l X l}
        Symbol & Meaning & Introduced In \\
		\toprule
		$\HG = (\nodes, \edges)$ & a hypergraph with nodes  $\nodes$ and edges $\edges$& \Cref{sec:notation}\\
		$n = \abs{\nodes}$   & the number of nodes & \Cref{sec:notation}\\
		\revision{$m = \abs{\edges}$} & \revision{the number of edges} & \revision{\Cref{sec:notation}}\\
		$\Kset$, $\kappa$ & set of possible edge sizes of a hypergraph, $\kappa = \abs{K}$ & \Cref{sec:notation}\\
		\revision{$\overline{k}$} & \revision{the maximum edge size} & \revision{\Cref{sec:notation}}\\
		$\edges_k$, $m_k$ & the set of edges of size $k$, $m_k = \abs{\edges_k}$ & \Cref{def:nonbacktracking}\\ 
		\revision{$\partial i$, ($\partial_k i$)} & \revision{the set of edges (of size $k$) that contain node $i$} & \revision{\Cref{sec:notation}} \\
		\midrule
		$iQ$ & a pointed edge, $i \in \nodes$ and $Q \in \edges$ & \Cref{dfn:pointed-line-graph}\\
		$\pointed{\edges}$, $\pointed{m}$ & set of pointed edges of $\HG$, $\pointed{m} = \abs{\pointed{\edges}}$ & \Cref{def:nonbacktracking}\\
		$\pointed{\edges_k}$, $\pointed{m}_k$ & set of pointed edges of size $k$, $\pointed{m}_k = \abs{\pointed{\edges}_k}$ & \Cref{def:nonbacktracking}\\
		\revision{$\pointed{\partial} i$}, ($\pointed{\partial}_{k}i$) & {the set pointed of edges (of size $k$) that contain node $i$} & \Cref{dfn:pointed-line-graph} \\ 
		$\cP$     & pointed line graph of a hypergraph  & \Cref{def:nonbacktracking}\\ 
		\midrule
		$\vectorspace{S}$ & finite vector space with orthonormal basis indexed by elements of set $S$ & \Cref{sec:notation} \\
		$\maps{S}{T}$ & Space of linear maps $\vectorspace{S}\rightarrow \vectorspace{T}$ & \Cref{sec:notation} \\ 
        $\mB_{k\rightarrow k'}$ & nonbacktracking operator from $k$-edges to $k'$-edges & \Cref{def:nonbacktracking-uniform}\\
        $\mB_{k}$ & nonbacktracking operator from $k$-edges to all edges & \Cref{def:nonbacktracking-uniform}\\
		$\mB$ & hypergraph nonbacktracking operator & \Cref{def:nonbacktracking-uniform}\\
		$\mA_k, \mD_k$ & $k$-th adjacency and degree operators for a hypergraph & \Cref{def:hypergraph-adjacency}\\
		\revision{$\mA, \mD$} & \revision{block matrix with diagonal blocks of $\mA_k$ or $\mD_k$} & \revision{\Cref{def:hypergraph-adjacency}}\\
		$\mK$      & square matrix with diagonal entries in $\Kset$ & \Cref{sec:ihara-bass-nonuniform}\\
        \midrule 
		$\tuples$, $\tuples_k$ & set of all subsets or $k$-subsets of nodes & \Cref{sec:HSBM}\\ 
		\revision{$\tuples(i)$}, $\tuples_k(i)$ & set of all subsets or $k$-subsets of nodes containing node $i$ & \Cref{sec:HSBM}\\ 
        $x\doteq y$ & $x = (1 + O(n^{-r}))y$ w.h.p. for some $r > 0$. & \Cref{sec:HSBM} \\ 
        \revision{$\eta$} & \revision{distribution of hypergraphs} & \revision{\Cref{sec:HSBM}} \\
        \revision{$\vz$} & \revision{labels for nodes in the blockmodel} & \revision{\Cref{sec:HSBM}} \\
		$\ck{k}{s}$ & Mean $k$-degree of nodes in cluster $s$. & \Cref{eq:ck}\\
		$\mJ$ & Non-vanishing block of the Jacobian matrix of belief-propagation evaluated at the uninformative fixed point. &  \Cref{thm:BP} \\ 
		\midrule
		\oursimplealg{} & \OurSimpleAlgorithm{}{} &\Cref{alg:vanilla} \\ 
		\ouralg{} & \OurAlgorithm{} &\Cref{alg:BP-spectral-clustering}\\ 
		\bottomrule
	\end{tabularx}
\end{table}
 
We now define adjacency and degree operators\rev{. We also distinguish these }by edge size. 
\begin{definition}[Adjacency and Degree Operators] \label{def:hypergraph-adjacency}
    The \emph{$k$-th adjacency operator} $\mA_k \in \maps{\nodes}{}$ associated to $\HG$ has entries
    \begin{equation*}
        a_{k; i,j} \eqdef \sum_{R \in \edges_k} \mathbbm{1}[\curlybrace*{ i,j } \subseteq R]\;.
    \end{equation*}
    The \emph{$k$-th degree matrix} $\mD_k \in \maps{\nodes}{}$ associated to $\HG$ has entries
    \begin{equation*}
        d_{k; i,j} \eqdef \delta_{i,j}\sum_{R\in \edges_k} \mathbbm{1}[i \in R]\;. 
    \end{equation*}
\end{definition}

Define the block matrices  
\begin{align*}
    \bA \eqdef \sqbracket*{\begin{matrix}
        \mA_2 & \cdots & \mA_{\kmax} \\
        \vdots & \ddots & \vdots \\ 
        \mA_{2}& \cdots & \mA_{\kmax}
    \end{matrix}} \in \maps{\Kset \times \nodes}{}\quad \text{and} \quad 
    \bD \eqdef \sqbracket*{\begin{matrix}
        \mD_2 & \cdots & \mD_{\kmax}\\ 
        \vdots & \ddots & \vdots \\ 
        \mD_{2}& \cdots & \mD_{\kmax}
    \end{matrix}} \in \maps{\Kset \times \nodes}{}\;,
\end{align*}
where $\mA_k$ and $\mD_k$ are as in \Cref{def:hypergraph-adjacency}. 
Let $\mK$ be the square matrix \revision{of size $\kappa \times \kappa$} with entries $(2, 3, \ldots,\kmax)$ along the diagonal, and zeroes everywhere else.

\begin{theorem}[Ihara-Bass for nonuniform hypergraphs]\label{thm:ib-nonuniform}
    For any hypergraph $\HG$, we have 
    \begin{equation}
        \det(\mI - \mu \bB) = f_\HG(\mu) \det \paren*{\mI_{\kappa n} + \mu ((\mK - 2\mI_{\kappa})\otimes \mI_n - \bA) + \mu^2 \revision{(\bD - \mI_{\kappa n})}((\mK - \mI_{\kappa}) \otimes \mI_n)}\;,  \label{eq:ib-equation}
    \end{equation}
    where 
    \begin{equation*}
        f_\HG(\mu) = \prod_{k\in \Kset} (1 - \mu)^{m_k(k-1) - n}(1 + \mu(k-1))^{m_k - n}\;,
    \end{equation*}
    and $\otimes$ is the Kronecker product.
\end{theorem}
We supply a proof of \Cref{thm:ib-nonuniform}, as well as of \Cref{cor:ihara-bass-hypergraph-eigenvalues-nonuniform,cor:ib-eigenvector-correspondence} below, in \Cref{sec:ib-nonuniform-proof}.   

The computational significance of \Cref{thm:ib-nonuniform} is that we can compute the interesting eigenvalues of $\bB$ via a \rev{different} matrix $\bB'$\rev{. This matrix is usually smaller than $\bB$, although it can also be less sparse}.  
\begin{cor}\label{cor:ihara-bass-hypergraph-eigenvalues-nonuniform}
    Let $\HG$ be a hypergraph with $m_k$ edges of size $k$ for each $k \in \Kset$. 
    Then: 
    \begin{itemize}
        \item For each $k$, if $m_k > n$, then  $\beta = 1-k$ is an eigenvalue of $\bB$ with algebraic multiplicity at least $m_k - n$. 
        \item If $\sum_{k\in \Kset} m_k(k-1) > \kappa n$, then $\beta = 1$ is an eigenvalue of $\bB$ with algebraic multiplicity at least $\sum_{k\in \Kset} m_k(k-1) - \kappa n$. 
        \item The remaining eigenvalues of $\bB$ are eigenvalues of the matrix 
        \begin{align}
            \bB' \eqdef \sqbracket*{\begin{matrix}
                \vzero_{\kappa n} &  \bD - \mI_{\kappa n} \\ 
                ( \mI_{\kappa}- \mK)\otimes \mI_n &  \bA + (2\mI_{\kappa} - \mK)\otimes \mI_n
            \end{matrix}} \in \maps{2\Kset \times \nodes}{}\;.\label{eq:b-prime-nonuniform}
        \end{align} 
    \end{itemize}
\end{cor}

\Cref{cor:ihara-bass-hypergraph-eigenvalues-nonuniform} expresses a relationship between the eigenvalues of $\bB$ and $\bB'$. 
There is an associated relationship between their eigenvectors. 

\begin{lemma} \label{cor:ib-eigenvector-correspondence}
    Let $\vu \in \vectorspace{\pointed{\edges}}$ be an eigenvector of $\bB$ with eigenvalue $\beta$. 
    Let $\vx = (\vx_1, \vx_2)^T\in \vectorspace{2\Kset\times \nodes}$, where $\vx_1, \vx_2 \in \vectorspace{\Kset\times \nodes}$ are doubly-indexed vectors defined entrywise by 
    \begin{align}
        x_{1;k,i} \eqdef \revision{\sum_{\substack{\substack{Q \in \partial_ki}}} \sum_{\substack{j \in Q\setminus i}} u_{jQ} }  \quad \text{and} \quad
        x_{2;k,i} \eqdef \revision{\sum_{\substack{Q \in \partial_ki}}u_{iQ}}\;. \label{eq:reduced-eigenvectors}
    \end{align}
    Then, $\bB'\vx = \beta \vx$. 
    In particular, either $\vx$ is an eigenvector of $\bB'$ with eigenvalue $\beta$ or $\vx = \vzero$. 
\end{lemma}
 
Heuristically, \Cref{cor:ib-eigenvector-correspondence} states that one can aggregate cluster information in $\bu$ to the level of nodes by summing over all edges of size $k$ with specified points.  
To \revision{complete} a first spectral clustering algorithm, we further sum $\vx_{\revision{2}}$ over edge sizes, obtaining vector $\bar{\vx}$ with entries 
\begin{equation}
    \bar{x}_{i} \eqdef \sum_{k\in \Kset} x_{\revision{2};k,i}\;.  \label{eq:aggregate-eigenvectors}
\end{equation}
\rev{In the case of uniform hypergraphs, there exist recent probabilistic guarantees ensuring that \rev{$\bar{\vx}$} is correlated with planted communities under a certain choice of data generating process~\cite{stephanSparseRandomHypergraphs2022}. 
Extending these guarantees to the nonuniform cases is an avenue of future work.}
We find experimentally that the sign of $\bar{x}_i$ is most informative, and it is therefore helpful to use the vector $\tilde{\vx} \eqdef \mathrm{sgn}(\bar{\vx})$, with the sign computed entrywise.\footnote{\rev{see \Cref{fig:threshold-experiment} for a simple experiment supporting the use of $\tilde{\vx}$ rather than $\bar{\vx}$.}}

We now state our first spectral clustering algorithm, \OurSimpleAlgorithm{} (\oursimplealg) in~\Cref{alg:vanilla}. 
We compute a desired number of eigenvectors $\curlybrace*{\vu^{(\ell)}}$, form $\tilde{\vx}^{(\ell)}$ for each, and then cluster in the Euclidean embedding described by the set $\curlybrace*{\tilde{\vx}^{(\ell)}}_\ell$. 
There are multiple choices for the Euclidean clustering subroutine. 
We assume that this subroutine accepts as input a matrix giving the coordinates of a point cloud in Euclidean space, and returns a vector of labels $\labelvec$. 
Throughout this paper we use the standard $k$-means algorithm \cite{macqueen1967some,sebestyen1962decision}, but other choices could in principle lead to superior performance. 
\begin{algorithm} 
    \caption{\OurSimpleAlgorithm{} (\oursimplealg{})}
    \label{alg:vanilla}
    \begin{algorithmic}[1]
    \STATE $\curlybrace*{\vu^{(\ell)}}_{\ell = 1}^h \gets$ $h$\text{ eigenvectors of }$\bB'$\text{ \revision{with real eigenvalues largest in magnitude}}
    \STATE Initialize $\tilde{\mX}$ 
    \FOR{$\ell = 2,\ldots,h $}
        \STATE Compute $\bar{\vx}^{(\ell)}$ via \cref{eq:reduced-eigenvectors,eq:aggregate-eigenvectors}
        \STATE $\tilde{\mX}_{\cdot\ell} \gets \mathrm{sgn}(\bar{\vx}^{\ell})$\;
    \ENDFOR
    \STATE $\labelvec = $ \text{Cluster}$(\tilde{\mX})$
    \RETURN $\labelvec$
    \end{algorithmic}
\end{algorithm}

\revision{
    Carrying out \oursimplealg{} requires the user to choose $h$, the number of eigenvectors with real eigenvalues to extract from $\mB'$. 
    Analogy with the uniform hypergraph case \citep{stephanSparseRandomHypergraphs2022} would suggest the use of $h-1$ eigenvectors to cluster into $h$ communities, provided that there are indeed $h-1$ such eigenvectors with eigenvalues separated from the bulk other than an uninformative eigenvector corresponding to the largest real eigenvalue. 
    We will later argue that \oursimplealg{} is limited in cases in which edges of different sizes carry different cluster information in $\HG$ and that no choice of $h$ can circumvent this limitation.
    To highlight these limitations and build a foundation for further development, we now discuss a generative model of random hypergraphs and study the spectral structure of $\mB$ under this model. 
    }






\section{The Sparse Hypergraph Stochastic Blockmodel} \label{sec:HSBM}

    In this section we briefly review the hypergraph stochastic blockmodel (HSBM), a generative model of clustered hypergraphs. 
    Our choice of notation and formulation most closely resembles that of~\citet{keCommunityDetectionHypergraph2019}. 
    Many other related formulations exist in the literature \cite{ghoshdastidarConsistencySpectralHypergraph2017,angeliniSpectralDetectionSparse2015,chodrowGenerativeHypergraphClustering2021a,dumitriuPartialRecoveryWeak2021,stephanSparseRandomHypergraphs2022,contiscianiPrincipledInferenceHyperedges2022}. 
    We prove in-expectation results for eigenpairs of the matrices $\mB_k$, and pose conjectures generalizing recent proofs by \citet{stephanSparseRandomHypergraphs2022} of eigenpair concentration results in the uniform case. 
    These conjectures will also inform our development of \revision{\ouralgorithm{}} in \Cref{sec:BP}.     
    
    Our blockmodel is a probability distribution over hypergraphs.
    \rev{We denote this distribution by} $\eta$. 
    To sample from this model, we first assign each node $i\in \nodes$ a label $z_i$ from a finite label alphabet $\alphabet$ of size $\ell$. 
    \rev{These labels are drawn independently from a probability vector $\vq \in \vectorspace{\alphabet}$, so that $\q{s}$ gives the probability that $z_i = s$ for each $i\in \nodes$.}
    We collect the labels in a vector $\labelvec$. 
    Let $\tuples$ give the set of possible edges, which we usually take to be sets of nodes with some specified possible sizes. 
    Let $\tuples_k$ denote the set of node subsets of size $k$. 
    Define $\tuples(i)$ to be the set of subsets containing node $i$, and let $\tuples_k(i) = \tuples_k \cap \tuples(i)$. 
    To realize edges, we consider each set of nodes $R\in \tuples$ and add this set to the edge set $\edges$ with probability $\eta(R \in \edges|\labelvec_R)$ depending on the labels in $\labelvec_R$.
    \rev{The structure of the model is specified by the functional form of $\eta$}.
    Here, we will focus on the sparse setting in which 
    \begin{equation}
        \eta(R \in \edges|\labelvec_R) = \frac{\omega(\labelvec_R)   }{n^{\abs{R}-1}}
    \end{equation}
    for some function $\omega$ that does not depend on $n$. 
    This structure imposes sparsity on the hypergraph; the number of $k$-edges for a given node is asymptotically constant with respect to $n$.  
    The overall probability to realize a given combination of label vector $\labelvec$ and edge set $\edges$ is 
    \begin{equation}
        \eta(\edges, \labelvec) = \paren*{\prod_{i\in \nodes} \q{z_i}} \paren*{\prod_{R \in \tuples} \eta(R\in \edges|\labelvec_R)}\;.\label{eq:blockmodel}
    \end{equation}
    \revision{
        The function $\omega$ controls the relative rates of edges within and between communities, and is an extension of the parameters $a$ and $b$ in the graph planted partition model described in \Cref{sec:intro}. 
        A common choice of the function $\omega$ is the ``all-or-nothing'' affinity, defined as 
        \begin{align*}
            \omega(z_R) = 
            \begin{cases}
                \rev{a}_{\abs{R}}  &\quad z_i = z_j \quad \forall i, j \in R \\ 
                \rev{b}_{\abs{R}} & \quad \text{otherwise}\, .
            \end{cases}
        \end{align*}
        When $\rev{a_{k} > b_{k}}$, edges \revision{of size $k$} form at higher rates \emph{within} planted communities than \emph{between} them. 
        This choice of $\omega$ is also the partition used by \citet{chodrowGenerativeHypergraphClustering2021a} to derive a fast modularity maximization algorithm for hypergraphs. 
        However, many other choices for $\omega$ are also possible, including ones that favor edge formation between communities rather than within them. 
        Our work is independent of the choice of $\omega$, subject only to the assumptions described in \Cref{sec:community-correlated-eigenvectors}. 
    }

    In our development below, it will be useful to reason asymptotically about many quantities. 
    Define $x \doteq y$ if there exists some constant $r>0$ such that $x = (1 + O(n^{-r}))y$ either deterministically or with high probability with respect to the blockmodel $\eta$ as $n$ grows large and the largest edge-size $k$ remains fixed. 

    Let 
    \begin{equation}
        \ck{k}{s} \eqdef \frac{1}{(k-1)!}\sum_{\labelvec \in \alphabet^{k-1}}\omega(\labelvec, s) \prod_{t \in \vz}\q{t}\;. \label{eq:ck}
    \end{equation}
    This is the asymptotic expected number of $k$-edges attached to a given node $i$ in cluster $s$, as can be verified via direct calculation: 
    \begin{equation*}
        \sum_{\substack{R \in \tuples_k(i)}} \eta(R \in \edges | \labelvec_R) \prod_{j \in R\setminus i} \q{z_j} 
        \doteq \binom{n-1}{k-1} \sum_{\labelvec \in \alphabet^{k-1}} n^{1-k}\omega(\labelvec, s) \prod_{t \in \vz}\q{t} = (1 + O(n^{-1}))\ck{k}{s} \doteq \ck{k}{s}\;. 
    \end{equation*}
    We also define 
    \begin{equation}
        \ck{k}{s,t} \eqdef \frac{1}{(k-2)!}\sum_{\labelvec \in \alphabet^{k-2}}\omega(\labelvec, s, t)\prod_{t \in \vz}\q{t}\;, \label{eq:c2-def}
    \end{equation}
    which gives the expected number of cluster $t$ neighbors of a node in cluster $s$ through edges of size $k$. 
    We have the identity 
    \begin{equation}
        \ck{k}{s} = \frac{1}{k-1}\sum_{t\in \alphabet}\q{t}\ck{k}{s,t}\;, \label{eq:degree-identity}
    \end{equation}
    which computes $\ck{k}{s}$ by conditioning on the label of a second node in each possible edge.

\subsection{Spectral Structure of $\mB$ in the HSBM} \label{sec:community-correlated-eigenvectors}
    
    We now aim to describe \rev{the} eigenvectors of $\mB_k$ \rev{which are correlated with community structure} under the blockmodel \cref{eq:blockmodel}. 
    We implicitly condition on a realization of the label vector $\labelvec$ for simplicity, so that the only remaining randomness is in the edge generation process. 
    We impose the following assumptions: 
    \begin{itemize}
        \item There are exactly two groups, labeled $1$ and $2$, of equal expected size. 
        Thus, $\q{1} = \q{2} = \frac{1}{2}$.  
        We further assume that the empirical distribution of group labels in the realized label vector $\vz$ is close to $q$; this holds with high probability as $n$ grows large. 
        \item The group-specific expected degrees are equal: $\ck{k}{1} = \ck{k}{2} \eqdef \ck{k}{}$ for each $k$. 
        \item We have $\ck{k}{1,1} = \ck{k}{2,2}\eqdef \ckin$ and $\ck{k}{1,2} = \ck{k}{2,1}\eqdef \ckout$ for all $k$. 
    \end{itemize}
    Formal concentration results under these hypotheses are available for graphs \cite{bordenaveNonbacktrackingSpectrumRandom2015} and $k$-uniform hypergraphs \cite{stephanSparseRandomHypergraphs2022}. 
    While we anticipate that many of the techniques used for these cases will generalize to nonuniform hypergraphs, pursuing formal proofs is beyond our present scope. 
    We \revision{instead} provide informal results by reasoning in expectation. 

    Let $\vu\in \vectorspace{\pointed{\edges}}$ have entries $u_{iQ} = \abs{Q} - 1$.    
    Let $\vsigma \in \vectorspace{\nodes}$ be the vector with entries
    \begin{equation*}
        \sigma_i \eqdef \begin{cases}
            + 1 &\quad z_i = 1 \\ 
            -1 &\quad z_i = 2\;.
        \end{cases}
    \end{equation*}
    Let $\vv \in \vectorspace{\pointed{\edges}}$ have entries $v_{iQ} \eqdef \sum_{j \in Q\setminus i}\sigma_{j}$. 
    
    \begin{theorem}\label{thm:eigen-expectation}
        Consider a hypergraph sampled from $\eta$. 
        Let $\alpha_k = \ck{k}{}(k-1)$. 
        Let $\beta_k = \frac{\ckin - \ckout}{2}$. 
        Then, for all tuples $Q$ and nodes $i\in Q$, we have 
        \begin{align}
            \E\sqbracket*{(\mB_k\vu)_{iQ} - \alpha_ku_{iQ}|Q\in \edges} \doteq 0\;, \label{eq:eigen-expectation-1} \\ 
            \E\sqbracket*{(\mB_k\vv)_{iQ} - \beta_kv_{iQ}|Q\in \edges} \doteq 0\;. \label{eq:eigen-expectation-2}
        \end{align}
    \end{theorem} 
    The full proof of \cref{eq:eigen-expectation-2} is provided in \Cref{sec:eigen-expectation-proof}; the proof of \cref{eq:eigen-expectation-1} is \rev{ommitted but} similar.

    In expectation, \Cref{thm:eigen-expectation} says that, $\mB_k$ has a Perron eigenvector with eigenvalue $\alpha_k$, and a community-correlated eigenvector with eigenvalue $\beta_k$. 
    To substantiate the claim that $\vv$ is indeed community-correlated, we can sum over $k-$edges incident to $i$, obtaining  
    \begin{equation}
        \E\sqbracket*{\sum_{Q\in \partial_ki} v_{iQ}} 
        = \sum_{Q\in \tuples_k(i)}\eta(Q \in \edges_k)\sum_{j \in Q\setminus i}\sigma(j) 
        \doteq \frac{1}{n}\sum_{j \neq i}\sigma_j \ck{k}{z_i, z_j}
        \doteq \beta_k \sigma_i\;. \label{eq:community-expectation}
    \end{equation} 
    This calculation is similar to the proof of \Cref{cor:ib-eigenvector-correspondence}. 
    \rev{Indeed, this calculation shows that  $\E[\vx_{2;k}] \doteq \beta_k \vsigma$ with $\vx_{2;k}$ as in \Cref{cor:ib-eigenvector-correspondence}, explicitly showing why the vector $\vx_{2}$ can be expected to be correlated with planted communities.}
    In practice, we do not have access to the expectation and only to our \revision{realization of $\vx_{2}$, and therefore our estimate of $\vsigma$, is} noisy. 
    Formal concentration results along the lines of \citet{bordenaveNonbacktrackingSpectrumRandom2015} and~\citet{stephanSparseRandomHypergraphs2022} are necessary to provide \revision{probabilistic} guarantees. 

    \revision{
        From \Cref{thm:eigen-expectation} we can also obtain in-expectation eigenpair relations for the full matrix $\mB$. 
        \begin{cor} \label{cor:eigen-expectation-B}
            Under the hypotheses of \Cref{thm:eigen-expectation}, and with $\vu$ and $\vv$ as defined there, 
            \begin{align}
                \E\sqbracket*{(\mB\vu)_{iQ} - \alpha u_{iQ}|Q\in \edges} &\doteq 0\; \\  
                \E\sqbracket*{(\mB\vv)_{iQ} - \beta v_{iQ}|Q\in \edges} &\doteq 0\;, 
            \end{align}
            where $\alpha \eqdef \sum_{k \in \Kset}\alpha_k$ and $\beta \eqdef \sum_{k \in \Kset}\beta_k$. 
        \end{cor} 
        \begin{proof}
            The corollary follows from the formula $\mB = \sum_{k \in \Kset} \mB_k$ and the fact that the vectors $\vu$ and $\vv$ do not depend on $k$. 
        \end{proof}
    }

    \subsection{Implications for \oursimplealg{}} \label{sec:vanilla-thresholds}

        In \cref{alg:vanilla}, we form the nonbacktracking operator $\bB$, obtain a spectral embedding of the nodes, and then cluster in the embedding space. 
        The spectral embedding is obtained by considering the \revision{
            $h-1$  eigenvectors of $\bB$ with real eigenvalues largest in magnitude, excepting the uninformative eigenvector $\vu$.
            A necessary condition for \oursimplealg{} to perform better-than-random guessing is that at least one of these eigenvectors is correlated with ground-truth community labels. 
            }
        
        \begin{figure}
            \centering 
            \includegraphics[width=\textwidth]{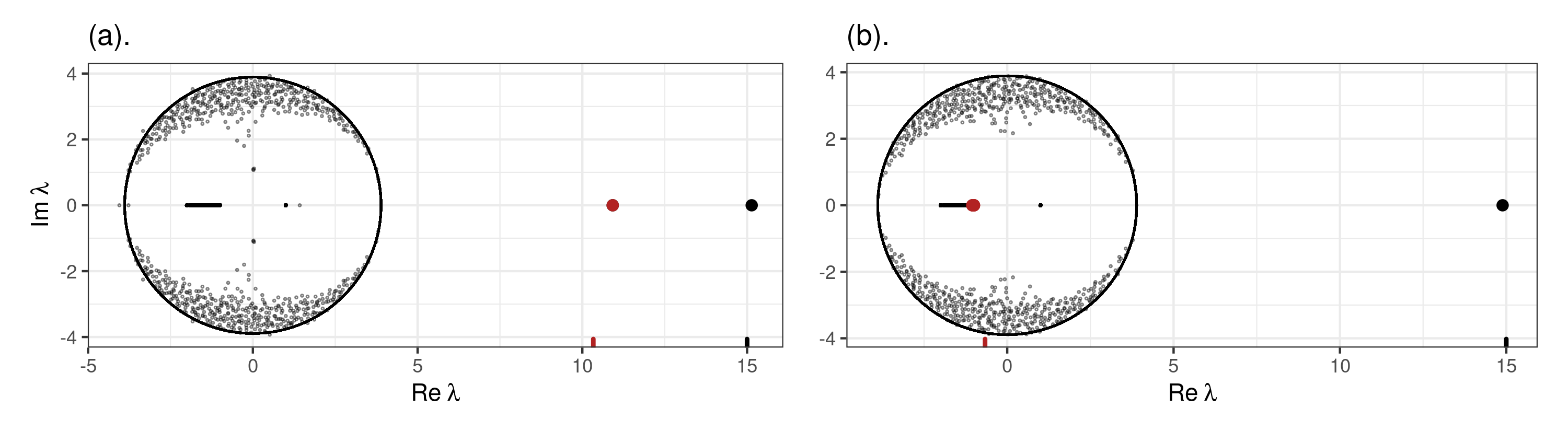}
            \caption{
                \revision{
                    (a): Spectrum of $\mB$ in an HSBM with \rev{edges of size $2$ and $3$, with} $c_2 = c_3 = 5$.
                    In expectation, 80\% of both 2-edges and 3-edges \rev{contain nodes within a single cluster} while 20\% contain nodes from both clusters. 
                    The uninformative real eigenvalue with largest magnitude and community-informative eigenvalue are highlighted (far right). 
                    Ticks on the bottom margin give the theoretical predictions $\alpha$ and $\beta$ for these two eigenvalues \rev{given by \Cref{cor:eigen-expectation-B}}. 
                    In this case, spectral clustering using the eigenvector $\vv$ associated to the second eigenvalue in magnitude achieves Adjusted Rand Index (ARI) 0.98 against ground truth. 
                    The solid circle has radius $\sqrt{\alpha}$. 
                    (b): As in (a), but \rev{only} 10\% of 2-edges and 50\% of 3-edges are within-cluster. 
                    A community-correlated eigenvalue is again present, but lies inside the bulk of the spectrum. 
                    Clustering based on the associated eigenvector would achieve ARI 0.87 against ground truth, but it is not possible to automatically distinguish this informative eigenvalue from nearby uninformative ones. 
                    }
                }\label{fig:eigenvalue-locations-vanilla}
        \end{figure}

        It was argued informally by \citet{angeliniSpectralDetectionSparse2015} and recently proven by \citet{stephanSparseRandomHypergraphs2022} that, in the uniform case, eigenvalues other than $\alpha$ and $\beta$ concentrate with high probability \rev{as $n$ grows large} in the disc of radius $\sqrt{\alpha}$. 
        The informal argument generalizes smoothly. 
        A generalized formal argument, alongside concentration results on the eigenvalues of $\mB$ that strengthen \Cref{cor:eigen-expectation-B}, would be sufficient to prove the following conjecture: 
        \begin{conj} \label{conj:vanilla}
            In the blockmodel $\eta$, with high probability as $n\rightarrow \infty$, 
            \begin{enumerate}
                \item The \revision{eigenvalue of largest magnitude} of $\mB$ \rev{approaches} $\alpha$. 
                \item There exists a community-correlated eigenvector with an associated eigenvalue \rev{approaching} $\beta$. 
                \item The bulk of the spectrum of $\bB$ is confined to the disk of radius \revision{$\sqrt{\alpha}$}. 
                \item \oursimplealg{} is able to \revision{recover labels correlated with ground-truth clusters} iff 
                \begin{equation}
                    \beta^2 > \alpha\;. \label{eq:detectability}
                \end{equation}         
            \end{enumerate}
        \end{conj}
        These conjectures parallel known results for graphs \cite{bordenaveNonbacktrackingSpectrumRandom2015} and uniform hypergraphs \cite{stephanSparseRandomHypergraphs2022}. 

        \Cref{conj:vanilla} highlights a limitation of \oursimplealg{} in the setting in which edges of different sizes reflect cluster structure in different ways. 
        Consider, for example, a setting with hyperedge sizes $k$ and $k'$ in which edges of size $k$ tend to be within-cluster ($\beta_k > 0$), while edges of size $k'$ tend to be between-cluster $(\beta_{k'} < 0$). 
        In such a case, the approximate eigenvalue $\beta = \beta_k + \beta_{k'}$ of $\mB$ can be smaller in magnitude than either of $\beta_{k}$ and $\beta_{k'}$, and could potentially be fully lost in the eigenvalue bulk, causing spectral \revision{clustering} to fail. 
        This can occur even when $\beta_k$ and $\beta_{k'}$ are sufficiently large in magnitude that \revision{clustering} based on edges of size $k$ or $k'$ alone would succeed. 
        \revision{This phenomenon is illustrated in \Cref{fig:eigenvalue-locations-vanilla}(b), in which an informative real eigenvalue lies hidden in the bulk of the spectrum and is hidden among other real, uninformative eigenvalues.}
        The need to adaptively synthesize signals from hyperedges of differing sizes is one motivation of spectral methods based on the belief-propagation Jacobian, \revision{which we develop now}.


\section{Nonbacktracking Operators and Belief Propagation} \label{sec:BP}

We \revision{will} derive a connection between nonbacktracking operators and the belief-propagation algorithm (BP) \cite{bishopPatternRecognitionMachine2006} for community detection in sparse hypergraphs. 
Our derivation extends arguments by \citet{krzakalaSpectralRedemptionClustering2013} in dyadic graphs and  \citet{angeliniSpectralDetectionSparse2015} in uniform hypergraphs. 
We first sketch out the bird's-eye view of the argument. 
BP is known to  be exact on graphical models which are trees, and often reliable on sparse models with local tree-like structure. 
Sparse stochastic blockmodels are an example of \rev{the latter case}. 
Following standard arguments \cite{krzakalaSpectralRedemptionClustering2013}, we show that under certain additional symmetry assumptions, the approximate belief-propagation dynamics have a distinguished fixed point which contains no cluster information. 
Perturbations around this point are encoded in a Jacobian matrix---expressible in terms of the nonbacktracking operators $\mB_k$---whose eigenvectors may therefore contain cluster information. 

\subsection{Belief-Propagation Algorithm}

We will work in the framework of the hypergraph stochastic blockmodel described in \Cref{sec:HSBM}. 
In detection problems, we assume that we observe the edge set $\edges$ but not the label vector $\labelvec$. 
We would like to obtain information about the conditional distribution $\eta(\labelvec|\edges) = \frac{\eta(\edges, \labelvec)}{\eta(\edges)}$ of labels given the edge data $\edges$. 
An especially relevant summary is the marginal distribution of labels for each node, $\eta(z_i|\edges)$. 

We will estimate the marginals using belief-propagation~\cite{bishopPatternRecognitionMachine2006}. 
We treat both subsets $R \in \tuples$ and labels $z_i \in \alphabet$ as factors in a factor graph. 
There is a label factor for each node $i$. 
The message that each label factor sends to node $i$ about its belief that $z_i = s$ is $\q{s}$. 
The \rev{node subset} factors are somewhat more complex.   
Let $\mess{i}{R}{s}$ denote the message that node $i$ passes to the factor $R$ expressing its belief that $z_i = s$. 
Let $\varmess{i}{R}{s}$ denote the message that factor $R$ passes to node $i$ expressing its belief that $z_i = s$. 
Then, the standard BP updates for this model read 
\begin{align}
    \mess{i}{R}{s} &\gets \frac{1}{Z_{iR}}\q{s}\prod_{\substack{Q \in \tuples(i) \setminus R}}\varmess{i}{Q}{s} \label{eq:BP-1} \\ 
    \varmess{i}{R}{s} &\gets \frac{1}{Z_{Ri}} \sum_{\labelvec_R:z_i = s} \eta(R\in \edges|\labelvec_R) \prod_{j \in R\setminus i} \mess{j}{R}{z_j}\;. \label{eq:BP-2}
\end{align}
Here, $Z_{iR}$ and $Z_{Ri}$ are normalizing constants ensuring that $\sum_{s} \mess{i}{R}{s} = \sum_{s} \varmess{i}{R}{s} = 1$. 
\rev{Here and below, sums over label vectors $\vz$ should be assumed to run over $\alphabet^{\ell}$ subject to the explicitly specified constraints.}

On factor graphs which are trees, the updates \cref{eq:BP-1,eq:BP-2} converge, and the desired marginals can be obtained by computing the marginal message for node $i$:  
\begin{equation}
    \mess{i}{}{s} \eqdef \frac{1}{Z_i}\q{s} \prod_{Q \in \tuples(i)} \varmess{i}{Q}{s}\;. \label{eq:marginal-message}
\end{equation}
Our factor graph is admittedly not a tree. 
One possibility is to modify the belief propagation algorithm to account for loops~\cite{kirkleyBeliefPropagationNetworks2021}.
We instead follow the standard argument that: 
\begin{itemize}
    \item The sparsity assumption on $\eta$ implies that the realized hypergraph is locally tree-like as $n$ grows large. 
    \item The updates \cref{eq:BP-1,eq:BP-2} can be approximated by updates that travel only along the edges of the realized hypergraph. 
\end{itemize}
Combined, these two points imply that BP can be approximated by an algorithm that takes place on a \rev{locally tree-like} factor graph. 
The first point is discussed in the graph setting by \cite{decelleAsymptoticAnalysisStochastic2011}. 
The second has been argued heuristically in several papers \cite{decelleAsymptoticAnalysisStochastic2011,krzakalaSpectralRedemptionClustering2013,angeliniSpectralDetectionSparse2015}.\footnote{\revision{Rigorous guarantees for the correctness of similar algorithms in certain sequences of graphs converging to infinite random trees are also available \cite{demboIsingModelsLocally2010,demboGibbsMeasuresPhase2010,salez2011some}.}
    }
We provide below a heuristic statement (\Cref{thm:BP}) that expresses the thrust of this argument. 
\rev{We offer a rigorous statement with proof} in \Cref{thm:edge-reduced-dynamics} (\Cref{sec:BP-proof}). 

\subsection{The BP Jacobian} \label{subsec:BP-jacobian}

For each $k$, define the matrix $\mG_k \in \maps{\alphabet}{}$ with entries 
\begin{equation}
    g_{k;s,t} \eqdef \q{s}\paren*{\frac{\ck{k}{s,t}}{(k-1)\ck{k}{s}} - 1}\;. \label{eq:G}
\end{equation}
Let $\mF$ be the function with components
\begin{equation*}
    \mF(\vmess)_{iR}^{(s)} \eqdef \frac{1}{Z_{iR}}\q{s} \prod_{Q \in \tuples(i)\setminus R} \sum_{\labelvec_Q: z_i = s} \eta(a_Q|\labelvec_Q) \prod_{j \in Q\setminus i} \mess{j}{Q}{z_j}\;.
\end{equation*}
This expression is obtained by substituting \cref{eq:BP-2} into \cref{eq:BP-1}.  
We can then write the dynamics, restricted to the variable $\vmess$, as $\vmess \gets \mF(\vmess)$.  

We are now prepared to relate the belief-propagation algorithm to the nonbacktracking matrices $\bB_k$. 
A precise statement and proof of the relationship requires large amounts of additional notation so we defer them to \Cref{sec:BP-proof}. 
 Here, we state a heuristic version of the result. 

\begin{claim} \label{thm:BP}
    Let $\HG$ be sampled from a sparse Bernoulli stochastic blockmodel $\eta$ in which $\ck{k}{s} = \ck{k}{t}$ for all $s, t \in \alphabet$ and $k \in \Kset$. 
    Then, with high probability as $n\rightarrow \infty$:
    \begin{itemize}
        \item The point $\bar{\vmess}$ with coordinates $\barmess{i}{R}{s} = q^{(s)}$ for all $i \in \nodes$,  $R \in \tuples(i)$, and $s \in \alphabet$ is approximately a fixed point of $\mF$, in the sense that $\mF(\bar{\vmess}) \doteq \bar{\vmess}$.  
        \item The Jacobian of $\mF$ at $\bar{\vmess}$, restricted to a space of appropriately normalized perturbations, has entries of order $O(n^{-1})$ except in a block corresponding to the set of realized edges. 
        The restriction of the Jacobian to this block is $\bJ + O(n^{-1})$, where 
        \begin{equation}
            \bJ \eqdef \sum_{k \in \Kset} \mG_k \otimes \bB_k \in \maps{\alphabet\times \pointed{\edges}}{}\;. \label{eq:jacobian}
        \end{equation}
    \end{itemize} 
\end{claim}

\Cref{thm:BP} \rev{states} that the uninformative point $\bar{\vmu}$ is approximately a fixed point of the belief-propagation update $\mF$. 
\rev{T}he stability of this fixed point is governed by the Jacobian matrix, which is dominated by a block which can be approximated by $\bJ$. 
Because $\bJ$ approximates the BP dynamics in a neighborhood of $\bar{\vmu}$, and because the BP dynamics by design attempt to find community structure, we expect that the eigenvectors \revision{of $\bJ$ with large real eigenvalues} to carry information about community structure in the hypergraph $\HG$. 

In general, $\bJ \in \maps{\alphabet \times \pointed{\edges}}{} \simeq \R^{\ell\pointed{m}\times \ell\pointed{m}}$ can be very large, and computing its eigenpairs can be costly.  
We therefore ask whether it is possible to compute on a smaller matrix. 
In the case of a $k$-uniform hypergraph, the relevant eigenvalues of $\bJ = \mG_k \otimes \bB_k$ are of the form $\gamma \beta$, where $\gamma$ is an eigenvalue of $\mG_k$ and $\beta$ an eigenvalue of $\bB_k$. 
This allows~\citet{angeliniSpectralDetectionSparse2015} to compute only on $\bB_k$, or indeed on the smaller $\bB_k'$. 
In nonuniform hypergraphs, however, such a simple reduction is not available. 
We therefore offer a partial generalization in \Cref{thm:reduced-jacobian} that enables us to compute on a smaller matrix on nonuniform hypergraphs. 

\revision{
Let $\vu\in \vectorspace{\alphabet \times \pointed{\edges}}$. 
Index the entries of $\vu$ as $u_{iQ}^{(s)}$, where $iQ \in \pointed{\edges}$ is a pointed edge and $s \in \alphabet$ is a group label. 
Let $\bL \in \maps{\alphabet \times \pointed{\edges}}{2 \alphabet \times \Kset \times\nodes}$ be \revision{the matrix} $\bL:\vu \mapsto \vx = (\vx_1, \vx_2)^T \in \vectorspace{2\alphabet \times \Kset \times \nodes}$, where 
\begin{equation}
    x_{1;i,k}^{(s)} \eqdef \revision{\sum_{Q \in \partial_ki} \sum_{j \in Q\setminus i} u_{jQ}^{(s)}} \quad \text{and} \quad 
    x_{2;i,k}^{(s)}  \eqdef \revision{\sum_{Q \in \partial_ki} u_{iQ}^{(s)}}\;.  \label{eq:alpha-beta-def}
\end{equation}
\begin{theorem}\label{thm:reduced-jacobian}
    There exists a matrix $\mJ' \in \maps{\alphabet \times \Kset \times \nodes}{}$ which can be expressed in terms of the parameter matrices $\curlybrace*{\mG_k}_{k \in \Kset}$ and the adjacency matrices $\curlybrace*{\mA_{k}}_{k \in \Kset}$ such that, if $\lambda \vu = \mJ \vu$, then  $\lambda \vx = \mJ' \vx$. 
    In particular, either $\vx = \vzero$ or $\vx$ is an eigenvector of $\mJ'$ with eigenvalue $\lambda$. 
\end{theorem}
Writing $\mJ'$ and proving \Cref{thm:reduced-jacobian} explicitly requires the introduction of some cumbersome notation. 
We defer the detailed statement and proof to \Cref{thm:jacobian-reduction} in \Cref{sec:reduced-jacobian-proof}. 
} 

\revision{
    We focus on the eigenvectors $\vu$ such that $\mL \vu \neq \vzero$ for clustering algorithms.
    In analogy to the use of $\vx_2$ in \oursimplealgorithm{}, we again use $\vx_2$ here. 
    One intuition for this choice is that  $x_{2;ik}^{(s)}$ sums beliefs that node $i$ belongs to cluster $s$ across edges of size $k$. 
    The vector $\vx_2$ can be obtained either by computing eigenvectors of $\mJ$ and applying the transformation $\mL$ or directly by computing eigenvectors of $\mJ'$.  
}

\subsection{Alternating Belief-Propagation Spectral Clustering}

    \Cref{thm:BP,thm:reduced-jacobian} jointly suggest a modified algorithm based on the eigenvectors of $\bJ$ or $\bJ'$ rather than the eigenvectors of $\bB$ or $\bB'$. 
    Since $\bJ$ depends on the blockmodel parameters through the matrices $\{\mG_k\}$, we alternate between spectral clustering steps and updates to these parameters. 
    This alternating structure is reminiscent of expectation-maximization~\cite{dempster1977maximum} and other coordinate-ascent algorithms. 
    However, our alternating algorithm is not literally a form of coordinate-ascent because the spectral clustering step does not maximize a likelihood objective. 
    
    Indeed, one can carry out \ouralgorithm{} without even specifying a likelihood objective. 
    While the absence of a likelihood makes certain tasks harder, there is also an important computational benefit. 
    In general, fully specifying a stochastic blockmodel requires specifying $\eta(R\in \edges |\labelvec_R)$ for every possible combination of labels $\labelvec_R$. 
    In a hypergraph with edges up to size $k$ and $\ell$ group labels, there are $\frac{\ell^k}{k!}$ such combinations. 
    Calculating a likelihood under the blockmodel therefore requires the estimation of a potentially very large number of parameters. 

    In contrast, \ouralgorithm{} does not require specification of $\eta(R\in \edges|\labelvec_R)$, but only the entries of $\mG_k$ for each $k$. 
    To do this, we need to estimate the label proportions $\q{s}$ and the pairwise edge counts $m_k(s,t)$ for each $k$, $s$, and $t$. 
    This is a total of $\ell + \frac{1}{2}k\ell^2$ parameters, a number which scales much more favorably than \rev{the} $\frac{\ell^k}{k!}$ \rev{parameters required for likelihood maximization}. 
    As a result, \ouralg{} is both less sensitive to the fine details of the HSBM parameters and more robust against overfitting concerns. 
    \rev{We show how to estimate the necessary parameters given an estimated label vector in \Cref{sec:estimation-of-G}.}

    The eigenvectors of $\bJ'$ are elements of $\vectorspace{2\alphabet\times \Kset \times \nodes}$. 
    In order to carry out the clustering step, we need to obtain from these eigenvectors a set of feature vectors in $\vectorspace{\nodes}$ that carry information on the level of nodes. 
    To do so, we use \cref{eq:marginal-message}, which suggests that the linearized perturbations in the marginal label distributions can be obtained by summing over all messages incoming to node $i$. 
    From the eigenvector $\revision{\vx}_j$ we form the matrix $\tilde{\revision{\mY}}_j \in \maps{\nodes}{\alphabet}$ with entries $\tilde{\revision{y}}_{j;i,s} = \mathrm{sign}\paren{\sum_{k\in K}\revision{x}_{\revision{2};ik}^{(s)}}$, with $x_{\revision{2};ik}^{(s)}$ \revision{as defined in} \cref{eq:alpha-beta-def}.
    The extraction of the sign is again useful for noise reduction in small instances.\footnote{\revision{A small experiment supporting the extraction of the sign of the expression $\sum_{k\in K}\revision{x}_{1;ik}^{(s)}$ is shown in \Cref{fig:threshold-experiment}}.} 
    The columns of the matrices $\{\revision{\tilde{\mY}_j}\}$ form the embedding coordinates on which we cluster. 
    This clustering step gives a new label vector $\labelvec$. 
    We use $\labelvec$ to compute updated estimates of the model parameters, and repeat. 
    A single stage of the clustering step is illustrated in \Cref{fig:algorithm}.

    \begin{algorithm} 
        \caption{Step of Alternating BP Hypergraph Spectral Clustering (\ouralg{})}
        \label{alg:BP-spectral-clustering}
        \begin{algorithmic}[1]
        \REQUIRE Hypergraph $\HG$, current clustering $\labelvec_0$ \revision{with $\ell$ groups}
        \STATE $\{\mG_k\} \gets \mathrm{estimateParameters}(\HG, \labelvec_0)$ 
        \STATE Form $\bJ'$ \revision{according to \Cref{thm:reduced-jacobian,thm:jacobian-reduction}}. 
        \STATE $\{\vx_\ell\}_{\ell = 1}^h \gets$ \text{\revision{$h$ eigenvectors of $\bJ'$ with real eigenvalues larger than 1 in magnitude.}}
        \STATE Initialize $\tilde{\mY}$. 
        \FOR{$j = 1,\ldots,h\;, s = 1,\ldots,\ell\;, i = 1,\ldots,n$}
            \STATE $\tilde{y}_{j;i,s} \gets \mathrm{sign}\paren{\sum_{k\in \Kset}x_{\revision{2};ik}^{(s)}}$
        \ENDFOR
        \STATE $\labelvec = $ \text{Cluster}$(\tilde{\mY})$
        \ENSURE $\labelvec$. 
        \end{algorithmic}
    \end{algorithm}

    \revision{
        It is also possible to carry out \ouralg{} using the full Jacobian matrix $\mJ$ rather than the reduced matrix $\mJ'$. 
        In this case, one computes eigenvectors $\{\vu_{\ell}\}_{\ell=1}^{h}$ of $\mJ$ and forms from them the vectors $\{\vx_\ell\}_{\ell=1}^{h}$ via \eqref{eq:alpha-beta-def}. 
        One then proceeds with the remainder of the algorithm. 
        Whether it is best to use $\mJ$ or $\mJ'$ depends on the size of the hypergraph under consideration. 
        The cost of \ouralg{} is dominated by (a) the cost of allocating the matrix used and (b) the cost of computing eigenpairs. 
        Experimentally, we have found that the allocation cost of $\mJ'$ can exceed that of $\mJ$ on small instances.
        When performing many small experiments, such as in  \Cref{fig:heatmaps}, for example, we therefore use $\mJ$ directly. 
        For larger instances, we found faster overall performance using $\mJ'$. 
    }
    
    \subsection{Number of Eigenvectors}
    
    \begin{figure}
        \centering 
        \includegraphics[width=\textwidth]{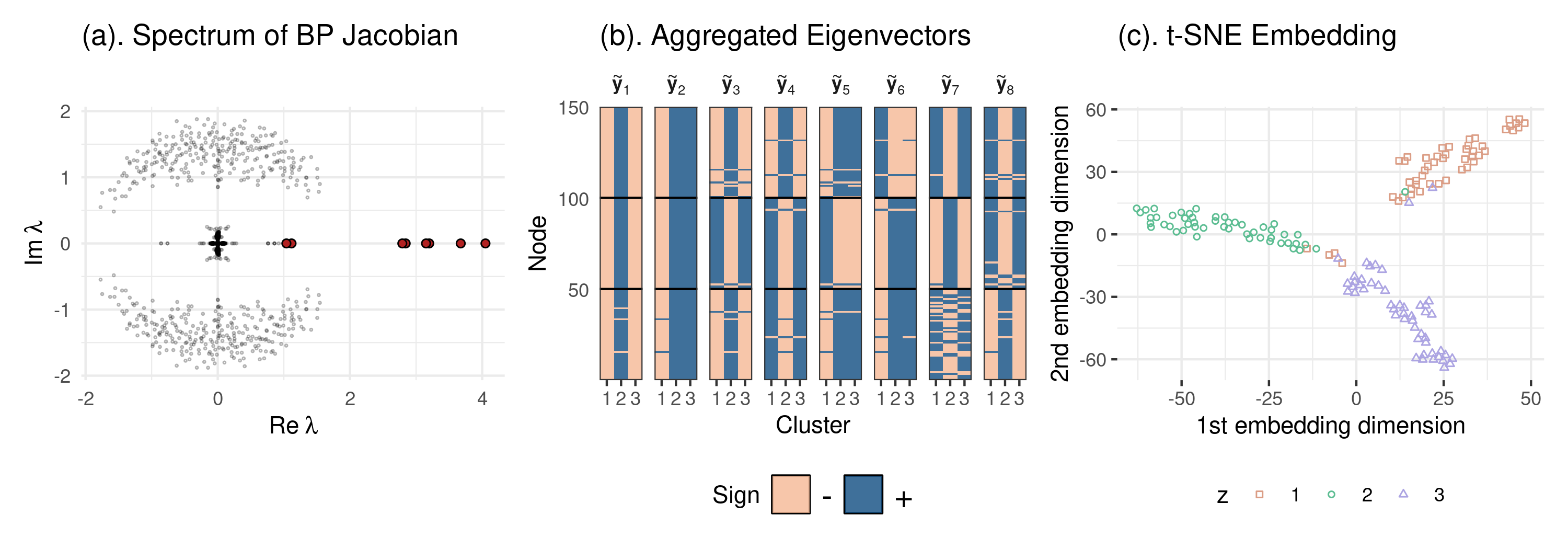}
        \caption{
            Illustration of the stages of \ouralg{} (\Cref{alg:BP-spectral-clustering}). 
            In this example, we generated a synthetic hypergraph with three clusters of 50 nodes each, with $\ck{2}{} = \ck{3}{} = 5$. 
            90\% of 2-edges and 10\% of 3-edges are within-cluster. 
            \revision{We have intentionally used a small hypergraph in order to promote the legibility of the figure.}
            (a): 
            First, we compute the spectrum of the matrix $\bJ'$ described in \Cref{thm:reduced-jacobian}.  
            There are eight real eigenvalues larger than unity in magnitude. 
            (b): 
            Next, we form \rev{features $\tilde{\vy}_h$ from the eigenvectors of $\bJ$}, which we visualize here as matrices. 
            Shown are the signs of the entries of aggregated eigenvectors for highlighted real eigenvalues in (a). 
            The entries are $\tilde{\revision{y}}_{j;i,s}$ \revision{as calculated in \Cref{alg:BP-spectral-clustering}}. 
            Thick horizontal lines separate nodes in different ground-truth clusters. 
            (c): 
            Finally, we study the nodes in the space spanned by the columns $\revision{\tilde{\vy}}_{j;\cdot,s}$ for each aggregate eigenvector $\revision{\tilde{\vy}}_j$. 
            In our example, this space has $8 \times 3 = 24$ dimensions. 
            In a full run of \Cref{alg:BP-spectral-clustering}, we would then perform a clustering algorithm such as $k$-means to obtain labels. 
            Here, we instead visualize the true clusters in two dimensions using t-SNE. 
            \rev{Nodes in the same cluster appear close in the embedding space.}
        }
        \label{fig:algorithm}
    \end{figure}

    \revision{
        In \Cref{alg:BP-spectral-clustering}, we extract from $\mJ'$ (or, alternatively, $\mJ$) the $h$ eigenvectors with real eigenvectors larger than unity in magnitude. 
        This is very different from the standard approach in the case of sparse uniform hypergraphs and graphs. 
        There, in a blockmodel with $\ell$ detectable communities, there are $\ell$ real eigenvalues outside the bulk of the spectrum with high probability. 
        The real \revision{eigenvalue of largest magnitude} is uninformative, and the remaining $\ell-1$ real eigenvalues have eigenvectors from which the $\ell$ communities can be detected.
        In moderately nonsparse regimes, there are also real eigenvalues with community-informative eigenvectors inside the bulk of the spectrum \cite{costeEigenvaluesNonbacktrackingOperator2021}; these eigenvalues can also be used in algorithms \cite{JMLR:v22:20-261}. 
    }\revision{
        In the nonuniform case, the need to consider the Jacobian $\mJ$ (or $\mJ'$) rather than than the nonbacktracking matrix $\mB$ complicates the situation considerably. 
        The Jacobian $\mJ$ is a sum of terms  \eqref{eq:jacobian}, each of which contains a different size-specific nonbacktracking operator $\mB_k$. 
        The \revision{informative} eigenvectors of an individual term are, unfortunately, not guaranteed to be even correlated with \revision{informative} eigenvectors of the complete sum $\mJ$. 
        As a result, simply counting eigenvalues outside the bulk of the spectrum of $\mJ$ is not a reliable guide to the total number of detectable communities. 
        Relating the number of detectable communities to the number of eigenvectors to extract from $\mJ$ is a nontrivial avenue of future work. 
    }\revision{
        Instead, we justify the choice of eigenvectors in \Cref{alg:BP-spectral-clustering} from a dynamical perspective. 
        Each eigenvalue of $\mJ$ larger than unity in magnitude corresponds to an unstable direction. 
        Since messages correspond to probability distributions, only real perturbations to messages are interpretable. 
        A reasonable heuristic choice is therefore to extract the eigenvectors of $\mJ$ or $\mJ'$ whose eigenvalues are both real and larger than unity in magnitude. 
        Unlike in the case of uniform hypergraphs and graphs, informative real eigenvalues can lie \emph{inside} the bulk of spectrum of $\mJ$, as illustrated in \Cref{fig:algorithm}(a). 
        The development of automated guidance for the location of informative eigenpairs in this setting in terms of the desired number of clusters and current parameter estimate would be a useful avenue of future study. 
        The need to compute more eigenvectors on the relatively large matrices $\mJ$ or $\mJ'$ is an important practical limitation of our proposed algorithms. 
    }
    \subsection{Number of Clusters}
    \revision{
        It is also necessary to specify the number of clusters $\ell$ in \cref{alg:BP-spectral-clustering}, as this number is not determined by the number of informative eigenvectors.     
        Many approaches to this problem take a model-selection perspective, choosing a number of clusters to optimize an information criterion or summarize a Bayesian posterior. 
        Unfortunately, such approaches are not available here because \ouralg{} does not optimize a likelihood or other objective function. 
    }

        In order to select the number of clusters and determine which clustering to finally accept, we therefore use a surrogate objective function. 
        We use $k$-means for the $\mathrm{Cluster}$ step, and we use as an objective function the proportion of variance in the embedding space explained by the returned clusters.  
        This enables direct comparisons between candidate clusterings with the same numbers of cluster labels, while scree plots can assist choices about the correct number of cluster labels to use. 
        

\section{Conjectured Thresholds for \ouralg{}} \label{sec:thresholds}

    We now consider the performance of our spectral clustering algorithms \oursimplealg{} and \ouralg{} on sparse synthetic data generated by a simple hypergraph stochastic blockmodel. 
    Our development is motivated in part by known behavior of spectral methods in sparse random graphs~\cite{nadakuditiGraphSpectraDetectability2012,krzakalaSpectralRedemptionClustering2013}. 
    We first restate a standard definition of \emph{detection} in clustering problems.  
    For each $n$, let $\eta_n(\rev{\vtheta}, \labelvec_n)$ be a probability distribution over graphs on $n$ nodes\rev{, parameterized by some vector $\rev{\vtheta}$}. 
    Each such distribution possesses the same shared set of parameters $\rev{\vtheta}$, as well as a planted partition $\labelvec_n$ of nodes. 
    Let $\G_n \sim \eta_n(\rev{\vtheta}, \labelvec_n)$. 
    Let $\mathtt{A}$ be a clustering algorithm, which we view as a map $\G_n \mapsto \hat{\labelvec}_n$ from the data to an estimate of the true labels. 
    Let $\rho$ be a correlation function measuring label agreement with the property that if $\labelvec_n$ and $\hat{\labelvec}_n$ are length-$n$ labels sampled independently, then $\rho(\labelvec, \hat{\labelvec})\rightarrow 0$ with high probability as $n\rightarrow \infty$. 
    Examples of measures include mutual information, \revision{the} Adjusted Rand \revision{Index}, and suitably adjusted versions of the overlap~\cite{decelleAsymptoticAnalysisStochastic2011,abbeCommunityDetectionStochastic2017}. 

    \begin{definition}\label{def:detectability}
        Algorithm $\mathtt{A}$ \emph{detects communities} in the sequence $\{\eta_n, \rev{\vtheta}, \labelvec_n\}$ with respect to the correlation function $\rho$ if there exists $\epsilon > 0$ such that, with high probability as $n$ grows large, $\rho(\labelvec_n, \mathtt{A}(\G_n)) > \epsilon$.
    \end{definition}
    
    \revision{
        As briefly described in \Cref{sec:intro}, a notable setting is the graph stochastic blockmodel with equal-sized communities in which the parameter $\rev{\vtheta} = (a,b)$ governs the rates of edges within and between communities. 
        In this case, the size of the signal-to-noise ratio $\phi$ from \eqref{eq:graph-detectability} determines whether any algorithm exists that detects communities in $\{\eta_n, \rev{\vtheta}, \vz_n\}$ as $n\rightarrow \infty$.  
    }
    
    \Cref{def:detectability} generalizes to the setting of hypergraphs with planted cluster structure. 
    It is necessary only to allow $\eta_n(\rev{\vtheta}, \labelvec_n)$ to be a probability distribution over hypergraphs. 
    Indeed,~\citet{angeliniSpectralDetectionSparse2015} have offered conjectures concerning the ability of nonbacktracking spectral methods to detect planted clusters in the setting of uniform hypergraphs. 
    Several of these conjectures were recently proven by \citet{stephanSparseRandomHypergraphs2022}. 
    Our purpose in this section is to extend these conjectures to the setting of nonuniform hypergraphs. 
    We will offer experimental support of these conjectures and leave their proofs to future work. 

    \revision{
        In \ouralg{}, it is necessary to form an estimate of the parameter matrices $\curlybrace*{\mG_k}_{k\in \Kset}$. 
        In \Cref{alg:BP-spectral-clustering}, we form new estimates of these matrices in each iteration. 
        While realistic, the need for re-estimation considerably complicates detectability analysis. 
        We therefore consider an idealized setting in which the correct values of the parameter matrices $\curlybrace*{\mG_k}_{k\in \Kset}$ are known \emph{a priori}. 
        The assumption that the parameters are known exactly is sometimes called the \emph{Nishimori condition} in statistical physics~\cite{decelleAsymptoticAnalysisStochastic2011}.\footnote{In graph SBMs, the relaxation of the Nishimori condition leads to a much more complicated analysis with qualitatively different conclusions~\cite{kawamotoAlgorithmicDetectabilityThreshold2018}.}
    }

        For each $k$, \rev{direct calculation shows that} the vector $\mathbf{1} = (1, 1)$ is an eigenvector of $\mG_k$ with eigenvalue $\frac{\beta_k}{\alpha_k}$. 
        \revision{
            Let $\vu\in \vectorspace{\pointed{\edges}}$ be as in \Cref{thm:eigen-expectation}. 
        }
        Let $\hat{\vu} = \mathbf{1}\otimes \vu$ and $\hat{\vv} = \mathbf{1}\otimes \vv$. 
        Let $\lambda = \sum_{k \in \Kset} \rev{\frac{\beta_k^2}{\alpha_k}}$. 
        \Cref{eq:jacobian,thm:eigen-expectation} now imply the following result: 
        \begin{cor} \label{cor:jacobian-eig}
            Under the two-group blockmodel $\eta$, for $\ell = 1,2$, we have \revision{
            \begin{align}
                \E\sqbracket*{(\mJ\hat{\vu})_{iQ}^{(\ell)} - \beta \hat{u}_{iQ}^{(\ell)}|Q\in \edges} &\doteq 0 \\  
                \E\sqbracket*{(\mJ\hat{\vv})_{iQ}^{(\ell)} - \lambda \hat{v}_{iQ}^{(\ell)}|Q\in \edges} &\doteq 0\;. 
            \end{align}}
        \end{cor} 
        \revision{
            Importantly, it is not guaranteed that $\lambda \leq \beta$. 
            Unlike in the case of uniform hypergraphs, therefore, it may be the real eigenvalue of largest magnitude that carries community information. 
            \Cref{fig:eigenvalue-locations} in \Cref{sec:additional-experiments} computes the two real eigenvalues of $\mJ$ of largest magnitude in a nonuniform hypergraph, finding excellent agreement with \Cref{cor:jacobian-eig} and illustrating this phenomenon. 
        }
        Motivated by these results, we pose the following conjecture: 
        \begin{conj} \label{conj:jacobian-threshold}
            In the same setting as \Cref{thm:eigen-expectation}, with high probability as $n$ grows large, 
            \begin{itemize}
                \item $\mJ$ possesses a real, community-correlated eigenvector with eigenvalue $\lambda + o(1)$. 
                \item \ouralg{}, when initialized with knowledge of the true parameters $\ckin$ and $\ckout$, is able to detect ground-truth clusters if $\abs{\lambda} > 1$.   
            \end{itemize}
        \end{conj}
        
        \revision{
            We now verify that \ouralg{}, initialized with knowledge of $\ckin$ and $\ckout$, is able to detect clusters in broader parameter regimes than \oursimplealg{}. 
            After algebraic rearrangement, we can write the respective detectability conditions from \Cref{conj:vanilla,conj:jacobian-threshold} as 
            \begin{align}
                \phi_1 &\eqdef \frac{\paren*{\sum_{k \in K}\beta_k}^2}{\sum_{k \in K} \alpha_{k }} > 1 \tag{\oursimplealg{}} \\ 
                \phi_2 &\eqdef \lambda = \sum_{k \in K}\frac{\beta_k^2}{\alpha_k} > 1\;. \tag{\ouralg{}}
            \end{align}
            Noting that $\alpha_k \ge 0$ and $\beta_k = 0$ if $\alpha_k = 0$, an application of Bergstr\"{o}m's inequality~\cite{bergstrom1949triangle, beckenbach2012inequalities} yields $\phi_2 \ge \phi_1$.
            Since $\phi_2 \geq \phi_1$, we conclude that, under \Cref{conj:vanilla,conj:jacobian-threshold}, \ouralg{} with true parameters succeeds in detecting communities in all cases for which \oursimplealg{} succeeds. 
            We also expect to find cases in which $\phi_2 > 1 \geq \phi_1$. 
            In such cases cases, our conjectures imply that \ouralg{} succeeds when \oursimplealg{} fails. 
            We support this expectation experimentally in \Cref{fig:heatmaps} below. 
            }

    \subsection{Parameterized Thresholds} \label{sec:parameterized-thresholds}

        We now illustrate \Cref{conj:vanilla,conj:jacobian-threshold} \rev{computationally}. 
        To do so, we first need to express conditions on the eigenvalues in terms of the parameters of an HSBM. 
        We specify the affinity function implicitly via a ball-dropping process \cite{ramani2019coin}. 
        We first generate a number of $k$-edges. 
        With probability $p_k$, a $k$-edge is sampled uniformly at random from the set of all node $k$-subsets $R$ such that all nodes in $R$ have the same cluster label. 
        With probability $1-p_k$, the $k$-edge is sampled uniformly at random from the set of all $k$-subsets in which at least two nodes have differing labels. 
        We refer to the former type of edge as \emph{within-cluster} and the latter as \emph{between-cluster}. 
        
        \revision{
            We first consider \oursimplealg{}. 
            In this HSBM, \cref{eq:detectability} in \Cref{conj:vanilla} defines a pair of hyperplanes in the coordinates $\{p_k\}_{k \in \Kset}$.
        }
        This follows from two facts. 
        First, by construction, $\alpha$ does not depend on $\{p_k\}$. 
        Second, $\beta$ is an affine function of  $\{p_k\}$:

        \begin{lemma} \label{lm:vanilla-parameterized-eigs}
            In this model, with $r_k \eqdef \frac{1 - 2^{2-k}}{2 - 2^{2-k}}$, we have 
            \begin{align}
                \alpha &=  \sum_{k \in \Kset} (k-1)c_k\\
                \beta  &= \sum_{k \in \Kset}(k-1)c_k\sqbracket*{2(1-r_k)p_k + 2r_k - 1}\;. \label{eq:vanilla-eig-parameterized}
            \end{align}
        \end{lemma}
        The proof is a direct calculation and provided in \Cref{sec:proof-of-vanilla-parameterized-eigs}. 
        \Cref{lm:vanilla-parameterized-eigs} in conjunction with \Cref{conj:vanilla} define a pair of hyperplanes in the coordinates $\{p_k\}$. 
        The region between these hyperplanes is, under these conjectures, the region in which \oursimplealg{} fails to detect clusters. 
        
        \Cref{fig:heatmaps}(a). shows an experiment on a blockmodel on \rev{$400$} nodes with $2$- and $3$-edges with varying $p_2$ and $p_3$. 
        In this panel, we run \oursimplealg{} 20 times for each parameter value, and compute the average Adjusted Rand Index (ARI) of the retrieved clustering against the planted clustering. 
        White lines are the boundaries given by \rev{$\beta^2 = \alpha$, with $\alpha$ and $\beta$ given by \cref{lm:vanilla-parameterized-eigs}}. 
        Under \Cref{conj:vanilla}, as $n\rightarrow \infty$, spectral clustering returns \rev{labels} with ARI bounded above $0$ iff $(p_2, p_3)$ does not lie between the two white lines. 
        The experimental results shown are consistent with this conjecture. 
           
    \begin{figure}[h]
        \centering
        \includegraphics[width=0.98\textwidth]{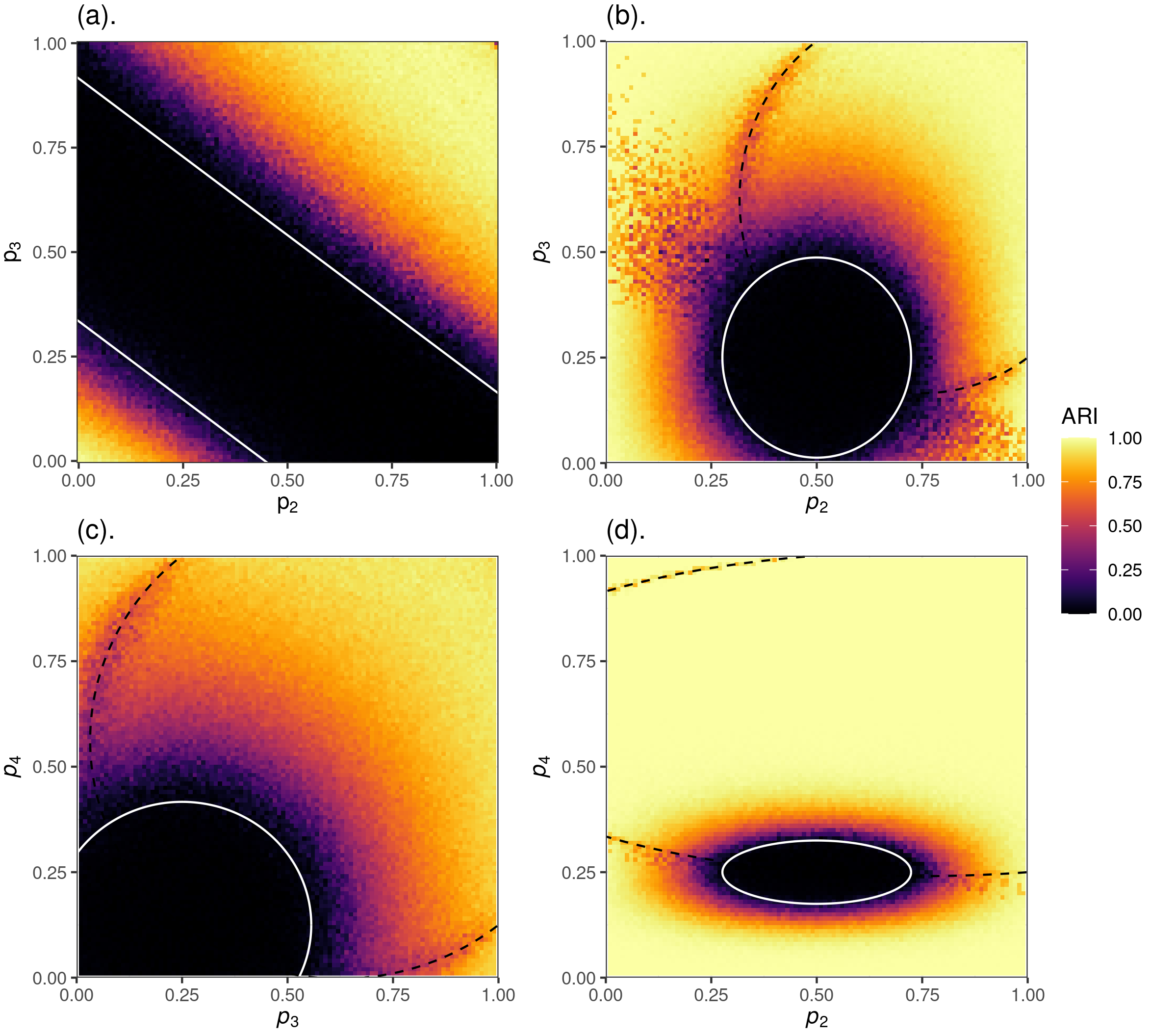}
        \caption{Experiments and analytical boundaries for nonbacktracking spectral clustering in synthetic hypergraphs of \revision{400} nodes.
        Each pixel is the mean Adjusted Rand Index (ARI) of the retrieved clustering against ground truth across 20 trials. 
        (a). Spectral clustering via the nonbacktracking operator $\bB$ according to \oursimplealg{} (\Cref{alg:vanilla}), $c_2 = c_3 = 5$. 
        The white lines are given by \rev{the equation $\beta^2 = \alpha$, with $\alpha$ and $\beta$ as given by \cref{lm:vanilla-parameterized-eigs}}. 
        (b)-(d). Spectral clustering in one round of \ouralg{} (\Cref{alg:BP-spectral-clustering}) using the true values of the parameter matrices $\curlybrace*{\mG_k}_{k\in \Kset}$.   
        In each panel, the white ellipse is given by \rev{\eqref{eq:ellipse}, which is in turn a consequence of the equation $\lambda = 1$}.   
        The dashed black curves trace the ellipse described by the collision of the community-correlated eigenvalue $\lambda$ with the uninformative eigenvalue $\beta$. 
        We show hypergraphs with edges of size $2$, $3$, and $4$. 
        The mean $k$-degrees for each edge-size $k$ vary in each panel.  
        (b). $c_2 = c_3 = 5$, $c_4 = 0$. 
        (c). $c_3 = c_4 = 5$, $c_2 = 0$. 
        (d). $c_2 = 5$, $c_4 = 50$, $c_3 = 0$. 
        }\label{fig:heatmaps}
    \end{figure}

        \rev{We now consider detectability thresholds for \ouralgorithm{} under the Nishimori condition.}
        We can compute the entries of the matrix $\mG_k$ in terms of $\ckin$ using \cref{eq:G}, obtaining
        \begin{equation*}
            \gkin \eqdef \gk{k}{s,s} = \frac{1}{2} \paren*{\frac{\ckin}{(k-1)\ck{k}{}} - 1} \quad \text{and} \quad \gkout \eqdef \gk{k}{s,t} = -\gkin\;
        \end{equation*} 
        The eigenvalue of $\mG_k$ that appears in \Cref{cor:jacobian-eig} is then 
        \begin{equation}
            \rev{\frac{\beta_k}{\alpha_k}} = \gkin - \gkout = 2\gkin = \frac{\ckin}{(k-1)\ck{k}{}} - 1  \;.  
        \end{equation}
        Direct computation now shows that the condition $\lambda = \sum_{k \in \Kset} \rev{\frac{\beta_k^2}{\alpha_k}} = 1$ defines an axis-aligned ellipsoid with coordinates $(p_{k_1}, p_{k_2},\ldots)$, centroid $(x_{k_1},x_{k_2},\ldots)$ and radii $(a_{k_1}, a_{k_2},\ldots)$, where 
        \begin{equation} 
            x_k = \frac{1 - 2r_k}{2 - 2r_k}  \quad \text{and}\quad  a_k = \frac{\sqrt{(k-1)\ck{k}{}}}{2-2r_k} \quad \text{with} \quad r_k \eqdef \frac{1 - 2^{2-k}}{2-2^{2-k}}\;.\label{eq:ellipse}
        \end{equation}
        \Cref{conj:jacobian-threshold} claims that \ouralg{} succeeds outside this ellipse ($\lambda > 1$) and fails inside ($\lambda < 1$). 

    Panels (b)-(d).\ of \Cref{fig:heatmaps} show a sequence of cluster recovery experiments. 
    As before, each pixel gives the average Adjusted Rand Index of the recovered cluster against ground truth across 20 runs of \ouralg{} in a hypergraph stochastic blockmodel on \revision{400} nodes, with varying parameters $\{p_k\}$. 
    In each experiment, the ellipse defined by \cref{eq:ellipse} is shown in white. 
    An implication of our conjectures is that, as $n\rightarrow \infty$, with high probability, belief-propagation spectral clustering returns a clustering with ARI approaching 0 iff the point $\{p_k\}$ lies in the interior of the ellipse. 
    Careful examination shows that the algorithm occasionally succeeds within the ellipse, and occasionally fails outside it. 
    We attribute these deviations from conjectured theory to finite-size effects. 
    Recalling that we have treated the true matrices $\curlybrace*{\mG_k}_{k \in \Kset}$ as known, the observed performance and thresholds should be regarded as idealizations of the more realistic case in which these matrices must be inferred along the way. 

    From \Cref{cor:jacobian-eig} we know that there is also an approximate eigenpair \revision{$(\beta, \hat{\vu})$} which is uncorrelated with planted cluster structure. 
    The equation \revision{$\lambda = \beta$} again describes an ellipsoid in parameter space at which the two eigenvalues collide. 
    This collision induces noise in the corresponding eigenvectors, resulting in observably lower-quality cluster recovery along this ellipsoid (\Cref{fig:heatmaps}(b-d), dashed black curves). 

    \revision{
        Our conjectured thresholds are asymptotic, while the synthetic hypergraphs on which we compute in \Cref{fig:heatmaps} are, at 400 nodes, relatively small. 
        Even in this regime, the agreement with theory is quite strong. 
        In \Cref{fig:large-hypergraph-1d}, we show an experiment in a smaller region of parameter space on hypergraphs of 10,000 nodes, again finding close agreement with conjectured theory. 
    } 

\section{Experiments on Data} \label{sec:experiments}

    We first study several data sets in which ground-truth labels are available. 
    The \dataprimaryschool{}~\cite{stehleHighResolutionMeasurementsFacetoFace2011,bensonSimplicialClosureHigherorder2018} data set logs close-proximity human contact interactions detected by wearable sensors. 
    Nodes are students or teachers. 
    A hyperedge exists between a set of nodes that were jointly in proximity to each other within a short time window. 
    Each student is assigned to a unique classroom, which we use as a ground-truth label. 
    The data contain timestamps associated to each interaction, although we do not use these timestamps here. 
    There are $n = 242$ nodes and $m = 12,704$ hyperedges in \dataprimaryschool{}. 
    Hyperedges range from size $k = 2$ to size $k = 5$. 

    \begin{figure}
        \includegraphics[width=\textwidth]{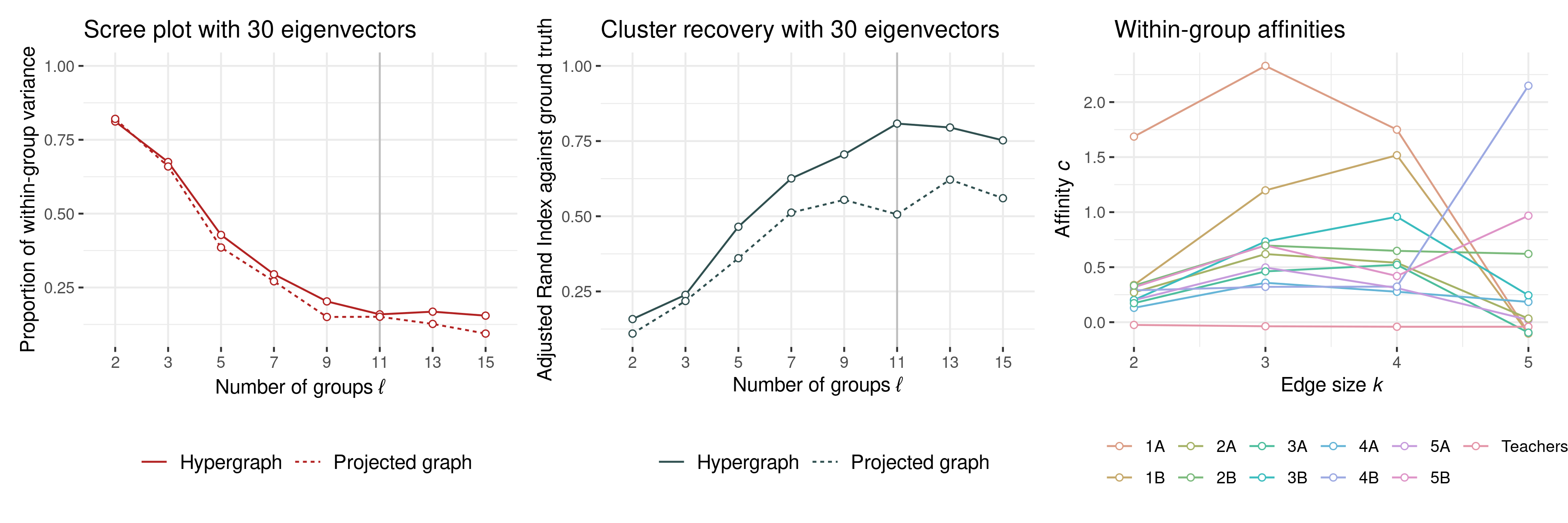}
        \caption{
            Cluster recovery in the \dataprimaryschool{} data set~\cite{stehleHighResolutionMeasurementsFacetoFace2011,bensonSimplicialClosureHigherorder2018}.
            We ran \ouralg{} on the data for 10 rounds, using the 30 eigenvectors of the belief-propagation Jacobian \revision{with largest real eigenvalues} and with a varying number of clusters to be estimated. 
            In each round, we update the estimate of the labels $\hat{\labelvec}$ by choosing the best of 20 runs of $k$-means according to the within-group sum-of-squares objective. 
            We repeat this experiment on the projected (clique-expansion) graph. 
            (Left): scree plot of the mean within-group sum-of-squares obtained by the $k$-means step as a function of the number of groups to be estimated. 
            The vertical grey line gives the true number of labels in the data. 
            (Center): Adjusted Rand Index of the clustering with lowest $k$-means objective against ground truth. 
            (Right): The diagonal entries of the matrix $\mC_k$ for varying edge size $k$. 
            A similar experiment for the \datahighschool{} data set~\cite{mastrandreaContactPatternsHigh2015,bensonSimplicialClosureHigherorder2018} is given in \Cref{fig:contact-high-school}. 
            }
        \label{fig:contact-primary-school}
    \end{figure}

    \Cref{fig:contact-primary-school} shows a suite of clustering experiments on \dataprimaryschool{}, which possesses 11 ground-truth clusters, including 10 homeroom classes and one cluster containing all teachers. 
    We ran \ouralg{} multiple times, choosing from among runs the clustering that resulted in the smallest value of the $k$-means within-group sum-of-squares objective.  
    We also varied the number of clusters $\ell$ to be learned. 
    We repeated this process for both the original hypergraph data and the projected (clique-expansion) graph obtained by replacing each $k$-hyperedge with a $k$-clique. 
    We refer to this algorithm as \graphalgorithm{} (\graphalg). 
    The lefthand plot shows that the $k$-means objective is able to give some guidance as to the appropriate number of groups to infer, with the objective function leveling off close to the true number of groups for both \graphalg{} and \ouralg{}. 
    At center, we observe that for most values of $\ell$, including all those close to the correct value, \ouralg{} considerably outperforms dyadic spectral clustering in retrieving labels correlated with ground truth. 
    One explanation for this phenomenon may be observed at right, where we plot the diagonal entries of the matrix $\mC_k$ for each edge size $k$, computed using the true cluster labels. 
    These diagonal entries measure the rate at which nodes connect to other nodes in their same group, and may therefore be interpreted as a measure of affinity or assortativity. 
    These affinities vary considerably according to edge size, suggesting that edges of differing sizes play meaningfully distinct roles in this data set. 
    \ouralg{} again outperforms spectral clustering on the projected graph, with the gap likely due to the different connection structure across edges of varying sizes. 
    While these results suggest that \ouralgorithm{} is preferable to \graphalg{} on these school contact data sets, neither algorithm achieves the perfect cluster recovery obtainable via other methods~\cite{chodrowGenerativeHypergraphClustering2021a}. 

    \begin{figure}
        \includegraphics[width=\textwidth]{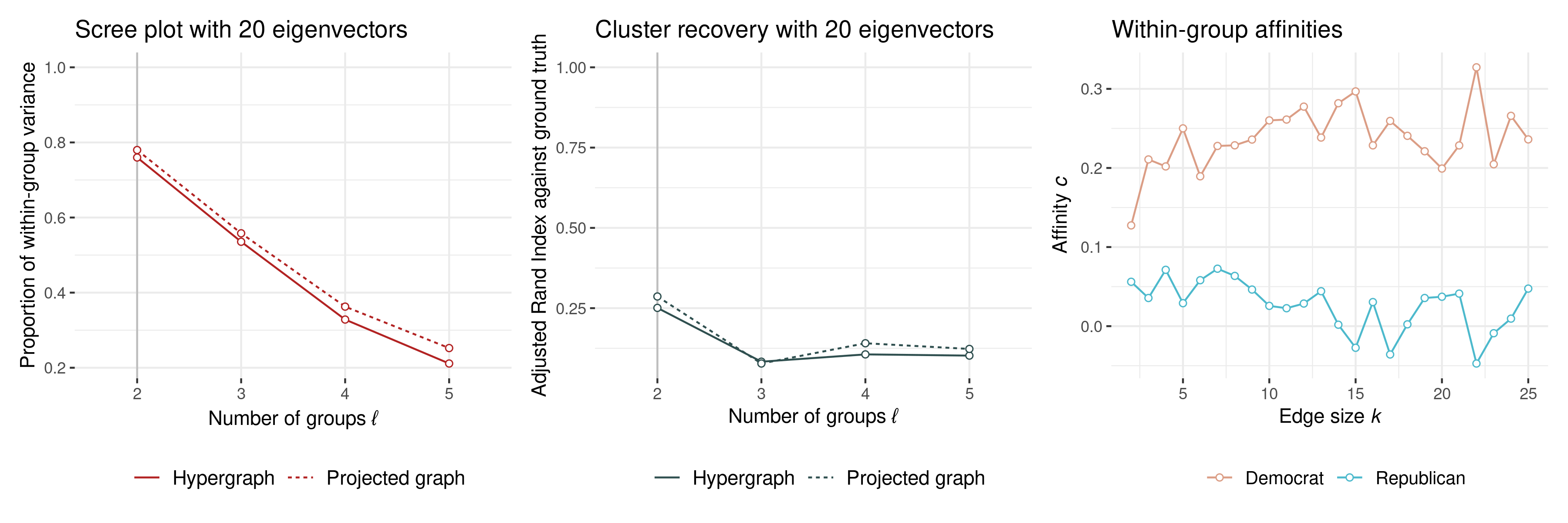}
        \caption{
            As in \Cref{fig:contact-primary-school}, using the \texttt{senate-bills} data set~\cite{fowler2006connecting,fowler2006legislative,bensonSimplicialClosureHigherorder2018}. 
            }
            \label{fig:senate-bills}
    \end{figure}

    The \datasenate{} data set \cite{fowler2006connecting,fowler2006legislative} provides a contrasting case. 
    Nodes are U.S. senators. 
    An edge exists on a set of senators for each bill that that set of senators cosponsored. 
    The data reflects bills from the 103rd through 115th U.S. Congresses, which span the years 1993–2017.
    There are 293 senators and 20,006 bills represented. 
    We consider bills cosponsored by between $k = 2$ and $k = 25$ senators, although a small number of bills sponsored by larger groups exist. 
    Labels correspond to the two major U.S. political parties. 
    Unlike in the contact networks, the scree plot entirely fails to track the number of true clusters. 
    Both \ouralgorithm{} and \graphalgorithm{} recover partitions with Adjusted Rand Index near $0.25$ when instructed to search for exactly two groups, and do much more poorly otherwise. 
    These results are qualitatively aligned with those of~\citet{chodrowGenerativeHypergraphClustering2021a}, who found that greedy modularity-based methods also struggle with the label recovery task on \datasenate{}. 
    One explanation for why \ouralgorithm{} does not achieve noticeable improvement over \graphalgorithm{} in this case is given by the plot of within-group affinities (far right). 
    Here, the variability of these affinities with edge size is much smaller than it is for \dataprimaryschool{}, suggesting that the ability of \ouralg{} to distinguish roles for different edge sizes may not confer a useful advantage over \graphalg{} on \datasenate{}.

    \begin{figure}
        \centering
        \includegraphics[width=1.0\textwidth]{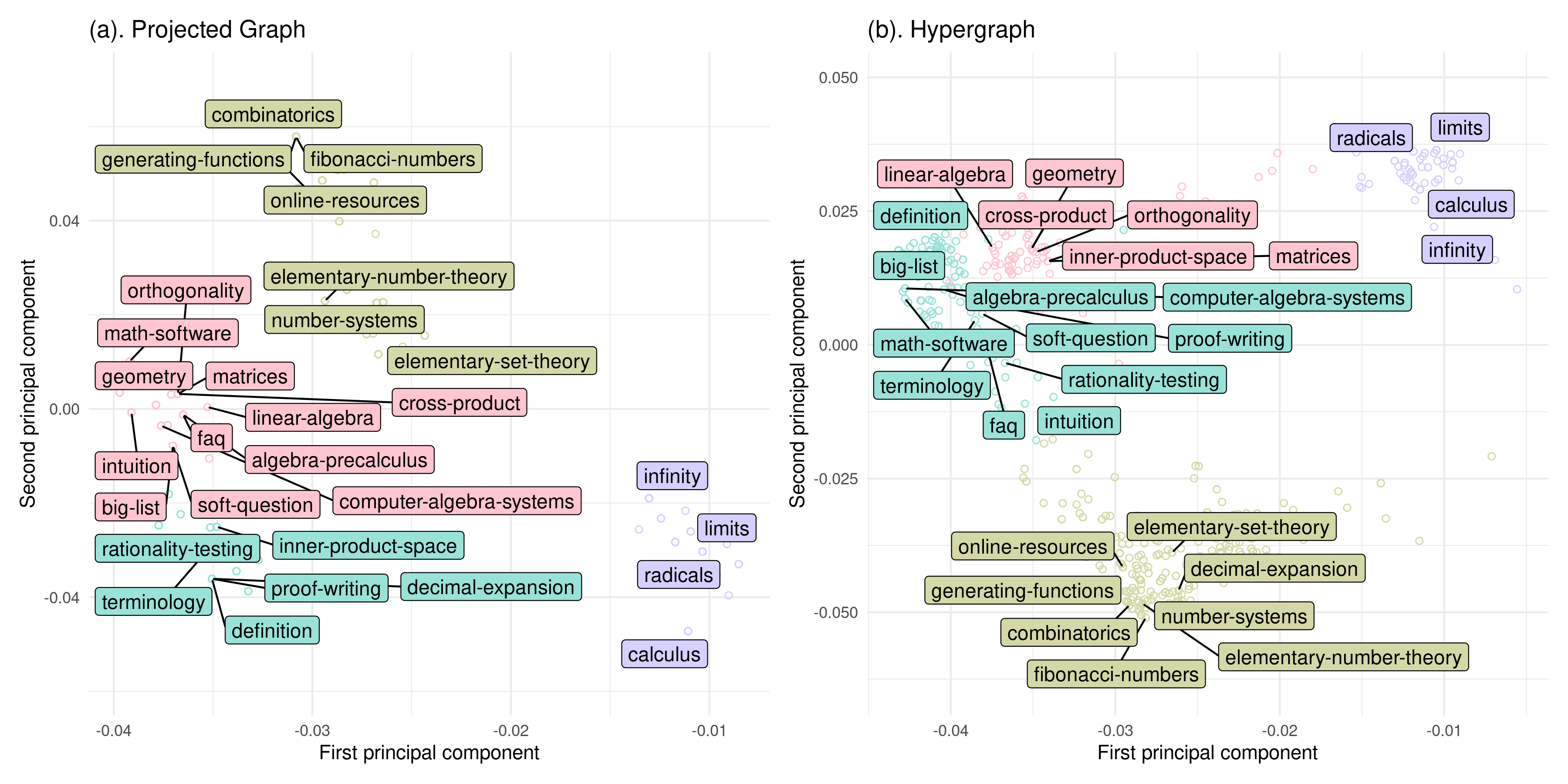}
        \caption{
            (Best viewed in color): Clustering the \texttt{tags-math-sx} data set~\cite{bensonSimplicialClosureHigherorder2018} using belief-propagation spectral clustering on the projected graph (\graphalg{}) and original hypergraph data (\ouralg{}). 
            Only tags that appeared in at least 20 questions are included. 
            In each panel, we performed 50 steps of \Cref{alg:BP-spectral-clustering} using $\ell = 4$ groups and $h = 15$ eigenvectors. 
            We repeated this process 50 times, resulting in 50 candidate clusterings. 
            For each algorithm, the clustering shown is the one which achieved the lowest $k$-means objective (total within-group sum-of-squares) of these 50 candidates. 
            Colors give the learned clustering, while coordinates in the plane give the 2-dimensional principal components projection of the eigenvector embedding from which that clustering was learned. 
            The 30 most frequently-used tags in the data set are labeled. 
        }
        \label{fig:math-sx}
    \end{figure}
    
    It is also possible to use \ouralg{} to cluster data in which no labels are natively present.  
    \Cref{fig:math-sx} shows an experiment in which we use \Cref{alg:BP-spectral-clustering} to cluster tags on the forum Math StackExchange~\cite{bensonSimplicialClosureHigherorder2018}. 
    Nodes are tags that usually express certain concepts or areas within mathematics. 
    Edges correspond to questions asked on the forum. 
    An edge exists between a set of nodes for each question asked to which the corresponding tags were applied. 
    The data contains edges of size $k = 2,\ldots,5$, with an average total degree of 720 edges per node. 
    There are a total of 1,419 nodes in the original data. 
    We removed nodes with degree less than 20, resulting in a hypergraph containing 1,038 nodes. 
    We show the result of belief-propagation spectral clustering on both the projected graph and the original hypergraph in \Cref{fig:math-sx}. 
    
    The large-scale output structure is relatively similar between the two algorithms. 
    In both cases, for example, there is a clearly separated cluster (purple) associated to calculus , including the topics \texttt{limits}, \texttt{infinity}, and \texttt{radicals}. 
    There is also a cluster  (gold) associated with topics in discrete mathematics, including \texttt{combinatorics}, \texttt{fibonacci-numbers}, and \texttt{elementary-set-theory}. 
    A third cluster (pink) focuses on linear algebra, with tags including \texttt{orthogonality}, \texttt{matrices}, and \texttt{inner-product-space}. 
    A final cluster (cyan) includes a number of tags that do not fit neatly in a single mathematical subfield, such as \texttt{definition}, \texttt{terminology}, and \texttt{proof-writing}. 

    There are, however, several notable differences in the clusterings produced. 
    The linear algebra cluster in panel (b) contains only six words each with clear connections to the field, while the cluster in panel (a) is much larger and includes several superficially unrelated tags, such as \texttt{soft-question} and \texttt{algebra-precalculus}. 
    In addition, \texttt{inner-product-space} is separated in a different cluster. 
    As another example, in panel (b) \texttt{decimal-expansion} is grouped with other topics in discrete mathematics, and is close in embedding space to \texttt{number-systems}. 
    In contrast, \texttt{decimal-expansion} is far from any related topics in panel (a). 
    
    We note that cluster assessment is a challenging task and our approach is admittedly \emph{ad hoc}. 
    That said, we do find it clear that clusters found by \oursimplealg{} are preferable to those found by \graphsimplealg{} in groupings of mathematical topic tags.


\section{Discussion} \label{sec:discussion}

    We have proposed and analyzed two spectral algorithms for clustering hypergraphs. 
    We have focused on the distinctive challenges posed by \emph{nonuniform} hypergraphs, which comprise many of the most interesting hypergraph data sets. 
    To address the challenge of representing edges of multiple sizes in compatible data structures, we have employed the hypergraph nonbacktracking matrix $\bB$ proposed by~\citet{stormZetaFunctionHypergraph2006a}. 
    We have considered both a simple spectral algorithm (\oursimplealg{}) that relies only on the eigenvectors of $\bB$, as well as a more complex spectral algorithm (\ouralg{}) based on the eigenvectors of the belief-propagation Jacobian $\bJ$ evaluated at the uninformative fixed point. 
    In each case, we have provided theorems that allow us to compute on alternative matrices that are usually smaller, thereby enabling faster computation \rev{in some cases}. 
    We have demonstrated the latter algorithm on several hypergraph data sets, showing that it enjoys superior performance over nonbacktracking methods on projected graphs due to its ability to explicitly represent distinct statistical roles for edges depending on size. 
    That said, we emphasize that \ouralgorithm{} is not state-of-the-art for label recovery in hypergraph clustering tasks as measured by accuracy and especially by scalability. 
    Even the use of \Cref{thm:reduced-jacobian} yields a matrix $\bJ'$ of size  $2n\bar{k}\ell$, which can become computationally challenging when either $k$ or $\ell$ are large. 
    \revision{In practice, the use of \ouralg{} is considerably constrained by both the memory requirements of forming $\mJ$ or $\mJ'$ and the performance of sparse eigensolvers~\cite{saad2011numerical}.}

    Part of the significance of our approach is that it admits closed-form analysis of the parameter regime in which \ouralg{} fails to recover clusters correlated with ground truth. 
    In the parameterization we use in \Cref{fig:heatmaps}, this regime is ellipsoidal in coordinates $\{p_k\}$. 
    Our description of this region relies on several conjectures related to the location of the eigenvalues of the matrices $\{\mB_k\}$ and their associated eigenvectors.  
    These conjectures are inspired by known results for the graph case~\cite{bordenaveNonbacktrackingSpectrumRandom2015}.
    Proofs of these conjectures would provide conclusive characterizations of the success and failure regimes for \ouralg{}. 
    We expect such proofs to require tools from random matrix theory similar to those used in~\citet{stephanSparseRandomHypergraphs2022}. 

    In the dyadic stochastic blockmodel, it has been proven that, for the stochastic blockmodel on graphs with two clusters, the regime in which nonbacktracking spectral clustering fails coincides precisely with the regime in which the regime in which \emph{no algorithm} can detect clusters~\cite{mossel2018proof,massoulie2014community}. 
    We pose a similar conjecture for the hypergraph blockmodel: within the ellipsoid described by \cref{eq:ellipse}, the cluster detection problem cannot be solved by any algorithm. 
    We propose a proof of this conjecture as a direction of future work. 
\section*{Acknowledgements}

Much of PSC's work on this paper was completed during his time in the Department of Mathematics of the University of California, Los Angeles. 
We are grateful to Mason Porter for donating computational resources to support this research. 
\section*{Software and Data Availability}

Code and data sufficient to reproduce experiments in this paper may be found at the GitHub repository \href{https://github.com/jamiehadd/HypergraphSpectralClustering}{jamiehadd/HypergraphSpectralClustering} 
Primary computations were performed in the Julia language using a custom-written package \cite{bezanson2017julia}. 
Visualizations were constructed using ggplot2 for the R programming language \cite{rcoreteamLanguageEnvironmentStatistical2022,wickhamGgplot2ElegantGraphics2016}. 
\section*{Gender Representation in Cited Works}

Recent work in several fields of science has identified gender bias in citation practices---papers by women and other gender-minoritized scientists are systematically under-cited in their fields \cite{dion2018gendered,caplar2017quantitative, maliniak2013gender,dworkin2020extent,teich2021citation, dworkin2020citing,wang2021gendered,llorens2021gender}.

In the spirit of \citet{zurn2020citation}, we performed an analysis of gender representation in the references cited in the main text of this manuscript. 
We manually gender-coded the first and last authors in the works cited according to personal acquaintance, instances of pronoun usage online, or first name. 
We focused on the first and last authors because typically, though not always, the former is the leading researcher and the latter the senior author in the disciplines included in our references. 
Our method of coding is limited in several ways. 
Gender is fundamentally nonbinary. 
Names and pronouns may not be indicative of gender. 
Gender may change over time. 
Manual coding is inherently \revision{flawed} and subject to error. 
Furthermore, the heuristic that the first and last authors correspond \revision{to those which} make the most important contributions to a manuscript is of varying validity in different areas of science, especially in mathematics. 

Of the works cited in the main text (excluding this statement), we estimate that \rev{18}\% had a non-male first author and \rev{19}\% had a non-male last author. Of those with at least two authors, \rev{25}\% had either a non-male first author or a non-male last author. 

\nocite{kemptonNonBacktrackingRandomWalks2016} 

\bibliographystyle{abbrvnat} 
\bibliography{zotero-refs,other-refs}

\appendix

\section{Proof of \Cref{thm:ib-nonuniform}} \label{sec:ib-nonuniform-proof}

    Our proof approach extends Kempton's proof of the Ihara-Bass formula for graphs~\cite{kemptonNonBacktrackingRandomWalks2016}. 
    The same approach was used by~\citet{stephanSparseRandomHypergraphs2022} for the case of uniform hypergraphs.

    For each $k$, define operators $\mS_k \in \maps{\pointed{\edges}}{\nodes}$, $\mT_k \in \maps{\nodes}{\pointed{\edges}}$, and $\mW_k \in \maps{\pointed{\edges}}{}$ with entries: 
    \begin{align*}
        s_{k;jR, i} &\eqdef \begin{cases} 
            1 &\quad i \in R\setminus j\;, \; \abs{R} = k \\ 
            0 &\quad \text{otherwise}
        \end{cases} \\ 
        t_{k;i, jR} &\eqdef \begin{cases} 
            1 &\quad i = j\;, \; \abs{R} = k \\ 
            0 &\quad \text{otherwise}
        \end{cases} \\ 
        w_{k;iQ, jR} &\eqdef \begin{cases} 
            1 &\quad Q = R\;,\; i \neq j\;, \; \abs{Q} = \abs{R} = k \\ 
            0 &\quad \text{otherwise}\;.
        \end{cases}
    \end{align*}
    These operators satisfy several important relations. 
    We begin with entrywise calculations: 
    \begin{align} \label{eq:entrywise-1}
        [\mT_k \mS_{k'}]_{i,j} &= \sum_{\ell R \in \pointed{\edges}} t_{k;i,\ell R}s_{k';\ell R,j}  
        = \delta_{k,k'} \abs{\curlybrace*{R \in \edges_k \;:\; i,j \in R}} 
        \eqdef \delta_{k,k'} a_{k;i,j} \\ 
        [\mT_k \mW_{\revision{k}} \mS_{k'}]_{i,j} &= \sum_{\revision{Q \in \partial i,\; R \in \partial j}} t_{k;i,iQ}w_{k;iQ, jR} s_{k';jR, j} \nonumber\\ 
        &= \delta_{k,k'}\begin{cases} 
            (k-1) d_{k;i,i} &\quad i = j \\ 
            (k-2)a_{k;i,j} &\quad \text{otherwise}
        \end{cases} \nonumber\\ 
        &= \delta_{k,k'} \sqbracket*{(k-1)\mD_k + (k-2) \mA_k}_{i,j}\;.\label{eq:entrywise-2} \\ 
        [\mS_k\mT_{k'} - \delta_{k,k'}\mW_k]_{iQ, jR} &= \sum_{\revision{h}\in \nodes} s_{k;iQ, \revision{h}}t_{k';\revision{h}, jR}  - \delta_{k,k'}w_{k;iQ, jR} \nonumber\\ 
        &= \begin{cases}
            1 &\quad iQ \rightarrow jR\;,\; \abs{iQ} = k\;, \; \abs{jR} = k' \\ 
            0 &\quad \text{otherwise} 
        \end{cases} \nonumber\\ 
        &= b_{k \rightarrow k';iQ, jR}
        \label{eq:entrywise-3}
    \end{align}
    Define block matrices 
    \begin{align*}
        \bS \eqdef 
        \sqbracket*{\begin{matrix}
            \mS_2  &  &  \\ 
            & \ddots &  \\ 
            &  & \mS_{\kmax} 
        \end{matrix}} \quad,\quad 
        \bT \eqdef 
        \sqbracket*{\begin{matrix}
            \mT_2 & \cdots & \mT_{\kmax} \\ 
            \vdots & \ddots & \vdots \\ 
            \mT_2 & \cdots & \mT_{\kmax}
        \end{matrix}} \quad \text{, and}\quad
        \bW \eqdef 
        \sqbracket*{\begin{matrix}
            \mW_2 & & \\ 
            & \ddots & \\ 
            & & \mW_{\kmax}
        \end{matrix}} \;.
    \end{align*} 
    Direct multiplication and use of \cref{eq:entrywise-1,eq:entrywise-2,eq:entrywise-3} gives the relations
    \begin{align}
        \bT\bS 
        & = \bA \label{eq:rel-1}\\ 
        \bT\bW\bS 
        & = ((\mK - \mI_{\kappa})\otimes \mI_n )\bD + ((\mK - 2\mI_{\kappa})\otimes \mI_n)\bA \label{eq:rel-2}\\ 
        \bS\bT - \bW 
        & = \bB \label{eq:rel-3} \;.
    \end{align}

    We are now prepared for the main computation. 
    The \emph{push-through identity} states that
    \begin{align}
        \det(\mX + \mY \mZ) = \det(\mX) \det(\mI + \mZ\mX^{-1}\mY)\;, \label{eq:push-through}
    \end{align} 
    provided that $\mX$ is invertible and all matrix products are well-defined. 
    \citet{kemptonNonBacktrackingRandomWalks2016} provides an elementary proof. 
    Using \cref{eq:rel-3} to write $\mI - \mu\bB = \mI - \mu\bS\bT + \mu\bW$ and applying \cref{eq:push-through} to the righthand side gives 
    \begin{align*}
        \det(\mI - \mu \bB) = \det(\mI + \mu \bW) \det(\mI - \mu \bT(\mI + \mu\bW)^{-1}\bS)\;. 
    \end{align*}
    Focusing on the second factor, we compute
    \begin{align}\label{eq:block-inverse}
        (\mI + \mu \bW)^{-1} = 
        \sqbracket*{\begin{matrix}
            \mP_2(\mu) & & \\ 
            & \ddots & \\ 
            & & \mP_{\kmax}(\mu)
        \end{matrix}}\;,
    \end{align}
    where we have defined 
    \begin{align*}
        \mP_k(\mu) \eqdef p_k(\mu)\mI + q_k(\mu) \mW_k
    \end{align*}
    with coefficients 
    \begin{align*}
        p_k(\mu) \eqdef \frac{1 + \mu(k-2)}{h_k(\mu)} \;,&\quad q_k(\mu) \eqdef \frac{-\mu}{h_k(\mu)}\;, & \text{and} \; \quad h_k(\mu) \eqdef (1-\mu)(1 + \mu(k-1))\;.
    \end{align*}
    The derivation of this inverse uses the fact that $\mI + \mu \bW$ is a block-diagonal matrix with one block for each edge. 
    Each block for an edge of size $k$ has the form $\mI + (\mu - 1)\mE$, where $\mE$ is a $k\times k$ matrix of ones. 
    \revision{
        The expressions $p_k(\mu)$, $q_k(\mu)$, and $h_k(\mu)$ can be derived by assuming $\paren*{\mI + (\mu -1)\mE}\paren*{p_k(\mu) \mI + q'_k(\mu) \mE} = \mI$ and solving.  
        Then, $q_k(\mu) = q'_{k}(\mu) - p_k(\mu)$. 
    }

    Using \cref{eq:rel-1,eq:rel-2,eq:block-inverse}, we now compute 
    \begin{align*}
        \bT(\mI + \mu \bW)^{-1}\bS &= 
        \sqbracket*{\begin{matrix}
            p_2(\mu) \mA_2 & \cdots & p_{\kmax}(\mu) \mA_{\kmax} \\ 
            \vdots & \ddots & \vdots \\ 
            p_2(\mu) \mA_2 & \cdots & p_{\kmax}(\mu)\mA_{\kmax} 
        \end{matrix}} \\ 
        &\quad + 
        \sqbracket*{\begin{matrix}
            q_2(\mu)\mD_2 & \cdots & q_{\kmax}(\mu)\paren*{(\kmax-1)\mD_{\kmax} + (\kmax-2)\mA_{\kmax} }\\ 
            \vdots & \ddots & \vdots\\ 
            q_2(\mu)\mD_2 & \cdots & q_{\kmax}(\mu)\paren*{(\kmax-1)\mD_{\kmax} + (\kmax-2)\mA_{\kmax} }
        \end{matrix}}\;.
    \end{align*}
Performing $n$ row multiplications  by $h_k(\mu)$ for each $k$ yields 
    \begin{align*}
        \det(\mI - \mu \bT(\mI + \mu \bW)^{-1}\bS) = \paren*{\prod_k h_k(\mu)^{-n}} \det \bM(\mu)\;,        
    \end{align*}
    where $\bM(\mu)$ is the matrix 
    \begin{align*}
        \bM(\mu) &\eqdef 
        \sqbracket*{\begin{matrix}
            h_2(\mu)\mI_n & & \\
            & \ddots & \\ 
            & & h_{\kmax}(\mu)\mI_n 
        \end{matrix}} - 
        \mu
        \sqbracket*{\begin{matrix}
            (1 - \mu) \mA_2 & \cdots & (1 + \mu(\kmax-2)) \mA_{\kmax}  \\ 
            \vdots & \ddots & \vdots \\ 
            (1 - \mu) \mA_2  & \cdots & (1 + \mu(\kmax-2))\mA_{\kmax} 
        \end{matrix}} \\ 
        & \quad 
        + \mu^2 
        \sqbracket*{\begin{matrix}
            \mD_2 & \cdots & (\kmax-1)\mD_{\kmax} + (\kmax-2)\mA_{\kmax} \\ 
            \vdots & \ddots & \vdots \\ 
            \mD_2 & \cdots & (\kmax-1)\mD_{\kmax} + (\kmax-2)\mA_{\kmax}
        \end{matrix}} 
        \\ 
        &= \mI_{\kappa} + \mu ((\mK - 2\mI_{\kappa})\otimes \mI_n - \bA) + \mu^2 (\bD - \mI_{\kappa n}) \rev{((\mK - \mI_{\kappa}) \otimes \mI_n)}\;.
    \end{align*}
    This gives the second factor in the statement of \Cref{thm:ib-nonuniform}, so our final step is to address the factor $\det(\mI + \mu \bW)$. 
    We have 
    \begin{align*}
        \det(\mI + \mu \bW) &= \prod_{k\in \Kset} (1 - \mu)^{m_k(k-1)}(1 + \mu(k-1))^{m_k} \;.
    \end{align*}
    We find in turn 
    \begin{align*}
        \frac{\det(\mI + \mu \bW)}{\prod_{k \in \Kset} h_k(\mu)^{n}} = \prod_{k\in \Kset} (1 - \mu)^{m_k(k-1) - n}(1 + \mu(k-1))^{m_k - n} = f_\HG(\mu) \;. \label{eq:prefactor}
    \end{align*}
    This completes the computation and the proof. 

\subsection{Proof of \Cref{cor:ihara-bass-hypergraph-eigenvalues-nonuniform}} \label{sec:ib-nonuniform-cor-proof}

    Recall that $\pointed{m}$ is the total number of pointed edges. 
    \revision{
        We make the substitution $\mu = \frac{1}{\beta}$ in \eqref{eq:ib-equation}. 
        From the result, we extract copies of the characteristic polynomials $p_\mB$ of $\mB$ and $p_{\mB'}$ of $\mB'$. 
    }
    We obtain 
    \begin{align}
        \beta^{-\pointed{m}} p_{\bB}(\beta) = \beta^{-2\kappa n} f_{\HG}(\beta^{-1}) p_{\bB'}(\beta)\;.
    \end{align}
    We can have $p_\bB(\beta) = 0$ only if either $f_{\HG}(\beta^{-1}) = 0$ or $p_{\bB'}(\beta) = 0$. 
    For each $k$, if $m_k > n$, then $f_{\HG}$ has $m_k - n$ roots of the form $\beta = 1-k$. 
    Similarly, if $\sum_{k\in \Kset}m_k(k-1) > \kappa n$, then $f_{\HG}$ has $\sum_{k\in \Kset}m_k(k-1) - \kappa n$ roots of the form $\beta = 1$. 
    These are the only roots of $f_{\HG}$, and any remaining roots of $p_{\bB}$ must therefore be roots of $p_{\bB'}$. 
    If on the other hand $p_{\bB'}(\beta') = 0$, then either $p_{\bB}(\beta) = 0$ or $\beta^{-1}$ is a pole of $f_\HG$. 
    By our factorization of $f_\HG$, this can occur only if $\beta = 1-k$ for some $k$  or $\beta = 1$. 
    These cases can occur only if $m_k < n$ and $\sum_{k\in \Kset}m_k(k-1) < \kappa n$, respectively.

\subsection{Proof of \Cref{cor:ib-eigenvector-correspondence}} \label{sec:proof-of-ib-eigenvector-correspondence}

\revision{
    Our proof closely follows that of \citet{stephanSparseRandomHypergraphs2022}.
    Define the matrices 
        \begin{align*}
            \bar{\bS}_k \eqdef 
            \sqbracket*{\begin{matrix}
                \vzero &   &  & &  \\ 
                & \ddots  &  &  &\\ 
                & & \mS_{k} &  & \\ 
                & &  &  \ddots & \\ 
                & &  &   & \vzero
            \end{matrix}} \quad,\quad 
            \bar{\bT}_k \eqdef 
            \sqbracket*{\begin{matrix}
                \vzero & \cdots & \mT_{k} & \cdots & \vzero \\ 
                \vdots & \cdots & \vdots & \cdots & \vdots \\ 
                \vzero & \cdots & \mT_{k} & \cdots  & \vzero 
            \end{matrix}} \quad \text{, and} \quad
            \bar{\bW}_k \eqdef 
            \sqbracket*{\begin{matrix}
                \vzero &   &  & &  \\ 
                & \ddots  &  &  &\\ 
                & & \mW_{k} &  & \\ 
                & &  &  \ddots & \\ 
                & &  &   & \vzero
            \end{matrix}} \;.
        \end{align*} 
    We then have
    \begin{align*}
        \mS = \sum_{k\in K} \bar{\mS}_k\quad, \quad \mT = \sum_{k\in K} \bar{\mT}_k\quad, \quad \text{and}\quad \mW = \sum_{k\in K} \bar{\mW}_k\;. 
    \end{align*}   
    Let $\vx_{h;k} \in \vectorspace{\nodes}$ be the vector containing only the entries of $\vx_{h}$ corresponding to edges of size $k$ for $h = 1,2$. 
    The calculation 
    \begin{align}
        \vx_{2; k,i} = \sum_{Q \in \partial_ki} u_{iQ} 
        = \sum_{jQ \in \pointed{\edges}} t_{k; i,jQ} u_{jQ} 
        = (\bar{\bT}_k \vu)_i\;,
         \label{eq:x1_identity}
    \end{align}
    shows that $\vx_{2;k} = \bar{\mT}_k \vu$.  
    A similar calculation shows that $\vx_{1;k} = \bar{\mT}_k\mW^{-1}\vu$. 
    We also make use of the following identities, which can be verified through calculations similar to those shown in \Cref{sec:ib-nonuniform-proof}. 
    \begin{align}
        \mW_k &= (k-1)\mW_k^{-1} + (k-2)\mI, \label{eq:W-identity}\\ 
        \mB_k &= \mS\bar{\mT}_k - \bar{\mW}_k, \label{eq:B_identity}\\ 
        \bS_k &= \bW_k \bT_k^\top, \text{ and} \label{eq:S_identity}\\ 
        \bD_k &= \bT_k \bT_k^\top\;. \label{eq:D_identity} 
    \end{align}
}

\revision{
    Let $\beta \vu = \mB \vu$. 
    Denote by $\bar{\vu}_k$ the vector with components $\bar{u}_{k;iQ} = \delta_{\abs{Q}, k}u_{iQ}$.
    Then, we have $\vu = \sum_{k \in K} \bar{\vu}_k$. 
    We also have $\mB_{k'} \bar{\vu}_k = \beta \delta_{k',k}\bar{\vu}_k$, which implies that $\mB \bar{\vu}_k = \sum_{k'\in K}\mB_k \bar{\vu}_k = \mB_k \bar{\vu}_k = \beta \bar{\vu}_k$.  
    We will premultiply both sides of the relation $\beta \bar{\vu}_k = \mB \bar{\vu}_k$ by the matrix $\bar{\mT}_k\mW^{-1}$. 
}

\revision{
    Define $\bar{\vx}_{h;k}$ to be the vector with components $[\bar{\vx}_{h;k}]_{k',i} = \delta_{k, k'}x_{h;k,i}$ for $h = 1,2$. 
    Then, the relation $\bar{\mT}_k \mW^{-1}\vu$ implies that $\bar{\mT}_k \mW^{-1} \bar{\vu}_k = \bar{\vx}_{1;k}$. 
    On the other hand, we compute 
    \begin{align}
           \bar{\mT}_k \mW^{-1}\mB\bar{\vu}_k 
        &= \bar{\mT}_k \mW^{-1}\mB_k\bar{\vu}_k  \\ 
        &= \bar{\mT}_k \mW^{-1}\sqbracket*{\mS\bar{\mT}_k - \bar{\mW}_k}\bar{\vu}_k  \\ 
        &= \sqbracket*{\bar{\mT}_k\mW^{-1} \mS  - \mI_{\kappa n}}\bar{\vx}_{2;k}\;.
    \end{align}
    We have used the fact that, since $\bar{\vu}_k$ is nonzero only in the entries in which $\bar{\mW}_k$ is, $\mW^{-1}\bar{\mW}_k \bar{\vu}_k = \bar{\vu}_k$. 
}
    
\revision{
    The identity \eqref{eq:S_identity} implies that $\mW^{-1} \mS = \diag (\mT_1,\ldots, \mT_{\bar{k}})$. 
    The identity \eqref{eq:D_identity} then gives that 
    $\bar{\mT}_k \mW^{-1}\mS = \bar{\mD}_k$, where 
    \begin{align}
        \bar{\mD}_k = \sqbracket*{\begin{matrix}
            \vzero & \cdots & \mD_{k} & \cdots & \vzero \\ 
            \vdots & \cdots & \vdots & \cdots & \vdots \\ 
            \vzero & \cdots & \mD_{k} & \cdots  & \vzero 
        \end{matrix}}\;.
    \end{align}
    We have $\sum_{k\in K} \bar{\mD}_k = \mD$. 
    We have thus far shown that 
    \begin{align}
        \beta \bar{\vx}_{1;k} = \bar{\mT}_k\mW^{-1}\mB_k\bar{\vu}_k = \sqbracket*{\bar{\mD}_k - \mI_{\kappa n}} \bar{\vx}_{2;k}\;. \label{eq:rel-1-partial}
    \end{align}
    Summing over $k$ gives 
    \begin{align}
        \beta \vx_{1} = \sqbracket*{\mD - \mI_{\kappa n}}\vx_{2}\;, \label{eq:rel-1}
    \end{align}
    which establishes our first needed relation. 
}

\revision{
    Let us now instead premultiply the relation $\beta \bar{\vu}_k = \mB \bar{\vu}_k$ by $\bar{\mT}_k$. 
    The lefthand side becomes $\beta \bar{\vx}_{2;k}$. 
    The righthand side becomes 
    \begin{align*}
           \bar{\mT}_k \mB \bar{\vu}_k 
        &= \bar{\mT}_k \mB_k \bar{\vu}_k \\ 
        &= \bar{\mT}_k \mS \bar{\mT}_k\bar{\vu}_k - \bar{\mT}_k\bar{\mW}_k\bar{\vu}_k \\ 
        &= \bar{\mT}_k \mS \bar{\vx}_{2;k} - \bar{\mT}_k\bar{\mW}_k\bar{\vu}_k\;.
    \end{align*}
    We have $\bar{\mT}_k \mS = \bar{\mA}_k$, where $\bar{\mA}_k$ is the matrix 
    \begin{align}
        \bar{\mA}_k = \sqbracket*{\begin{matrix}
            \vzero & \cdots & \mA_{k} & \cdots & \vzero \\ 
            \vdots & \cdots & \vdots & \cdots & \vdots \\ 
            \vzero & \cdots & \mA_{k} & \cdots  & \vzero 
        \end{matrix}}\;.
    \end{align}
    Using this and \eqref{eq:W-identity}, we obtain 
    \begin{align}
        \bar{\mT}_k \mB \bar{\vu}_k &= \bar{\mA}_k \bar{\vx}_{2;k} - \bar{\mT}_k \sqbracket*{(k-1)\mW^{-1} + (k-2)\mI}\bar{\vu}_k \nonumber\\ 
        &= \bar{\mA}_k \bar{\vx}_{2;k} - (k-1)\bar{\vx}_{1;k} - (k-2)\bar{\vx}_{2;k} \nonumber\\ 
        &= \sqbracket*{\bar{\mA}_k - (k-2)\mI_{\kappa n}} \bar{\vx}_{2;k} - (k-1)\mI_{\kappa n}\bar{\vx}_{1;k} \;. \label{eq:rel-2-partial}
    \end{align} 
    Summing over edge sizes $k$ gives 
    \begin{align}
        \beta \vx_2 = \sqbracket*{\mA - (k-2)\mI_{\kappa n}}\vx_2 - (k-1)\mI_{\kappa n}\vx_1 \;. \label{eq:rel-2} 
    \end{align}
    Writing \eqref{eq:rel-1} and \eqref{eq:rel-2} in combined matrix form yields $\beta \paren*{\begin{matrix}\vx_1 \\ \vx_2 \end{matrix}} = \bB' \paren*{\begin{matrix}\vx_1 \\ \vx_2 \end{matrix}}$, as was to be shown.  
}

\section{Precise Statement and Proof of \Cref{thm:BP}} \label{sec:BP-proof}

Let $\alphabet$ be the alphabet of cluster labels, with $\abs{\alphabet} = \ell$. 
Let $\cM$ be the space of possible messages $\vmess$; we can identify $\cM$ with a product of probability $(\ell-1)$-simplices containing one factor for each node-tuple pair.
Let $\bar{\vmess}$ be the vector of messages with entries $\bar{\mu}_{iR}^{(s)} = q^{(s)}$ for all nodes $i$, subsets $R \in \tuples(i)$, and labels $s \in \alphabet$. 

We will consider perturbations to the belief-propagation dynamics. 
The normalization condition $\sum_{s \in \alphabet} \mess{i}{R}{s} = 1$ on elements of $\cM$ requires that perturbations $\vpmess$ to a message vector $\vmess$ must satisfy $\sum_{s \in \alphabet} \pmess{i}{R}{s} = 0$. 
Letting $\mPi$ denote the projection operator onto the subspace defined by this relation, we have 
\begin{align}
    [\mPi\vpmess]_{iR}^{(s)} = \pmess{i}{R}{s} - \frac{1}{\ell}\sum_{t\in\alphabet} \pmess{i}{R}{t}\;. 
\end{align}
For a given hypergraph realization, we can separate $\vmu$ into components $(\vmu_0, \vmu_1)$, where entries of $\vmu_0$ correspond to unrealized edges and entries of $\vmu_1$ correspond to realized edges. 
We can similarly separate the components of the function $\mF$. 
This allows us to write the BP update dynamics in the form 
\begin{align*}
    \vmu_0 &\gets \mF_0(\vmu_0, \vmu_1) \\     
    \vmu_1 &\gets \mF_1(\vmu_0, \vmu_1)\;. 
\end{align*}
We are now prepared to state a precise analog to to the heuristic \Cref{thm:BP}.

\begin{theorem}\label{thm:edge-reduced-dynamics}
    Let $\HG$ be sampled from the sparse Bernoulli blockmodel. 
    Then, as $n$ grows large, $\mF(\bar{\vmess}) \doteq \bar{\vmess}$. 
    Furthermore, 
    \begin{align}
        \mJ_{10} \eqdef \mPi \frac{\partial \mF_1(\bar{\vmu})}{\partial \vmess_{0}} \doteq \vzero \quad \text{and} \quad 
        \mJ_{11} \eqdef \mPi \frac{\partial \mF_1(\bar{\vmu})}{\partial \vmess_{1}} \doteq \sum_{k \in \Kset} (\mG_k \otimes \mB_k). \label{eq:BP-jacobian-expressions}
    \end{align}
\end{theorem}

\begin{proof}

    Throughout this proof, sums indexed by $\labelvec_R$ are assumed to run through $\alphabet^{\abs{R}}$ \rev{subject to stated constraints}. 
    
    We'll first compute several approximations describing how messages propagate along unrealized edges, i.e. subsets $R$ such that $a_R = 0$. 
    Since $\eta(a_R = 0|\labelvec_R) = 1-\omega(\labelvec_R)n^{1 - \abs{R}}$, \cref{eq:BP-2} becomes
    \begin{align}
        \varmess{i}{R}{s} &\gets \frac{1}{Z_{Ri}}\sum_{\labelvec_R:z_i = s}(1-\omega(\labelvec_R)n^{1 - \abs{R}}) \prod_{j \in R\setminus i}\mess{j}{R}{z_j} \nonumber\\ 
        &= \paren*{1 - O(n^{1-\abs{R}})}\frac{1}{Z_{Ri}}\sum_{\labelvec_R:z_i = s} \prod_{j \in R\setminus i}\mess{j}{R}{z_j} \nonumber\\ 
        &= \paren*{1 - O(n^{1-\abs{R}})}\frac{1}{Z_{Ri}}\;. \label{eq:almost-constant-1}
    \end{align}
    The last line follows from the normalization of the messages $\mess{j}{R}{s}$, since the sum ranges over all possible labelings of the nodes in $R\setminus i$. 
    Since the smallest possible edge size is $k = 2$, we have shown that $\varmess{i}{R}{s} \doteq Z_{Ri}^{-1}$. 
    In particular, $\varmess{i}{R}{s}$ is approximately constant with respect to $s$.

    We can also approximate $\mess{i}{R}{s}$ in the case $a_R = 0$, using \cref{eq:almost-constant-1,eq:BP-1,eq:marginal-message} to obtain 
    \begin{align}
        \mess{i}{R}{s} &\gets \frac{1}{Z_{iR}}\q{s} \prod_{Q \in \tuples(i)\setminus R} \varmess{i}{Q}{s} \nonumber\\ 
        &= \paren*{1 + O(n^{1-\abs{R}})} \frac{Z_{Ri}}{Z_{iR}}\q{s} \prod_{Q \in \tuples(i)} \varmess{i}{Q}{s} \nonumber\\ 
        & \doteq\mess{i}{}{s}\;. \label{eq:almost-constant-2}
    \end{align}
    Here, we are able to identify the normalizing constant $Z_i = \frac{Z_{iR}}{Z_{Ri}}$ independent of $R$ because it normalizes an expression independent of $R$. 

    Let's now consider how messages are passed along edges $R$ such that $a_R = 1$. 
    This corresponds to the consideration of $\mF_1$. 
    Substituting \cref{eq:BP-2} into \cref{eq:BP-1} and absorbing normalizing constants allows us to eliminate the messages $\varmess{i}{R}{s}$ entirely, obtaining an explicit form for $\mF_1$:
    \begin{align}
        \messcoords{\mF_1(\vmess_0, \vmess_1)}{i}{R}{s} = \frac{1}{Z_{iR}}\q{s} \prod_{Q \in \tuples(i)\setminus R} \sum_{\labelvec_Q: z_i = s} \eta(a_Q|\labelvec_Q) \prod_{j \in Q\setminus i} \mess{j}{Q}{z_j}\;. \label{eq:BP-3}
    \end{align}
    The updates in $\vmess_1$ are now $\vmess_1 \gets \mF_1(\vmess_0, \vmess_1)$.

    Let us write $\mF_1$ in the form 
    \begin{align*}
        \messcoords{\mF_1(\vmess_0, \vmess_1)}{i}{R}{s} = \frac{1}{Z_{iR}}\q{s} \messcoords{M}{i}{R}{s} \messcoords{N}{i}{R}{s}\;,
    \end{align*}
    where $\messcoords{M}{i}{R}{s}$ contains factors corresponding to sets $Q$ such that $a_Q = 1$, while $\messcoords{N}{i}{R}{s}$ contains factors for sets $Q$ such that $a_Q = 0$. 
    We can expand $\messcoords{N}{i}{R}{s}$: 
    \begin{align*}
        \messcoords{N}{i}{R}{s} &= \prod_{\substack{Q \in \tuples(i)\setminus R \\ a_Q = 0}} \sum_{\labelvec_Q: z_i = s} \eta(a_Q = 0 | \labelvec_Q) \prod_{j \in Q\setminus i} \mess{j}{Q}{z_j} \\ 
        &= \prod_{\substack{Q \in \tuples(i)\setminus R \\ a_Q = 0}} \sum_{\labelvec_Q: z_i = s} (1-\omega(\labelvec_Q)n^{1 - \abs{Q}}) \prod_{j \in Q\setminus i} \mess{j}{Q}{z_j} \\ 
        &= \prod_{\substack{Q \in \tuples(i)\setminus R \\ a_Q = 0}} \paren*{1 - n^{1 - \abs{Q}}\sum_{\labelvec_Q: z_i = s} \omega(\labelvec_Q) \prod_{j \in Q\setminus i} \mess{j}{Q}{z_j}}\;,
    \end{align*}
    where we have used the normalization of the messages $\mess{j}{Q}{z_j}$ in the last line. 

    We will now approximate $\messcoords{N}{i}{R}{s}$ by a ``field term'' $\messcoords{N}{}{}{s}$ which does not depend on $i$ or $R$. 
    First, since $a_Q = 0$ for each $Q$ appearing in the product defining $\messcoords{N}{i}{R}{s}$, we approximate $\mess{j}{Q}{z_j} = (1 + O(n^{1-\abs{Q}}))\mess{j}{}{z_j}$. 
    Next, 
    \begin{align*}
        \messcoords{N}{i}{R}{s} &= \prod_{\substack{Q \in \tuples(i)\setminus R \\ a_Q = 0}} \paren*{1 - n^{1 - \abs{Q}}\sum_{\labelvec_Q: z_i = s} \omega(\labelvec_Q) \prod_{j \in Q\setminus i} (1 + O(n^{1-\abs{Q}}))\mess{j}{}{z_j}}\\ 
        &\doteq (1 + O(n^{-1}))\prod_{\substack{Q \in \tuples(i)\setminus R \\ a_Q = 0}} \paren*{1 - n^{1 - \abs{Q}}\sum_{\labelvec_Q: z_i = s} \omega(\labelvec_Q) \prod_{j \in Q\setminus i}\mess{j}{}{z_j}} \\ 
        &= (1 + O(n^{-1}))\frac{\prod_{\substack{Q \in \tuples(i)}} \paren*{1 - n^{1 - \abs{Q}}\sum_{\labelvec_Q: z_i = s} \omega(\labelvec_Q) \prod_{j \in Q\setminus i}\mess{j}{}{z_j}}}{\prod_{\substack{Q \in \tuples(i) \\ a_Q = 1}} \paren*{1 - n^{1 - \abs{Q}}\sum_{\labelvec_Q: z_i = s} \omega(\labelvec_Q) \prod_{j \in Q\setminus i}\mess{j}{}{z_j}}}\;. 
    \end{align*}
    The number of factors in the denominator is equal to the degree of node $i$, which is binomial and therefore concentrates about its mean $\ck{k}{s}$.  
    We therefore have that, with high probability as $n$ grows large, the entire denominator is then also $1 - O(n^{-1})$. 
    With high probability, then, 
    \begin{align*}
        \messcoords{N}{i}{R}{s} &\doteq \prod_{\substack{Q \in \tuples(i)}} \paren*{1 - n^{1-\abs{Q}}\sum_{\labelvec_Q: z_i = s} \omega(\labelvec_Q) \prod_{j \in Q\setminus i} \mess{j}{}{z_j}} \\ 
        &\eqdef  N^{(s)}\;,
    \end{align*}
    where we have defined 
    \begin{align}
        N^{(s)} \eqdef \prod_k\prod_{\substack{Q \in \tuples_k}} \paren*{1 - n^{1-k}\sum_{\labelvec_Q: z_{q_1} = s} \omega(\labelvec_Q) \prod_{j \in Q\setminus q_1} \mess{j}{}{z_j}}
    \end{align}
    \revision{to} be a constant ``field term'' which does not depend on $i$ or $R$. 
    Thus, with high probability as $n$ grows large, the message passing update \cref{eq:BP-3} satisfies
    \begin{align}
        \mF_1(\vmess_0, \vmess_1)_{iR}^{(s)} &\doteq\frac{N^{(s)}\q{s}}{Z_{iR}}  \prod_{\substack{Q \in \tuples(i)\setminus R \\ a_Q = 1}} n^{1-\abs{Q}}\sum_{\labelvec_Q: z_i = s} \omega(\labelvec_Q)\prod_{j \in Q\setminus i} \mess{j}{Q}{z_j} \nonumber \\ 
        &= \frac{N^{(s)}\q{s}}{Z_{iR}}  \prod_{\substack{Q \in \tuples(i)\setminus R \\ a_Q = 1}}\sum_{\labelvec_Q: z_i = s} \omega(\labelvec_Q)\prod_{j \in Q\setminus i} \mess{j}{Q}{z_j}\;. \label{eq:approx-dynamics}
    \end{align}
    Here we have absorbed factors that do not depend on $s$ into $Z_{iR}$. 
    Importantly, $\mF_1$  depends on $\vmess_0$ only through the field term $N^{(s)}$. 

    We next consider the behavior of \cref{eq:approx-dynamics} near the point $\bar{\vmess}$ with entries $\bar{\mu}_{iR}^{(s)} = \q{s}$. 
    Let $\vpmess$ be a perturbation vector assumed small. 
    We'll first consider the field term. 
    Let $\Delta$ refer to terms of order $O(n^{-1} + \norm{\vpmess})$. 
    The perturbed field term $N^{(s)}$ now reads 
    \begin{align*}
        N^{(s)}(\bar{\vmess} + \vpmess)_{iR}^{(s)} &= 
         \prod_k\prod_{\substack{Q \in \tuples_k}} \paren*{1 - n^{1-k}\sum_{\labelvec_Q: z_{q_1} = s} \omega(\labelvec_Q) \prod_{j \in Q\setminus q_1} \paren*{\q{z_j} + \pmess{j}{Q}{z_j}}} \\ 
        &\doteq \prod_k\prod_{\substack{Q \in \tuples_k}} \paren*{1 - n^{1-k}\sum_{\labelvec_Q: z_{q_1} = s} \omega(\labelvec_Q) \paren*{\prod_{j \in Q \setminus q_1} \q{z_j} + \sum_{j \in Q\setminus q_1}\pmess{j}{Q}{z_j} \prod_{\ell \in Q\setminus q_1, j} \q{z_\ell}}} \\
        &= \prod_k \prod_{Q \in \tuples_k}\paren*{1 - (k-1)!n^{1-k}c_k - n^{1-k}\sum_{\labelvec_Q: z_{q_1} = s}\omega(\labelvec_Q) \sum_{j \in Q\setminus q_1}\pmess{j}{Q}{z_j} \prod_{\ell \in Q\setminus q_1, j} \q{z_\ell}  } \\ 
        &= (1 + \Delta)\prod_k \prod_{Q \in \tuples_k}\paren*{1 - (k-1)!n^{1-k}c_k} \\ 
        &\eqdef (1 + \Delta)N\;,
    \end{align*}
    where we have defined $N$ to absorb the products. 
    We have shown that, near $\bar{\vmess}$, the field term $N^{(s)}$ approximately does not depend on $s$.

    Paralleling the partition $\vmess = (\vmess_0, \vmess_1)$,  we can partition the entries of $\vpmess$ as $\vpmess = (\vpmess_0, \vpmess_1)$, again corresponding to unrealized and realized edges. 
    From \cref{eq:approx-dynamics},  
    \begin{align*}
        \messcoords{\mF_1(\bar{\vmess} + \vpmess)}{i}{R}{s} &= 
        \messcoords{\mF_1(\bar{\vmess}_0 + \vpmess_0,\bar{\vmess}_1 + \vpmess_1 )}{i}{R}{s} \\ 
        &= (1+\Delta)\frac{N\q{s}}{Z_{iR}}  \prod_{\substack{Q \in \tuples(i)\setminus R \\ a_Q = 1}}\sum_{\labelvec_Q: z_i = s} \omega(\labelvec_Q)\prod_{j \in Q\setminus i} \paren*{\q{z_j} + \pmess{j}{Q}{z_j}} \\ 
        &\doteq (1+\Delta)\frac{N\q{s}}{Z_{iR}} \prod_{\substack{Q \in \tuples(i)\setminus R \\ a_Q = 1}}\sum_{\labelvec_Q: z_i = s} \omega(\labelvec_Q) \sqbracket*{\prod_{j \in Q \setminus i}\q{z_j} + \sum_{j \in Q \setminus i}\pmess{j}{Q}{z_j}\prod_{h \in Q \setminus i,j} \q{z_h}} \\ 
        &= (1+\Delta)\frac{N\q{s}}{Z_{iR}} \prod_k \prod_{\substack{Q \in \tuples_k(i)\setminus R \\ a_Q = 1}}\paren*{(k-1)!c_k + \sum_{\labelvec_Q: z_i = s} \omega(\labelvec_Q)  \sum_{j \in Q \setminus i}\pmess{j}{Q}{z_j}\prod_{h \in Q \setminus i,j} \q{z_h} }\;.
    \end{align*}
    By conditioning on the label of $j$, we can simplify the second term in the factor: 
    \begin{align}
        \messcoords{\mF_1(\bar{\vmess} + \vpmess)}{i}{R}{s} &=(1+\Delta)\frac{N\q{s}}{Z_{iR}} \prod_k \prod_{\substack{Q \in \tuples_k(i)\setminus R \\ a_Q = 1}}\paren*{(k-1)!c_k + (k-2)! \sum_{t \in \alphabet} \ck{k}{s,t} \sum_{j \in Q\setminus i}\pmess{j}{Q}{t} } \nonumber \\ 
        &= (1+\Delta)\frac{N\q{s}}{Z_{iR}} \prod_k \prod_{\substack{Q \in \tuples_k(i)\setminus R \\ a_Q = 1}}\paren*{1 +  \sum_{t \in \alphabet} \frac{\ck{k}{s,t}}{(k-1)c_k} \sum_{j \in Q\setminus i}\pmess{j}{Q}{t} }\;. \label{eq:product-to-expand}
    \end{align}
    When $\vpmess = \vzero$, we have that $\mF(\bar{\vmess})_{iR}^{(s)} \doteq \frac{N}{Z_{iR}} \q{s}$. 
    Since the messages must normalize in $s$, we have that, up to errors 
    that can be absorbed into the first factor, $\frac{N}{Z_{iR}} = 1$. 
    This shows that  
    \begin{align}
        \mF_1(\bar{\vmess}_0, \bar{\vmess}_1) &\doteq \bar{\vmess}_1\;. \label{eq:fixed-point} 
    \end{align}
    A similar calculation shows that 
    \begin{align*}
        \mF_0(\bar{\vmess}_0, \bar{\vmess}_1) &\doteq \bar{\vmess}_0\;,
    \end{align*}
    which relations jointly give $\mF(\bar{\vmess}) \doteq \bar{\vmess}$.  
    This proves the first clause of the theorem. 

    Furthermore, since $\mF_1$ depends on $\bar{\vmess}_0$ only through the factor $N/Z_{iR} = 1 + \Delta = 1 + O(n^{-1} + \norm{\vpmess})$, we have that any derivative of $\mF_1$ in a direction corresponding to $\vmess_0$ is of order $O(n^{-1})$.
    Thus, the Jacobian $\frac{\partial \mF_1}{\partial \vmess_0}$, evaluated at $\bar{\vmess}$, has entries of order  $O(n^{-1})$. 
    The projected matrix $\mJ_{10} = \mPi \frac{\partial \mF_1}{\partial \vmess_0}$ is also of order $O(n^{-1})$, proving the first equation in \cref{eq:BP-jacobian-expressions}.

    It remains to compute $\mJ_{11}$ at $\bar{\vmess}$. 
    Expanding the product in \cref{eq:product-to-expand} to first order in $\vpmess$ and separating the arguments of $\mF_1$ gives 
    \begin{align*}
        \messcoords{(\bar{\vmess}_0, \bar{\vmess}_1 + \vpmess_1)}{i}{R}{s} = (1+\Delta) \q{s}\paren*{1 + \sum_{k \in \Kset} \sum_{t \in \alphabet} \frac{\ck{k}{s,t}}{(k-1)c_k} \sum_{\substack{Q \in \tuples_k(i)\setminus R \\ a_Q = 1}} \sum_{j \in Q\setminus i} \pmess{j}{Q}{t} }\;. 
    \end{align*}
    Using \cref{eq:fixed-point} gives 
    \begin{align*}
        \messcoords{\mF_1(\bar{\vmess}_0, \bar{\vmess}_1 + \vpmess_1) -  \mF_1(\bar{\vmess}_0, \bar{\vmess}_1)}{i}{R}{s} = (1+\Delta) \q{s}\sum_{k \in \Kset} \sum_{t \in \alphabet} \frac{\ck{k}{s,t}}{(k-1)c_k} \sum_{\substack{Q \in \tuples_k(i)\setminus R \\ a_Q = 1}} \sum_{j \in Q\setminus i} \pmess{j}{Q}{t}\;. 
    \end{align*}

    We now apply the projection $\mPi$ onto the subspace of admissible perturbations, yielding 
    \begin{align*}
        \messcoords{[\mPi\mF_1(\bar{\vmess}_0, \bar{\vmess}_1 + \vpmess_1) -  \mPi\mF_1(\bar{\vmess}_0, \bar{\vmess}_1) ]}{i}{R}{s} = (1+\Delta) \q{s}\sum_{k \in \Kset} \sum_{t \in \alphabet} \paren*{\frac{\ck{k}{s,t}}{(k-1)c_k} - 1} \sum_{\substack{Q \in \tuples_k(i)\setminus R \\ a_Q = 1}} \sum_{j \in Q\setminus i} \pmess{j}{Q}{t}\;.
    \end{align*}
    We can identify $R$ with an edge $e$, and the pair $(i,R)$ as a pointed edge $\pointed{e}$ with $p(\pointed{e}) = i$. 
    Doing the same for $Q$ and $j$, we can recognize the two rightmost sums as the action of $\mB_k$ on the perturbation vector $\vpmess$. 
    Using the definition of $\mG_k$, we can write this relation as 
    \begin{align*}
        \mPi\mF_1(\bar{\vmess}_0, \bar{\vmess}_1 + \vpmess_1) - \mPi\mF_1(\bar{\vmess}_0, \bar{\vmess}_1) = (1+\Delta) \sum_{k \in \Kset} (\mG_k \otimes \mB_k)\vpmess\;. 
    \end{align*}
    Ignoring the error term, this relation would define the Jacobian $\mJ_{11} \eqdef \mPi\frac{\partial \mF_1}{\partial \vmess_1}$ as equal to the righthand side. 
    Allowing $\norm{\vpmess}\rightarrow 0$,  we conclude that $\mJ_{11}$ satisfies 
    \begin{align*}
        \mJ_{11} \doteq (1 + O(n^{-1})) \sum_{k \in \Kset} (\mG_k  \otimes \mB_k)\;,
    \end{align*}
    which establishes the second clause of \cref{eq:BP-jacobian-expressions} and completes the proof. 
\end{proof}

\section{Proof of \Cref{thm:eigen-expectation}} \label{sec:eigen-expectation-proof}

    We will prove \cref{eq:eigen-expectation-2}; the proof of \cref{eq:eigen-expectation-1} is similar but somewhat shorter. 
    Let $\tilde{\vv}$ be a vector indexed by tuples $Q \in \binom{[n]}{k}$ and nodes $i \in Q$ with entries
    \begin{align*}        
        \tilde{v}_{iQ} &\eqdef \begin{cases}
            v_{iQ} &\quad Q \in \edges\;, \; \rev{i \in Q} \\ 
            0 &\quad \text{otherwise.}
        \end{cases}
    \end{align*}
    Let $\tilde{\bB}_k$ be the matrix with entries 
    \begin{align*}
        \tilde{b}_{k;iQ, jR} = 
        \begin{cases}
            b_{k;iQ, jR} &\quad Q, R \in \edges\;,\; \rev{i \in Q\;, \; j \in R}\\ 
            0 &\quad \text{otherwise.} 
        \end{cases}
    \end{align*}
    We are going to show that  $\E\sqbracket*{(\tilde{\mB}_k\tilde{\vv})_{iQ} - \beta_k\tilde{v}_{iQ}|Q \in \edges} \doteq 0$; since $\tilde{\mB}_k$ and $\tilde{\vv}$ agree with $\mB_k$ and $\vv$ conditioned on the event $Q \in \edges$, this will imply \cref{eq:eigen-expectation-2}. 

    We proceed via direct computation. 
    Expanding the expectation, we can write  
    \begin{align*}
        \E[(\tilde{\mB}_k\tilde{\vv})_{iQ}|Q\in \edges]
        &= \sum_{jR\in \pointed{\edges}_k}\E[\tilde{b}_{iQ, jR}\tilde{v}_{jR}|Q\in \edges] \\ 
        &= \sum_{jR\in \pointed{\edges}_k}\E[\tilde{b}_{iQ, jR}\tilde{v}_{jR}|Q, R\in \edges]\eta(R\in \edges) \\ 
        &= \sum_{jR\in \pointed{\edges}_k}\eta(R\in \edges) b_{iQ, jR}v_{jR} \\ 
        &= \sum_{jR\in \pointed{\edges}_k}\eta(R\in \edges) b_{iQ, jR}\sum_{\ell \in R\setminus j}\sigma_\ell \;. 
    \end{align*}
    The third line follows from the fact that, conditioned on the event $Q, R \in \edges$, $\tilde{b}_{iQ, jR} = b_{iQ, jR}$ and $\tilde{v}_{jR} = v_{jR}$. 
    Proceeding from the fourth line, we can evaluate $b_{iQ, jR}$ and rearrange the sums: 
    \begin{align*}
        \E[(\tilde{\mB}_k\tilde{\vv})_{iQ}|Q\in \edges]
        &= \sum_{j \in Q\setminus i}\sum_{\substack{R\in \tuples_k(j)\setminus Q}}\eta(R\in \edges)\sum_{\ell \in R\setminus j}\sigma_\ell \\ 
        &= \sum_{j \in Q\setminus i}\sum_{\substack{R'\in \binom{[n]\setminus j}{k-1}\\R'\neq Q\setminus j}} \eta(R'\cup j \in \edges)\sum_{\ell \in R'}\sigma_\ell \\ 
        &= \sum_{j \in Q\setminus i}\sum_{\ell \neq j}\sigma_\ell\sum_{\substack{R'' \in \binom{[n]\setminus \{j,\ell\}}{k-2} \\ R'' \neq Q\setminus \{\ell,j\}}} \eta(R'' \cup \{j , \ell\} \in \edges)\;. 
    \end{align*}
    The inner sum satisfies 
    \begin{align*}
        \sum_{\substack{R'' \in \binom{[n]\setminus \{j,\ell\}}{k-2} \\ R'' \neq Q\setminus \{\ell,j\} }} \eta(R'' \cup \{j , \ell\} \in \edges) \doteq \sum_{\substack{R'' \in \binom{[n]\setminus \{j,\ell\}}{k-2} }} \eta(R'' \cup \{j , \ell\} \in \edges)\;,
    \end{align*}
    The asymptotic equality holds because there are $\binom{n-2}{k-2}$ terms, of which the condition $R^{''} \neq Q\setminus\{\ell, j\}$ excludes only one.\footnote{In the edge case $k = 2$, the two sides are in fact exactly equal. } 
    We proceed to compute the sum on the righthand side. 
    \begin{align*}
        \sum_{\substack{R'' \in \binom{[n]\setminus \{j,\ell\}}{k-2}}} \eta(R'' \cup \{j , \ell\} \in \edges) 
        &= n^{1-k}\sum_{\substack{R'' \in \binom{[n]\setminus \{j,\ell\}}{k-2} }} \omega(\labelvec_{R''}, z_j, z_\ell) \\ 
        &= n^{1-k}\sum_{\labelvec \in \alphabet^{k-2}}\sum_{\substack{R'' \in \binom{[n]\setminus \{j,\ell\}}{k-2} \\ \labelvec_{R''} = \labelvec}} \omega(\labelvec, z_j, z_\ell)\;. 
    \end{align*} 
    We can make progress by counting the number of subsets $R''$ that realize each specified label vector $\labelvec$. 
    There are $\binom{n-2}{k-2}$ possible choices, and the proportion of these choices satisfying $\labelvec_R = \labelvec$ is asymptotically $\prod_{s \in \vz} \q{s}$. 
    This gives   
    \begin{align*}
        \sum_{\substack{R'' \in \binom{[n]\setminus j,\ell}{k-2}}} \eta(R'' \cup \{j , \ell\} \in \edges)
        &\doteq n^{1-k} \binom{n-2}{k-2}\sum_{\labelvec \in \alphabet^{k-2}}\omega(\labelvec, z_j, z_\ell) \prod_{s \in \vz} \q{s} \\ 
        &\doteq \frac{1}{n}\frac{1}{(k-2)!}\sum_{\labelvec \in \alphabet^{k-2}}\omega(\labelvec, z_j, z_\ell) \prod_{s \in \vz} \q{s} \\
        &= \frac{1}{n}\ck{k}{z_j,z_\ell}\;,
    \end{align*}
    where we have used \cref{eq:c2-def} in the final line. 
    We therefore have 
    \begin{align}
        \E[(\tilde{\mB}_k\tilde{\vv})_{iQ}|Q\in \edges] \doteq \frac{1}{n}\sum_{j \in Q\setminus i}\sum_{\ell \neq j}\sigma_\ell \ck{k}{z_j, z_\ell}\;. 
    \end{align}
    Let us split this sum according to whether $z_\ell = z_j$: 
    \begin{align*}
        \E[(\tilde{\mB}_k\tilde{\vv})_{iQ}|Q\in \edges] &\doteq \frac{1}{n}\sum_{j \in Q\setminus i}\paren*{\sum_{\substack{\ell \neq j \\ z_\ell = z_j}  }\sigma_\ell \ck{k}{z_j, z_\ell} + \sum_{\substack{\ell \neq j \\ z_\ell \neq z_j}  }\sigma_\ell \ck{k}{z_j, z_\ell}} \\ 
        &= \frac{1}{n}\sum_{j \in Q\setminus i}\sigma_j\paren*{\sum_{\substack{\ell \neq j \\ z_\ell = z_j}  } \ckin - \sum_{\substack{\ell \neq j \\ z_\ell \neq z_j}  } \ckout} \\ 
        &\doteq \frac{1}{2}\sum_{j \in Q\setminus i}\sigma_j\paren*{\ckin - \ckout} \\ 
        &= \frac{1}{2}\paren*{\ckin - \ckout} \tilde{v}_{iQ}\;. 
    \end{align*}
    For the third line, we have used the fact that there are approximately $\frac{n}{2}$ terms in each sum. 
    This completes the proof. 

\section{Proof of \Cref{thm:reduced-jacobian}} \label{sec:reduced-jacobian-proof}
\revision{
We now provide a more detailed statement and proof of \Cref{thm:reduced-jacobian}. 
We include a more explicit description of the matrix $\mL$, as well as the matrix $\mJ'$. 
\begin{theorem} \label{thm:jacobian-reduction}
    Suppose that $\vu \in \vectorspace{\alphabet\times \pointed{\edges}}$ and that $\xi \vu = \mJ\vu$ for some $\xi \neq 0$. 
    Let $\mL$ be the matrix
    \begin{align}
        \mL = 
        \sqbracket*{\begin{matrix}
            \mI_\ell \\ 
            \mI_\ell 
        \end{matrix}} 
        \otimes
        \sum_{k\in \Kset} 
        \sqbracket*{\begin{matrix}
            \bar{\mT}_k \mW^{-1} \\ 
            \bar{\mT}_k
        \end{matrix}}\;,
    \end{align}
    with $\bar{\mT}_k$ and $\mW$ as defined in \Cref{sec:proof-of-ib-eigenvector-correspondence}. 
    Let $\vx \eqdef (\vx_1, \vx_2)^T \eqdef \mL \vu$.  
    Then, $\xi \vx = \mJ' \vx$, where 
    \begin{align*}
        \mJ' \eqdef \paren*{\mI_{2}\otimes \mG\otimes \mI_n}
        \sqbracket*{
            \begin{matrix}
                \vzero & \mI_\ell \otimes \tilde{\mD} \\ 
                \vzero & \mI_\ell \otimes \tilde{\mA}
            \end{matrix}
        } 
        -  \paren*{\mI_{2}\otimes \mH\otimes \mI_n}
        \sqbracket*{
            \begin{matrix}
                \vzero & \mI_{\ell \kappa n} \\ 
                \mI_\ell \otimes \sqbracket*{\mK - \mI_\kappa}\otimes \mI_n & \mI_\ell \otimes (\mK - 2\mI_{\kappa})\otimes \mI_{n}
            \end{matrix}
        }
    \end{align*} 
    and
    \begin{align*}
        \mG \eqdef \sqbracket*{\begin{matrix}
            \mG_{2} & \cdots & \mG_{\bar{k}} \\ 
            \vdots & \ddots  & \vdots \\ 
            \mG_{2} & \cdots & \mG_{\bar{k}} 
        \end{matrix}} \; \text{,} \; 
        \mH \eqdef \sqbracket*{\begin{matrix}
            \mG_{2} & &  \\ 
             & \ddots  &\\ 
             & & \mG_{\bar{k}} 
        \end{matrix}}\; \text{, } \; 
        \tilde{\mD} \eqdef \sqbracket*{\begin{matrix}
            \mD_{2} & &  \\ 
             & \ddots  &\\ 
             & & \mD_{\bar{k}} 
        \end{matrix}} \; \text{, and} \; 
        \tilde{\mA} \eqdef \sqbracket*{\begin{matrix}
            \mA_{2} & &  \\ 
             & \ddots  &\\ 
             & & \mA_{\bar{k}} 
        \end{matrix}}\;.
    \end{align*}
    In particular, either $\vx = \vzero$ or $\vx$ is an eigenvector of $\mJ'$ with eigenvalue $\xi$. 
\end{theorem}
}

\begin{proof}
    \revision{
    Our proof broadly parallels the proof of \Cref{cor:ib-eigenvector-correspondence} given in \Cref{sec:proof-of-ib-eigenvector-correspondence}. 
    We multiply the relation $\xi \vu = \mJ \vu$ by each of the two blocks of $\mL$, obtaining the relationship $\xi \vx = \mJ' \vx$. 
    Starting with the first block, premultiply by the matrix $\mI_\ell \otimes \sum_{k \in \Kset} \bar{\mT}_k \mW^{-1}$.  
    On the lefthand side we obtain $\xi \vx_1$. 
    On the right we compute 
    \begin{align*}
        \xi \vx_1 &= \paren*{\mI_\ell \otimes \sum_{k \in \Kset} \bar{\mT}_k \mW^{-1}} \mJ \vu \\ 
        &= \paren*{\mI_\ell \otimes \sum_{k \in \Kset} \bar{\mT}_k \mW^{-1}} \paren*{\sum_{k' \in \Kset} \mG_{k'} \otimes \mB_{k'}} \vu \\ 
        &= \paren*{\sum_{k, k' \in \Kset} \paren*{\mG_{k'} \otimes \sqbracket{\bar{\mT}_k \mW^{-1} \mB_{k'}}}} \vu\;.
    \end{align*}
    Let $\bar{\vu}_k^{(s)} \in \vectorspace{\pointed{\edges}}$ be the vector with entries of $\vu$ in cluster $s$ with a pointed edge of size $k$, and zero for pointed edges of different size. 
    Let $\bar{\vu}_k \eqdef \sum_{s \in \alphabet} \ve^{(s)} \otimes \bar{\vu}_{k}^{(s)}$, where $\ve^{(s)}$ is the standard basis vector in the direction $s$. 
    We then have $\vu = \sum_{k \in \Kset} \bar{\vu}_k$. 
    We then write 
    \begin{align*}
        \xi \vx_1  &= \paren*{\sum_{k, k' \in \Kset} \paren*{\mG_{k'} \otimes \sqbracket{\bar{\mT}_k \mW^{-1} \mB_{k'}}}} \sum_{k'' \in \Kset, s \in \alphabet} \paren*{\ve^{(s)} \otimes \bar{\vu}_{k''}^{(s)}} \\ 
        &= \sum_{k, k', k'' \in \Kset, s \in \alphabet} \paren*{\mG_{k'} \ve^{(s)}}\otimes \paren*{\bar{\mT}_k \mW^{-1} \mB_{k'} \bar{\vu}^{(s)}_{k''}} \\ 
        &= \sum_{k \in \Kset, s \in \alphabet} \paren*{\mG_{k} \ve^{(s)}}\otimes \paren*{\bar{\mT}_k \mW^{-1} \mB_{k} \bar{\vu}^{(s)}_{k}}\;.
    \end{align*}
    In the second line we have used the mixed product property. 
    The third follows from direct multiplication, finding that the products involving mixed edge sizes zero out. 
    The same argument as in \Cref{sec:ib-nonuniform-cor-proof} shows that $\bar{\mT}_k \mW^{-1} \mB_{k} \bar{\vu}^{(s)}_{k} = \paren*{\bar{\mD}_{k} - \mI_{\kappa n}}\bar{\vx}_{2;k}^{(s)}$, where $\bar{\vx}_{2;k}^{(s)}\in \vectorspace{\nodes}$ has entries equal to $\vx_2$ in edge size $k$ and zero otherwise, with all entries corresponding to group $s$.  
    We thus have 
    \begin{align*}
        \xi \vx_1 &= \sum_{k \in \Kset, s \in \alphabet} \paren*{\paren*{\mG_{k} \ve^{(s)}}\otimes\paren*{\bar{\mD}_{k} - \mI_{\kappa n}}}\bar{\vx}_{2;k}^{(s)} \\ 
        &= \sum_{k \in \Kset, s \in \alphabet} \paren*{\mG_{k}\otimes \paren*{\bar{\mD}_{k} - \mI_{\kappa n}}}\paren*{\ve^{(s)} \otimes\bar{\vx}_{2;k}^{(s)}} \\ 
        &= \sum_{k \in \Kset} \paren*{\mG_{k}\otimes \paren*{\bar{\mD}_{k} - \mI_{\kappa n}}} \bar{\vx}_{2;k}\;.
    \end{align*}
    The second line is again the mixed product property, while for the third we have defined $\bar{\vx}_{2;k}$ to have entries that agree with $\vx_2$ on edges of size $k$ and which are zero otherwise. 
    Summing over $k$ yields, after some algebra, our first reduced relation: 
    \begin{align}
        \xi \vx_{1} = \sqbracket{\paren*{\mG \otimes \mI_n}\paren*{\mI_\ell \otimes \tilde{\mD}} - \mH \otimes \mI_n }\vx_{2} \;, \label{eq:reduced-jacobian-rel-1} 
    \end{align}
    where the matrices $\mG$, $\mH$, and $\tilde{\mD}$ are defined in the statement of \Cref{thm:jacobian-reduction}. 
    }

    \revision{
    We now premultiply both sides of the eigenvector relation $\xi \vu = \mJ \vu$ by the matrix $\mI_{\ell} \otimes \sum_{k \in \Kset }\bar{\mT}_k$. 
    The lefthand side becomes $\xi \vx_2$. 
    For the righthand side, we compute 
    \begin{align*}
        \xi \vx_2 &= \paren*{\mI_{\ell} \otimes \sum_{k \in \Kset }\bar{\mT}_k}\mJ\vu \\ 
        &= \paren*{\mI_\ell \otimes \sum_{k \in \Kset} \bar{\mT}_k } \paren*{\sum_{k' \in \Kset} \mG_{k'} \otimes \mB_{k'}} \vu \\ 
        &= \paren*{\sum_{k, k' \in \Kset} \paren*{\mG_{k'} \otimes \sqbracket{\bar{\mT}_k  \mB_{k'}}}} \vu \\ 
        &= \paren*{\sum_{k, k' \in \Kset} \paren*{\mG_{k'} \otimes \sqbracket{\bar{\mT}_k  \mB_{k'}}}} \sum_{k'' \in \Kset, s \in \alphabet} \paren*{\ve^{(s)} \otimes \bar{\vu}_{k''}^{(s)}} \\ 
        &= \sum_{k, k', k'' \in \Kset, s \in \alphabet} \paren*{\mG_{k'} \ve^{(s)}}\otimes \paren*{\bar{\mT}_k  \mB_{k'} \bar{\vu}^{(s)}_{k''}} \\ 
        &= \sum_{k \in \Kset, s \in \alphabet} \paren*{\mG_{k} \ve^{(s)}}\otimes \paren*{\bar{\mT}_k  \mB_{k} \bar{\vu}^{(s)}_{k}}\;.
    \end{align*}
    The steps of the calculation so far precisely parallel the calculation of \eqref{eq:reduced-jacobian-rel-1}. 
    Defining $\bar{\vx}^{(s)}_{1;k}$ similarly to $\bar{\vx}^{(s)}_{2;k}$ and retracing the argument from \Cref{sec:ib-nonuniform-cor-proof}, we further simplify 
    \begin{align*}
        \xi \vx_2 &= \sum_{k \in \Kset, s \in \alphabet} \paren*{\mG_{k} \ve^{(s)}}\otimes \paren*{\sqbracket*{\bar{\mA}_k - (k-2)\mI_{\kappa n}}\bar{\vx}^{(s)}_{2;k} - (k-1)\mI_{\kappa n}\bar{\vx}^{(s)}_{1;k}} \\ 
        &= \sum_{k \in \Kset, s \in \alphabet} \paren*{\mG_k\ve^{(s)}} \otimes \paren*{\sqbracket*{\bar{\mA}_k - (k-2)\mI_{\kappa n}}\bar{\vx}^{(s)}_{2;k} - (k-1)\mI_{\kappa n}\bar{\vx}^{(s)}_{1;k}} \\ 
        &= \sum_{k \in \Kset, s \in \alphabet}  \curlybrace*{\paren*{\mG_k\otimes \sqbracket*{\bar{\mA}_k - (k-2)\mI_{\kappa n}}} \paren*{\ve^{(s)}\otimes \bar{\vx}^{(s)}_{2;k}} - (k-1) \paren*{\mG_k \otimes \mI_{\kappa n}}\paren*{\ve^{(s)}\otimes \bar{\vx}^{(s)}_{1;k}}} \\ 
        &= \sum_{k \in \Kset} \curlybrace*{ \paren*{\mG_k\otimes \sqbracket*{\bar{\mA}_k - (k-2)\mI_{\kappa n}}} \bar{\vx}_{2;k} - (k-1) \paren*{\mG_k \otimes \mI_{\kappa n}}\bar{\vx}_{1;k}} \;.
    \end{align*}
    Performing the sum over $k$ yields, after some further algebra, 
    \begin{align}
        \xi \vx_{2} &=  \paren*{\mG \otimes \mI_n}\paren*{\mI_\ell \otimes \tilde{\mA}}\vx_2 - \paren*{\mH \otimes \mI_n}\paren*{\mI_\ell \otimes \sqbracket{\mK - 2\mI_\kappa} \otimes \mI_n}\vx_2 - \paren*{\mH \otimes \mI_n}\paren*{\mI_\ell \otimes\sqbracket{\mK - \mI_\kappa} \otimes \mI_n}\vx_1\;, \label{eq:reduced-jacobian-rel-2} 
    \end{align}
    where $\tilde{\mA}$ is as defined in the statement of \Cref{thm:jacobian-reduction}. 
    }
    
    \revision{
    Combining \cref{eq:reduced-jacobian-rel-1,eq:reduced-jacobian-rel-2} yields the eigenvector relation 
    \begin{align*}
        \xi \paren*{\begin{matrix} \vx_1 \\ \vx_2\end{matrix}} = 
        \curlybrace*{
            \paren*{\mI_{2}\otimes \mG\otimes \mI_n}
            \sqbracket*{
                \begin{matrix}
                    \vzero & \mI_\ell \otimes \tilde{\mD} \\ 
                    \vzero & \mI_\ell \otimes \tilde{\mA}
                \end{matrix}
            } 
            -  \paren*{\mI_{2}\otimes \mH\otimes \mI_n}
            \sqbracket*{
                \begin{matrix}
                    \vzero & \mI_{\ell \kappa n} \\ 
                    \mI_\ell \otimes \sqbracket*{\mK - \mI_\kappa}\otimes \mI_n & \mI_\ell \otimes (\mK - 2\mI_{\kappa})\otimes \mI_{n}
                \end{matrix}
            }
        }
            \paren*{\begin{matrix} \vx_1 \\ \vx_2\end{matrix}}\;. 
    \end{align*}
    Defining the matrix inside the braces as $\mJ'$ completes the proof. 
    }
\end{proof}

\section{Proof of \Cref{lm:vanilla-parameterized-eigs}} \label{sec:proof-of-vanilla-parameterized-eigs}

    We'll first calculate $\E[\mk{k}{s, t}]$,  where $\mk{k}{s,t}$ is the number of edges containing a node in cluster $s$ and a node in cluster $t$, counting multiplicities, discounting label order.  
    For example, an edge with group labels $(s, s, t, t, r)$ counts four times towards $\mk{5}{s,t}$ and twice each towards $\mk{5}{s, r}$ and $\mk{5}{t, r}$. 
    Another useful way to think of $\mk{k}{s,t}$ is as the number of pairwise edges joining nodes in cluster $s$ to nodes in cluster $t$ in the clique-projected graph, counted in both directions.  

    We'll now compute $\E[\mk{k}{s,s}]$. 
    There are a total of $\frac{nc_k}{k}$ $k$-edges in expectation, of which fraction $p_k$ are within-cluster and fraction $1-p_k$ are between-cluster. 
    The within-cluster edges contribute $2\binom{k}{2}$ within-cluster pairwise connections, with the factor of 2 reflecting the fact that each such connection must be counted once in each of two directions. 
    Since a given within-cluster edge is equally likely to lie within either of the two clusters, the total contribution to $\E[\mk{k}{s,s}]$ by within-cluster edges is $\frac{n c_kp_k}{k}\binom{k}{2}$. 
    There is also a contribution to $\E[\mk{k}{s,s}]$ from between-cluster edges. 
    There are in expectation $\frac{nc_k}{k}(1-p_k)$ such edges. 
    In a given such edge, if $H$ nodes are elements of cluster $s$, then there is a contribution of $2\binom{H}{2}$ to $\mk{k}{s, s}$. 
    Here, $0 \leq H \leq k-1$, since $H = k$ would yield a within-cluster edge. 
    We can therefore treat $H$ as a multinomial random variable with $k$ trials and uniform probability of each cluster label, conditioned on the event that the $k$ labels do not all agree. 
    Let $\tilde{H}$ be a binomial random variable with $k$ trials and success probability $\frac{1}{2}$. 
    Then, the expectation we want is 
    \begin{align*}
        \E\sqbracket*{\binom{H}{2}} = \frac{\E\sqbracket*{\binom{\tilde{H}}{2}} - 2^{-k}\binom{k}{2}}{1 -  2^{1-k}} = \frac{\E\sqbracket*{\binom{\tilde{H}}{2}} - 2^{-k}\binom{k}{2}}{1 - 2^{1-k}} = \frac{1}{4}\binom{k}{2}\frac{1 - 2^{2-k}}{1 - 2^{1-k}} \eqdef \frac{1}{2}\binom{k}{2}r_k\;,
    \end{align*}
    where $r_k = \frac{1 - 2^{2-k}}{2 - 2^{2-k}}$. 
    Combining this with our previous results, we have 
    \begin{align*}
        \E[\mk{k}{s,s}] = \frac{n(k-1)}{2}c_k\sqbracket*{p_k  + (1-p_k)r_k}\;.
    \end{align*}
    In turn, we have 
    \begin{align}
        \ckin \eqdef \ck{k}{s,s} = \frac{4\E[\mk{k}{s,s}]}{n} =  2(k-1)c_k\sqbracket*{p_k + (1-p_k)r_k}\;. \label{eq:ck-eq}
    \end{align}
    We can also now compute $\ckout$ via \cref{eq:degree-identity}: 
    \begin{align}
        \ckout \eqdef \ck{k}{s,t} = 2(k-1)c_k - \ckin\;. \label{eq:ck-neq}
    \end{align}
    \Cref{eq:ck-eq,eq:ck-neq} give us $\ckin$ and $\ckout$ as affine functions of $p_k$, which substantiates our claim that, under \Cref{conj:vanilla}, \cref{eq:detectability} defines a pair of hyperplanes in the coordinates $\{p_k\}$.

\section{Estimation of $\mG_k$} \label{sec:estimation-of-G}

    A natural candidate for a spectral algorithm would be to alternate between estimates of the community labels $\labelvec$ and the connectivity parameters contained in the matrix $\mG_k$. 
    Doing so requires the ability to estimate the entries of $\mG_k$ from the observed hypergraph and a label estimate $\hat{\labelvec}$. 
    We'll use $\hat{q}$ to refer to the estimate of the cluster population sizes using $\hat{\labelvec}$. 

    While there may be more subtle ways to do this, we proceed by identifying the expected average $k$-degree $c_k$ with the \emph{empirical} average $k$-degree $\frac{k}{n} m_k$, where $m_k = \abs{\HG_k}$ is the number of $k$-edges. 
    To estimate $\ck{k}{s,t}$, first let $\mk{k}{s,t}$ give the number of edges containing a node in cluster $s$ and a node in cluster $t$. 
    We'll compute $\E[\mk{k}{s,t}]$. 
    There are $\qhat{s}n$ nodes with label $s$ and $\qhat{t}$ nodes with label $t$. 
    Let us now select an additional $k-2$ nodes, with no distinction in their identities or labels. 
    There are approximately $n^{k-2}/(k-2)!$ ways to do so.
    The probability of a given node set $R \in \tuples_{k-2}$ yielding a specific label sequence $\labelvec$ is $\prod_{\ell \in R}\qhat{z_\ell}$, and in this case an edge is realized with probability $\lambda(\labelvec_R, s, t)$. 
    We therefore compute 
    \begin{align*}
        \E[\mk{k}{s,t}] &= \frac{\qhat{s}\qhat{t}n^{k-2}}{(k-2)!} \sum_{\labelvec \in \alphabet_{k-2}} \lambda(\labelvec, s, t) \prod_{\ell \in [k-2]} \qhat{z_\ell} \\ 
        &= \qhat{s}\qhat{t}n \sqbracket*{\frac{1}{(k-2)!}\sum_{\labelvec \in \alphabet_{k-2}}\omega(\labelvec, s, t) \prod_{\ell \in [k-2]} \qhat{z_\ell} } \\ 
        &= \qhat{s}\qhat{t}n \ck{k}{s,t}\;.
    \end{align*}
    So, to form an estimate $\hat{c}_k(s,t)$, we can first form an estimate of the population sizes $\hat{q}$ from an estimate of the cluster labels $\hat{z}$. 
    We then compute $\hat{m}_k^{s,t}$, the number of edges with a node in cluster $s$ and a node in cluster $t$, and then compute 
    \begin{align}
        \hat{c}_k^{(s,t)} = \frac{\hat{m}_k^{s,t}}{\qhat{s}\qhat{t} n}\;. \label{eq:estimate_C}
    \end{align}
    On a small technical note, $\hat{m}_k^{s,t}$ should be computed counting multiplicities; for example, a $5$-edge with labels $(a, a, b, b, c)$ would make four contributions to $\hat{m}_k(a, b)$ and two contributions to both $\hat{m}_k^{a,c}$ and $\hat{m}_k^{b, c}$.

\section{Additional Experiments} \label{sec:additional-experiments}

\revision{
    \Cref{fig:threshold-experiment} supports the use of sign-based thresholding in \Cref{alg:BP-spectral-clustering}, finding overall improved recovery as measured by the Adjusted Rand Index when using thresholding. 
    \Cref{fig:eigenvalue-locations} offers support of our conjectures for the locations of informative eigenvalues in binary detection experiments for both of the matrices $\mB$ and $\mJ$. 
    This figure also illustrates that the informative eigenvalue for the matrix $\mJ$ may be the largest real eigenvalue in magnitude, rather than the second-largest as is true for $\mB$. 
    \Cref{fig:large-hypergraph-1d} supports the accuracy of the thresholds predicted by \Cref{conj:jacobian-threshold} in a much larger synthetic hypergraph of $10,000$ nodes. 
    The parameter space explored corresponds to a vertical slice of \Cref{fig:heatmaps}(a-b) with $p_2 = 0.35$. 
} 

\begin{figure}
    \includegraphics[width=0.7\textwidth]{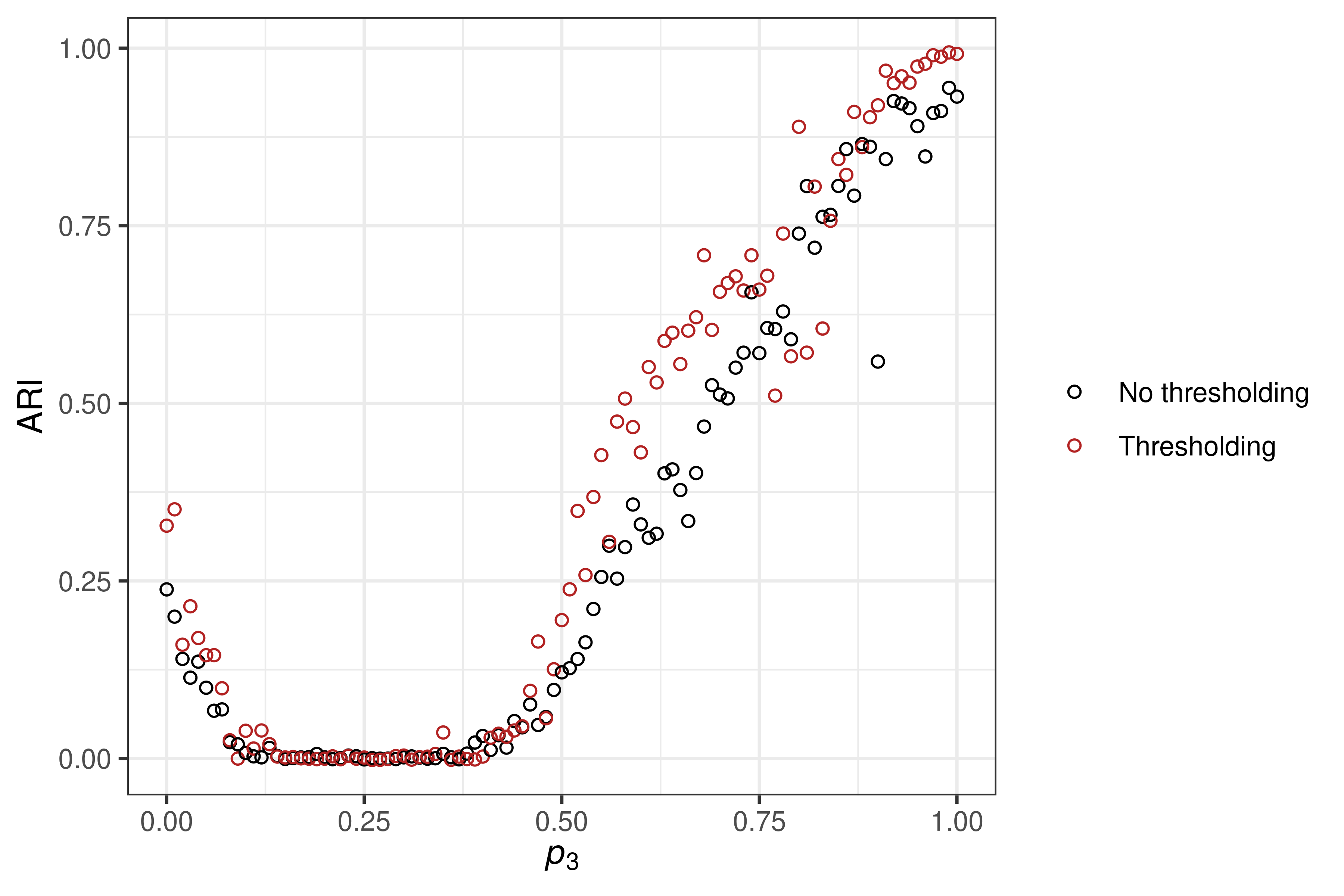}
    \caption{
        \centering
        \revision{
            Experimental test of the sign-based thresholding step in line 6 of \Cref{alg:BP-spectral-clustering}. 
            Binary detection experiment similar to those in \cref{fig:heatmaps}, with $p_2 = 0.35$ and varying $p_3$. 
            In this experiment, $c_2 = c_3 = 5$ and $c_4 = 0$, with $n = 400$ nodes. 
            The thresholding method stores only the sign of of the expression $\sum_{k \in K} x^{(s)}_{2;i,k}$, whereas the non-thresholding method stores the value of this expression. 
            The Adjusted Rand Index (ARI) averaged across 10 repetitions is shown. 
            } \label{fig:threshold-experiment}
    }
\end{figure}

\begin{figure}
    \centering 
    \includegraphics[width=0.7\textwidth]{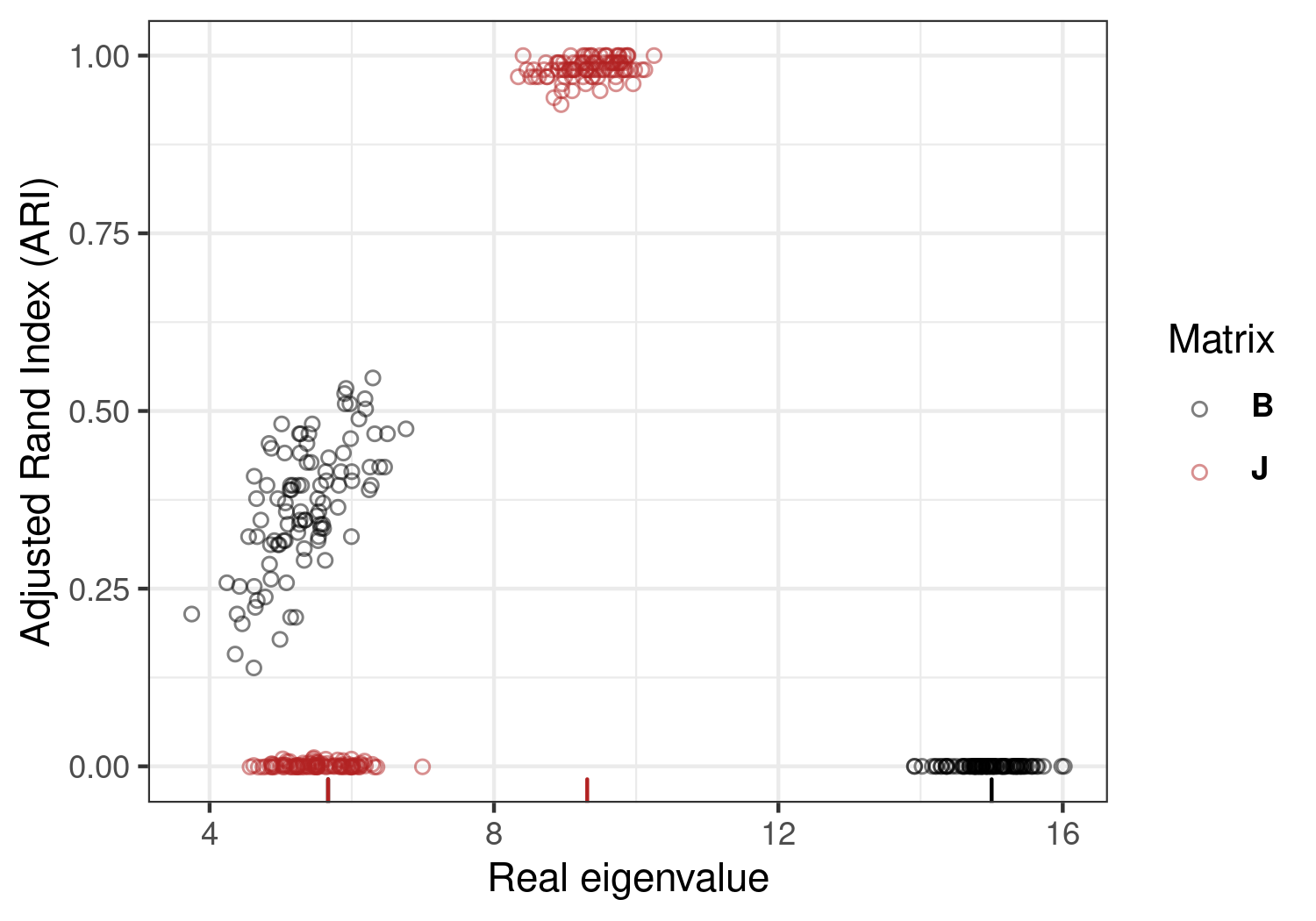}
    \caption{
        \revision{
            Locations of the two real eigenvalues with largest magnitude for the nonbacktracking matrix $\mB$ and the belief-propagation Jacobian matrix $\mJ$ (with true parameters) for 100 realizations of a random hypergraph as described in \cref{sec:parameterized-thresholds}. 
            The parameters are $p_2 = 0.2$, $p_3 = 0.9$, and $c_2 = c_3 = 5$, with $200$ nodes in each of the two groups.  
            The vertical axis gives the Adjusted Rand Index of the corresponding signs of the eigenvector against the planted community labels. 
            Ticks on the horizontal axis give the predictions according to \Cref{cor:eigen-expectation-B} and \Cref{cor:jacobian-eig}. 
            The predictions for the second eigenvalues of these two matrices overlap exactly. 
            The real eigenvalue of $\mB$ with largest magnitude is uncorrelated with community structure, while the eigenvalue of second-largest magnitude is correlated. 
            In contrast, the situation is reversed for $\mJ$. 
        }
    }\label{fig:eigenvalue-locations}
\end{figure}

\begin{figure}
    \centering
    \includegraphics[width=0.7\textwidth]{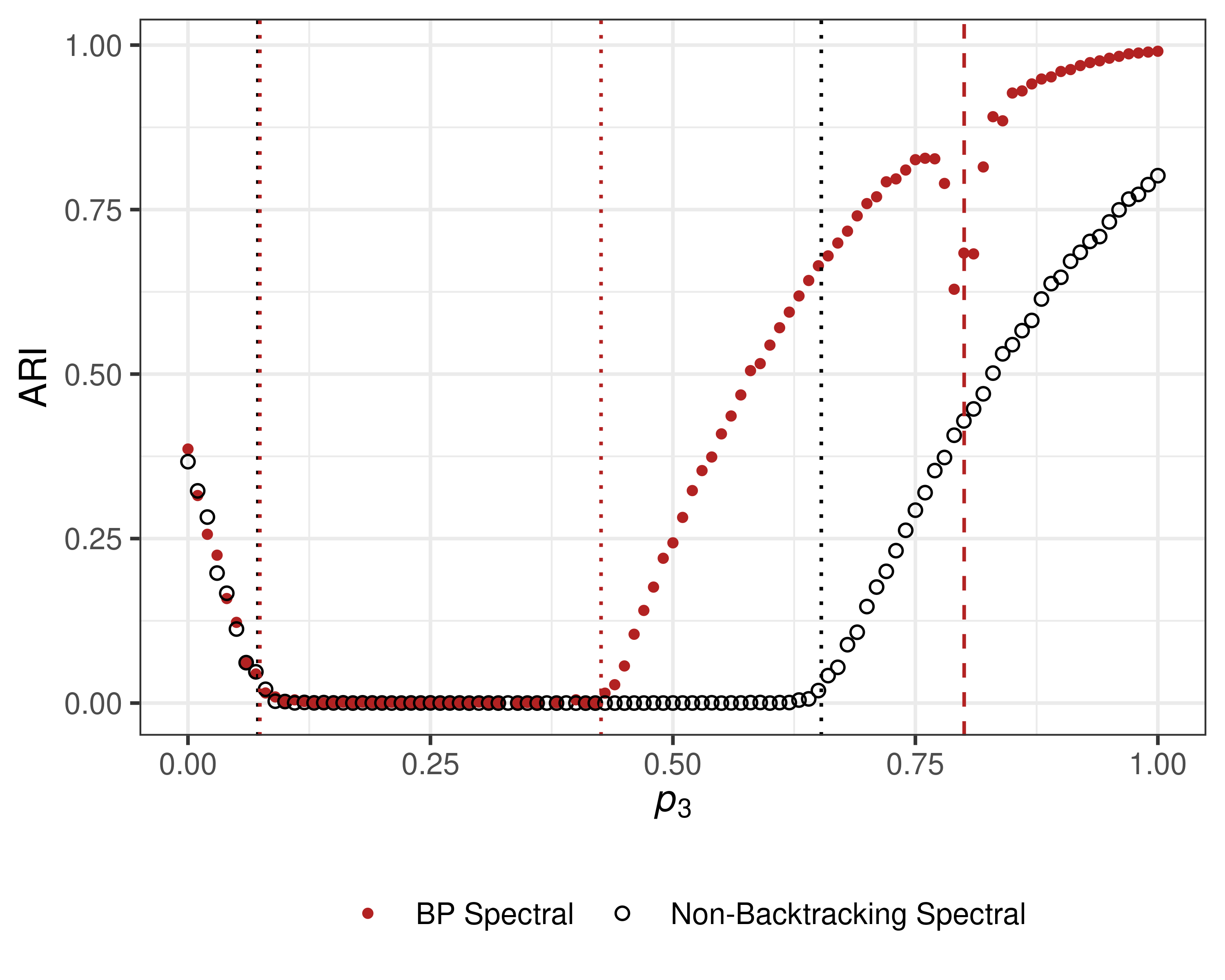}
    \caption{
        \revision{Binary detection experiment similar to those in \cref{fig:heatmaps}, with $p_2 = 0.35$ and varying $p_3$. 
        In this experiment, $c_2 = c_3 = 5$ and $c_4 = 0$, with $n = 10,000$ nodes. 
        We show performance as measured by ARI for both \ouralg{} and \oursimplealg{}. 
        Dotted lines give detectability thresholds for each algorithm (the two lower thresholds nearly overlap) from \Cref{conj:jacobian-threshold}. 
        The dashed line gives the estimated location collision of the two eigenvalues described by \Cref{cor:jacobian-eig}. 
        }}
        \label{fig:large-hypergraph-1d}
\end{figure}

\revision{
    We also include \Cref{fig:contact-high-school} below, which replicates \Cref{fig:contact-primary-school} for the \texttt{contact-high-school} data set. 
}

\begin{figure}
    \includegraphics[width=\textwidth]{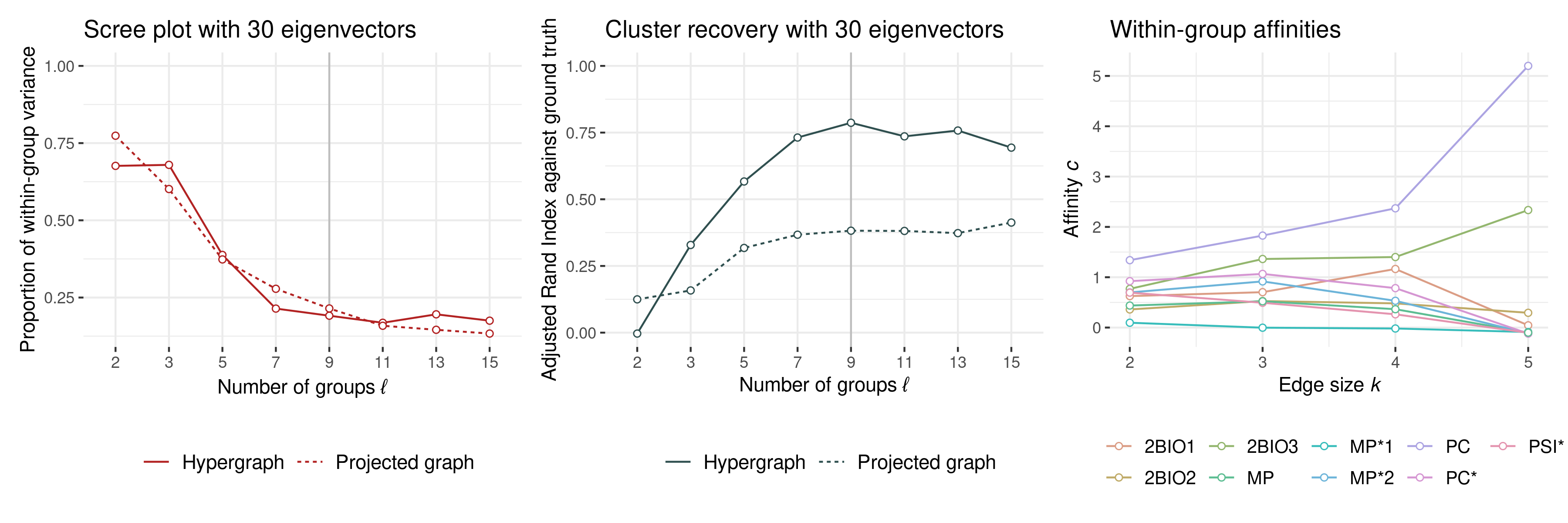}
    \caption{
        Experiment on the \datahighschool{} data set~\cite{mastrandreaContactPatternsHigh2015,bensonSimplicialClosureHigherorder2018}, analogous to the experiment shown in \Cref{fig:contact-primary-school}.
        There are $n = 327$ nodes and $m = 7,818$ hyperedges. 
        We ran \ouralg{} on the data for 10 rounds, using the 30 eigenvectors of the belief-propagation Jacobian with real eigenvalues of greatest magnitude and with a varying number of clusters to be estimated. 
        In each round, we update the estimate of the labels $\hat{\labelvec}$ by choosing the best of 20 runs of $k$-means according to the within-group sum-of-squares objective. 
        We repeat this experiment on the projected (clique-expansion) graph. 
        (Left): scree plot of the mean within-group sum-of-squares obtained by the $k$-means step as a function of the number of groups to be estimated. 
        The vertical grey line gives the true number of labels in the data. 
        (Center): Adjusted Rand Index of the clustering with lowest $k$-means objective against ground truth. 
        (Right): The diagonal entries of the matrix $\mC_k$ for varying edge size $k$. 
        }
        \label{fig:contact-high-school}
\end{figure}

\end{document}